\documentclass[]{imsart}

\RequirePackage{amsthm,amsmath,amsfonts,amssymb}
\RequirePackage[numbers]{natbib}
\usepackage[shortlabels]{enumitem}
\RequirePackage[colorlinks,citecolor=blue,urlcolor=blue]{hyperref}
\RequirePackage{graphicx}
\usepackage{arydshln}

\startlocaldefs
\theoremstyle{plain}

\newtheorem{theorem}{Theorem}[section]
\newtheorem{lemma}[theorem]{Lemma}
\theoremstyle{remark}

\newtheorem*{remark}{Remark} 

\endlocaldefs

\usepackage{algorithm,algpseudocode}
\usepackage{booktabs} 
\usepackage{subcaption}
\newcommand{\bA}{\mathbf{A}}

\newcommand{\bG}{\mathbf{G}}

\newcommand{\bR}{\mathbf{R}}

\newcommand{\bU}{\mathbf{U}}
\newcommand{\bV}{\mathbf{V}}

\newcommand{\ba}{\mathbf{a}}

\newcommand{\be}{\mathbf{e}}

\newcommand{\bss}{\mathbf{s}}

\newcommand{\sC}{\mathcal{C}}
\newcommand{\sD}{\mathcal{D}}

\newcommand{\sG}{\mathcal{G}}

\newcommand{\sK}{\mathcal{K}}
\newcommand{\sL}{\mathcal{L}}

\newcommand{\sS}{\mathcal{S}}
\newcommand{\sT}{\mathcal{T}}

\newcommand{\sV}{\mathcal{V}}
\newcommand{\sW}{\mathcal{W}}
\newcommand{\sX}{\mathcal{X}}

\newcommand{\sZ}{\mathcal{Z}}

\newcommand{\bbG}{\mathbb{G}}

\newcommand{\bbN}{\mathbb{N}}

\newcommand{\bbP}{\mathbb{P}}

\newcommand{\bbR}{\mathbb{R}}

\newcommand{\E}{\mathbb{E}}
\newcommand{\Var}{\text{Var}}
\newcommand{\cov}{\text{cov}}
\newcommand{\tp}{\text{T}}

\newcommand{\diag}{\text{diag}}

\newcommand{\bbeta}{\boldsymbol{\beta}}

\newcommand{\bOmega}{\boldsymbol{\Omega}}

\newcommand{\bPsi}{\boldsymbol{\Psi}}

\newcommand\numberthis{\addtocounter{equation}{1}\tag{\theequation}}

\usepackage{xr}
\makeatletter
\newcommand*{\addFileDependency}[1]{
  \typeout{(#1)}
  \@addtofilelist{#1}
  \IfFileExists{#1}{}{\typeout{No file #1.}}
}
\makeatother

\begin{document}
\begin{frontmatter}
\title{Identifying arbitrary transformation between the slopes in scalar-on-function regression}
\runtitle{Identifying arbitrary transformation}

\begin{aug}
\author[A]{\fnms{Pratim}~\snm{Guha Niyogi}\ead[label=e1]{pnyogi1@jhmi.edu}\orcid{0000-0003-2439-9047}}
\author[B]{\fnms{Subhra Sankar}~\snm{Dhar}\ead[label=e2]{subhra@iitk.ac.in}\orcid{0000-0003-1355-3635}}

\address[A]{Department of Biostatistics,
Johns Hopkins University, Baltimore, USA\printead[presep={,\ }]{e1}}

\address[B]{Department of Mathematics and Statistics,
IIT Kanpur, Kanpur, India\printead[presep={,\ }]{e2}}
\end{aug}

\begin{abstract}
In this article, we study whether the slope functions of two scalar-on-function regression models in two samples are associated with any arbitrary transformation along the vertical axis. The problem is formally stated as a statistical hypothesis test, and corresponding test statistic is formed based on the estimated second derivative of the unknown transformation. The asymptotic properties of the test statistic are investigated using some advanced techniques related to the empirical process. Moreover, to implement the test for small sample size data, a bootstrap algorithm is proposed, and it is shown that the bootstrap version of the test is as good as the original test for sufficiently large sample size. Furthermore, the utility of the proposed methodology is shown for simulated datasets, and DTI data is analyzed using the proposed methodology.
\end{abstract}



\end{frontmatter}

\section{Introduction}\label{Intro}

\subsection{Background, Problem of Interest and Literature Review}\label{Back}
In recent times, technology has made significant advancements, resulting in an increasing amount of functional data. Instead of scalar or multivariate vectors, each observation now represents a curve, as highlighted in monographs such as \citet{ramsay2005springer, hsing2015theoretical, kokoszka2017introduction}. These functional data are not limited to a specific field and can be found in diverse areas such as medicine, biology, economics, chemistry, engineering, and phonetics. Specifically speaking, functional data are usually obtained through technologies such as imaging, accelerometers, spectroscopy, spectrometry, and any other measurement on a dense grid collected over time, space, or any other ordered functional domain. For all these research problems of functional data analysis (FDA), the scalar-on-function linear model is a well-known modeling concept in the literature (see the review article of \citet{reiss2017methods}, entirely dedicated to scalar-on-function regression). Strictly speaking, the scalar-on-function regression model regresses the scalar-valued response on the function-valued covariates. In the course of work on the scalar-on-function linear model, the statistical inferences on the unknown slope function, i.e., the functional coefficients linked to functional covariates (see Model \eqref{eq:model} for details), are of particular interest in many practical problems. Some potential issues related to this context are described below.  
\par
There have been a few works related to statistical inferences, such as the estimation and testing of hypotheses on the slope functions available in the literature (see, e.g., \citet{MR4515711, MR4665579} and a few references therein). In particular, the estimation of $\beta_{s}(t)$ ($s = 1, 2$) in Model \eqref{eq:model} is studied by \citet{MR3318024, MR2766857, MR4655783, cardot2003spline} based on various methodologies, and a few other articles investigated testing of hypothesis problems on $\beta_{s}(t)$  (see, e.g., \citet{
yao2005functional, muller2005generalized, cai2006functional, 
hall2007methodology, hall2006properties1, li2007rates, zhang2007statistical, 
MR2766857, MR3014309, 
shang2015nonparametric} among many more). Among them, \citet{MR1965105, MR3099123, MR3223738} compared 
the slope functions in the regression model having the same covariates to a common response variable for different groups, and we come across such comparison often in checking whether the growth curves of the boys and the girls are same in checking whether climate pattern of a country is changing over a period of time (see, e.g., \citet{MR2413533}). In the course of this study, they formulated the test statistic based on the pooled sample using some appropriate distance between estimated slope functions. 
\par
However, the same techniques cannot be used once one would like to test whether the two slope functions, i.e., $\beta_{1}(t)$ and $\beta_{2}(t)$ are the same up to some arbitrary non-linear (i.e., excluding constant and linear functions) transformation $g$ as described in Statement \eqref{Hypo:state-alt}, and for example, such complex relationship between the slope functions often appear in diffusion tensor imaging (DTI) analysis (see, e.g., \citet{MR3470583}).  In DTI analysis, the paced auditory serial addition test (PASAT) scores (in this context, scalar valued response variable) of the patients with multiple sclerosis (MS) with respect to the continuous summary of the white matter tracts obtained from the corpus callosum (CCA) (in this context, the functional covariate) may vary over the time of the conducting the tests (i.e., different visits). See \citet{greven2011longitudinal} for more details, along with the description of the data set provided in Section \ref{sec:real-data}. From a medical science point of view, the researchers would like to know whether the rate of change of PASAT score with
respect to the continuous summary of the white matter from the CCA for the patients with MS has
any non-linear relationship or not over different visits. Motivated by such examples, in this article, we develop a unified statistical inference framework for testing whether the slope function of one group is the same as a non-linear function of that of the other group or not. To the best of our knowledge, this is the first initiative to demonstrate such a framework. 
\par
In addition to the DTI application described in the preceding paragraph, this inference problem is also of substantial importance in understanding growth patterns among children. For example, the Berkeley growth study (see \citet{rd1954physical}) is a widely used benchmark dataset that records the heights of 39 boys and 54 girls from ages 1 to 18 years. A primary objective of this dataset is to study growth trajectories for boys and girls, where longitudinal height measurements are collected repeatedly throughout childhood and adolescence, along with each child’s final height at age 18. In this setting, we treat the height at age 18 as the response variable (denoted as $Y$) and the growth curve over time as the functional covariate (denoted as $X(t)$). For both boys and girls, one may be interested in determining whether the rate of change of $Y$ with respect to $X(t)$ exhibits a nonlinear relationship. Equivalently, under a scalar-on-function regression framework, this amounts to assessing whether the functional slope parameters for boys and girls demonstrate any nonlinear structure.
\par
The formal description of the problem is available in Sections \ref{sec:con}, \ref{subsec:model}, and \ref{SP}. 

\subsection{Contributions}\label{sec:con}
As mentioned earlier, the first major contribution in this work is that we have checked whether two slope functions in scalar-on-function regression on two groups are the same or not up to some non-linear transformation. Suppose that for each group $s =1, 2$, $Y_{s}$ is a scalar-valued response variable, and $X_{s}$ is a predictor variable which is assumed to be a random function defined on a compact set $\sT = [0, 1]$. Consider now a scalar-on-function regression model for $s$-th group, 
\begin{equation}
\label{eq:model_1}
    Y_{s} = \alpha_{s} + \int_{0}^{1}\beta_{s}(t)[X_{s}(t) - \E\{X_{s}(t)\}]dt + \epsilon_{s},  
\end{equation}
where $\alpha_{s}$ are unknown constants and $\beta_{s} \in \sL_{2}([0,1])$ are unknown functional slopes for two groups.
Moreover, $\E\{ \epsilon_{s} \} = 0$ and $\E\{\epsilon_{s}^{2}\} = \sigma_{s}^{2}$, which are finite and unknown, and a few more technical assumptions will be stated at the appropriate places. We are now interested in testing the following hypothesis. 
\begin{equation}
\label{Hypo:state}
H_{0}:\beta_{1}(t) = g(\beta_{2}(t))~\mbox{for some linear transformation $g$},
\end{equation} 
against 
\begin{equation}
\label{Hypo:state-alt}
H_{1}:\beta_{1}(t) = g^{*}(\beta_{2}(t))~\mbox{for some non-linear transformation $g^{*}$},
\end{equation}
where $t\in [0, 1]$. Moreover, the null hypothesis may alternatively be expressed as 
$$H_{0}: \beta_{1} (t) = a \beta_{2}(t) + b~\mbox{for some constants}~ a, b \in \mathbb{R}.$$ For notational convenience, however, throughout the article we will adopt the formulation of the hypotheses given in Statements \eqref{Hypo:state} and \eqref{Hypo:state-alt} OR equivalent formulation described in Statements \eqref{eq:null} and \eqref{eq:alt}.  In order to test the above hypothesis \eqref{Hypo:state} against \eqref{Hypo:state-alt}, we propose a test statistic based on the $L_{\infty}$ norm of the estimator of the second derivative of $g$. The detailed explanation of considering the second derivative of $g$ is discussed in Section \ref{SP}. In notation, the test statistic $T_{n}$ is of the form $T_{n} = \sup\limits_{t\in \mathcal{T}}|\widehat{g}''(t)|$, where $\mathcal{T}$ is a certain random interval, and $\widehat{g}''(t)$ is a certain estimator of the second derivative of $g$. In this study, we derive the asymptotic distribution of $T_{n}$ after appropriate normalization. Note that, when $g(x) = x$, it is then equivalent to test whether two slope functions are point-wise equal or not, and this test was studied by \citet{horvath2009two}. More generally, even when $g(x) = ax + b$ for some $a\in\bbR$ and $b\in\bbR$, it is also equivalent to test whether two slope functions are point-wise equal or not after some linear transformation, and the methodology adopted by \citet{horvath2009two} can be applied for such cases as well. However, for some non-linear transformation, testing hypotheses $H_{0}$ against $H_{1}$ (see Statements \eqref{Hypo:state} and \eqref{Hypo:state-alt}) is an entirely different and much more complex problem than that of the work done by \citet{horvath2009two}, and this problem is studied in this work. 
\par
In this context, we stress that the aforementioned test is fundamentally distinct from curve registration problems (see, e.g., \citet{MR3323102, MR2369028, MR1041386, MR4550234}). In the curve registration problems, the curves are related to each other by some transformations along the horizontal axis, whereas in Statement \eqref{Hypo:state-alt}, $\beta_{1}(t)$ and $\beta_{2}(t)$ are associated with some non-linear transformation along the vertical axis. As indicated before, the relationship described in Statement \eqref{Hypo:state-alt} often appears in reality. For example, recently researchers in medical science are extensively studying the relationship between systolic blood pressure and white matter lesions in individuals with hypertension, and a few studies have found that the relationship between these two time-dependent variables is quadratic in nature (see, e.g., \citet{10.1001/jama.2019.10551}). Another example is that in the inter-phase of Finance and Environmental Science (see, e.g., \citet{XU2022105994}), social scientists are interested in knowing the relationship between the financial development and the carbon dioxide emissions over the last one hundred years or so in the G7 countries (i.e., Canada, France, Germany, Italy, Japan, the United Kingdom (UK), and the United State (US)). Such impressive real-life examples further indicate the importance of the hypothesis problem described in Statements \eqref{Hypo:state} and \eqref{Hypo:state-alt}.  
\par
The next major contribution of the work is related to the implementation of the test. It is indeed true that using the results described in Theorem \ref{thm:T}, one can implement the test when the sample size is sufficiently large. However, for a moderate or small sample size, implementing a test with the assertion in Theorem \ref{thm:T} may not have adequate performance. To overcome this problem, we propose a bootstrap procedure with a better rate of convergence (see Theorem \ref{thm:bootstrap}), which leads to better performance of the test for the data with a small sample size. 

\subsection{Challenges}\label{subsec:challenges}
The first challenge was related to the fact that the supremum involved in the test statistic is defined over a random interval (see Equation \eqref{stat}), and hence, it is not possible to apply the well-known continuous mapping theorem to tackle the issue related to the supremum operator. To overcome this problem, we first establish the asymptotic distribution of the modified test statistic, where the supremum is taken over a certain fixed interval. Afterward, we show that the difference between the modified test statistic and the original test statistic is negligible.
\par
The second major challenge is related to the estimation and other issues associated with statistical inference on the non-linear transformation $g$ and its second-order derivatives. In this study, $g$ is estimated based on $\widehat{\beta}_{1}(t)$ and $\widehat{\beta}_{2}(t)$ evaluated at discrete time points, where $\widehat{\beta}_{1}(t)$ and $\widehat{\beta}_{2}(t)$ are certain estimators of the slope parameters of two groups, denoted by $\beta_{1}(t)$ and $\beta_{2}(t)$, respectively. Moreover, the number of discrete time points (i.e., $m_{n}$; see the description in Section \ref{subsec:formulation}) may vary over the sample size as well. Hence, one needs to carefully modify the well-known techniques of non-parametric regression, such as the local polynomial regression technique, to study the various properties of the estimator of $g$ and its second-order derivatives.  
\par
The third major challenge is related to the choice of discrete points, where $\widehat{\beta}_{1}(t)$ and $\widehat{\beta}_{2}(t)$ are observed, and those observations are used in estimating $g$ and its second derivative. Here, the number of discrete points varies with sample size $n$, and hence, the optimum choice of the number of discrete points is not easily tractable. To overcome it, one can choose various functions of the sample size, such as logarithmic or exponential transformation, and afterward, one can check which type of transformation gives the best result in various numerical studies.
Moreover, the issue that we mentioned as the third challenge also often leads to a skewed signal-to-noise ratio as the second derivative of $g$ is estimated based on $\hat{\beta}_{1}(t)$ and $\hat{\beta}_{2}(t)$. Note that it follows from the Model \eqref{eq:single-index} that the estimators of $g$ and its derivatives depend on the variables $U_{j}$ and $V_{j}$ ($j = 1, \ldots, m_n$) having similar to the measurement errors, where $U_{j}$ and $V_{j}$ are constructed based on $\beta_{1}(t_{j})$ and $\beta_{2}(t_{j})$, respectively. This issue makes the estimator of the second derivative of $g$ unstable, and consequently, the test statistic $T_{n}$ is so. To have a stable estimator of the second derivative of $g$, one needs to carefully choose the bandwidth and the kernel function associated with the estimator of the second derivative of $g$ (see the conditions described in Section \ref{sec:theory}).

\subsection{Organization of the article}
\label{subsec:Organization}
The article is organized as follows. In Section \ref{sec:model}, we introduce our model and the details of the preliminaries of the methodology that is proposed in this article. In Section \ref{subsec:model}, we describe the statistical model and frame the problem of hypothesis testing in Section \ref{SP}. Later, in Section \ref{subsec:method}, we present the methodology to carry out the testing of hypothesis problems in detail. In Section \ref{sec:theory}, we thoroughly investigate various large sample statistical properties and related facts of the proposed test. We perform numerical analysis in Section \ref{sec:simulation} and the analysis of one data-set on DTI in Section \ref{sec:real-data}. Section \ref{sec:discussion} concludes with a brief discussion and the future direction of this research. Additional technical details and the figures related to finite sample performance and real data analysis are included in the Appendix.

\subsection{Notation}
\label{subsec:notation} 
We here summarize the notations used in this article. $\textbf{1}(x \in A)$ denotes the indicator function, which takes value 1 if $x \in A$ and 0 if $x \notin A$ for a set $A$. For any two sequences $\{a_{n}\}_{n\geq 1}$ and $\{b_{n}\}_{n\geq 1}$,  $a_{n}\lesssim b_{n}$ indicates $a_{n}\leq C b_{n}$ for all $n\in\mathbb{N}$ for some $C > 0$, and $C$ does not depend on $n$. For any two sequences $\{a_{n}\}_{n\geq 1}$ and $\{b_{n}\}_{n\geq 1}$, $a_{n} = O(b_{n})$ indicates that there exists some $M > 0$ such that $|{a_{n}}/{b_{n}}| < M$ for all $n\geq N_{0} (M)$. Further, $a_{n} = o(b_{n})$ indicates that $\displaystyle\lim_{n\rightarrow\infty}{a_{n}}/{b_{n}} = 0$. For any sequence of random variables $\{Z_{n}\}_{n\geq 1}$, $Z_{n} = O_{r}(a_{n})$ indicates that $\E|Z_{n}|^{r} = O(a_{n}^{r})$. For any vector $\ba \in \bbR^{p}$ ($p\geq 1$), $\ba^{\otimes^{2}} = \ba\ba^{\tp}$. $\be_{k} \in \bbR^{p}$ is denoted as a vector of length $p$, such that $\be_{k} = (e_{1}, \ldots, e_{p})^{\tp}$,  where $e_{j} = \textbf{1}(j = k)$. $||f(.)||_{\infty} = \displaystyle\sup_{x\in S_{x}}|f(x)|$, where $S_{x}$ is the support of $f$. Next, $\diag(a_{1}, \ldots, a_{m})$ denotes a $m\times m$ matrix, whose $i$-th diagonal element is $a_{i}$ ($i = 1, \ldots, m$), and all non-diagonal elements are zero. For any real number $a$, $\lfloor a \rfloor$ is the largest integer less than $a$. For any non-negative integers $a$ and $b$,  $\nu_{a, b} = \int v^{a}K^{b}(v)dv$, where $K: \bbR\rightarrow\bbR^{+}$ is such that $\int K(x) dx = 1$. For any two constants $c_{1}, c_{2}$, $\sS(c_{1}, c_{2}) = \{ f: |f^{(d)}(u_{1}) - f^{(d)}(u_{2})| \leq c_{2}|u_{1} - u_{2}|^{c_{1} - d}, \text{ for all } d = \lfloor c_{1} \rfloor\}$, where $f^{(d)}$ denotes the $d$-th derivative of $f$. This is a well-known H\"older class of functions (see, e.g., \citep{tsybakov1997nonparametric}). For any  arbitrary point $v_{0}$ in the support of a real valued smooth function $f$,  $f^{[j]}(v_{0})$ denotes the $j$-th derivative $f$ at $v_{0}$.

\section{Formulation of the problem}
\label{sec:model}

\subsection{Description of the model}
\label{subsec:model}
Suppose that $(Y_{s, i}, X_{s, i})_{i = 1}^{n_{s}}$ ($s =1, 2$) are two independent random samples identically distributed with $(Y_s, X_s)$, where $Y_s$ is a scalar-valued response, and $X_{s}$ is $\sL_{2}(\cal{T})$-valued random element. Here $${\sL}_{2}(\sT) = \{ f: \mathcal{T}\rightarrow\mathbb{R}~|~ f\text{ is measurable and } \int_{\sT} f^{2}(x)dx < \infty \}, {\cal{T}}\subset\bbR$$ is a compact set, and without loss of generality, we consider ${\cal{T}} = [0, 1]$ unless mentioned otherwise. Note that the continuity of the sample paths of $X_{s}(t)$ ($t\in [0, 1]$) is sufficient to be a ${\sL}_{2}([0, 1])$-valued random element in view of the fact that a continuous function on a compact set is uniformly bounded. Now, recall from Section \ref{sec:con} that for $s$-th group ($s = 1, 2$), 
\begin{equation}
\label{eq:model}
    Y_{s} = \alpha_{s} + \int_{0}^{1}\beta_{s}(t)[X_{s}(t) - \E\{X_{s}(t)\}]dt + \epsilon_{s}. 
\end{equation}
The model \eqref{eq:model} is a well-known scalar-on-function regression in the statistics literature (see, e.g., the Introduction in \citet{MR4515711}). 
In this model, for $s = 1, 2$, $\alpha_{s}$ are unknown constants and $\beta_{s} \in \sL_{2}([0,1])$ are unknown functional slopes for two groups. Moreover, we assume that $X_{s}$ and $\epsilon_{s}$ are independent for each group along with $\E\{ \epsilon_{s} \} = 0$ and $\E\{\epsilon_{s}^{2}\} = \sigma_{s}^{2} < \infty$, where $\sigma_{s}^{2}$ is unknown. 

\subsection{Statement of the problem}\label{SP}
Recall the statement of the hypothesis described in Statements \eqref{Hypo:state} and \eqref{Hypo:state-alt}.  
$$H_{0}:\beta_{1}(t) = g(\beta_{2}(t))~\mbox{for some linear transformation}~g, $$
against 
$$H_{1}:\beta_{1}(t) = g^{*}(\beta_{2}(t))~\mbox{for some non-linear transformation $g^{*}$},$$ 
where $t\in [0, 1]$. In general, we are interested in testing whether $\beta_{2}$ and $\beta_{1}$ are associated with some non-linear transformation or not along the vertical axis, i.e., the hypothesis described in Statement \eqref{Hypo:state-alt}. First, it is easy to see that the constant and linear transformation between $\beta_{1}(t)$ and $\beta_{2}(t)$ along the vertical axis, which is asserted in Statement \eqref{Hypo:state} is redundant. Note that for the constant association, it becomes equivalent to check $\beta_{1} (t)$ equals some constant for all $t\in\sT$, which does not involve any information about $\beta_{2}(t)$. Secondly, for linear transformation, it is equivalent to test equality between $\beta_{1}(t) = \beta_{2}(t)$ after appropriate standardization of the data, and this test is well studied in the literature (see, e.g., \citet{MR2413533}). Hence, to avoid the aforementioned cases, the alternative equivalent statement of the hypothesis is formulated in a strict sense, and a technical description of it is as follows. 
\par
Consider the following class of functions on a certain interval $[v_{1}, v_{2}]$: 
\begin{equation}
\label{eq:C}
   \sC_{v_{1}, v_{2}} = \left\{ l: l''(u) = 0~\mbox{for all}~u \in [v_{1}, v_{2}] \right\},  
\end{equation} where $l$ is twice differentiable, and $l''$ denotes the second order derivatives of $l$. Therefore, strictly speaking, equivalent to the hypothesis described at the beginning of this subsection, we are interested in testing 
\begin{equation}
\label{eq:null}
H_{0}:  g\in \sC_{v_{1}, v_{2}}
\end{equation}
against the alternative 
\begin{equation}
\label{eq:alt}
H_{1}: g\notin \sC_{v_{1}, v_{2}},
\end{equation}
where $v_{1} = \min\limits_{t\in [0,1]}\beta_{2}(t)$ and $v_{2} = \max\limits_{t\in [0,1]}\beta_{2}(t)$. 

The next section develops the methodology to carry out the test $H_{0}$ against $H_{1}$ described in Statements \eqref{eq:null} and \eqref{eq:alt}, respectively.

\section{Development of methodology}
\label{subsec:method}

The proposed testing procedure consists of two key steps: first, estimating the unknown regression coefficient functions $\beta_{s}(\cdot)$, and then computing the test statistic.

\subsection{Estimation of $\beta$}
\label{subsec:estimation_of_beta}
For $X_{s}(t)$, $t\in{\cal{T}}$, let $V_{s}(t_{1}, t_{2}) = \cov\{X_{s}(t_{1}), X(t_{2})\}$ for $t_{1}, t_{2} \in \sT$ and assume that the integral operator from $\sL_{2}(\sT)$ into itself with kernel $V_{s}$ (which is known as covariance operator) being injective, self-adjoint and non-negative definite. 
Now, due to Mercer's theorem \citep{3085302a-23ec-3d75-8a35-30223fd7593a}, for a symmetric, continuous and non-negative definite kernel function $V_{s}$, we have the following representation 
\begin{equation}
\label{Chapter2-qif-Eq:Mercer}
    V_{s}(t_{1}, t_{2}) = \sum_{r=1}^{\infty}\lambda_{sr}\phi_{sr}(t_{1})\phi_{sr}(t_{2}),
\end{equation} 
where $\{(\lambda_{sr}, \phi_{sr})\}_{r \geq 1}$ are the set of eigen-components, and the convergence of Equation \eqref{Chapter2-qif-Eq:Mercer} is in $L_2$ sense. Here the eigen-values for each group $\lambda_{s1} \geq \lambda_{s2} \geq \ldots > 0$ are non-increasing sequence of eigen-values tending to zero, and $\{ \phi_{sr}\}_{r=1}^{\infty}$ is an orthonormal basis of $\sL_{2}(\sT)$ with the following relation 
\begin{equation}
\label{int:eq}
    \int_{\sT}V_{s}(t_{1}, t_{2})\phi_{sr}(t_{2})dt_{1} = \lambda_{sr}\phi_{sr}(t_{2}); \qquad r \geq 1.
\end{equation}
Further, in the same spirit of Equations \eqref{Chapter2-qif-Eq:Mercer} and \eqref{int:eq}, due to Kosambi-Karhunen-Lo\`eve expansion \citep{MR0009816, karhunen1946spektraltheorie, loeve1946functions}, we have $$\beta_{s}(t) = \sum\limits_{r=1}^{\infty}b_{sr}\phi_{sr}(t)$$ in $L_2$ sense, and $$X_{s}(t) = \E\{X_{s}(t)\} + \sum\limits_{r=1}^{\infty}\xi_{sr}\phi_{sr}(t)$$ with probability 1,  where $b_{sr}$ and $\xi_{sr}$ ($r \geq 1$) are coefficients of the expansions. In other words, $$b_{sr} = \int_{\sT}\beta_{s}(t)\phi_{sr}(t)dt,$$ and $$\xi_{sr} = \int_{\sT}[X_{s}(t) - \E\{X_{s}(t)\}]\phi_{sr}(t)dt.$$ Using these relationships, Model \eqref{eq:model} becomes $$Y_{s} = \alpha_{s} + \sum\limits_{r=1}^{\infty}b_{sr}\xi_{sr} + \epsilon_{s}, \qquad s = 1, 2, $$ where for a fixed $s$, $\xi_{sr}$, ($r \geq 1$) are uncorrelated centered random variables with variance $\lambda_{sr}$. Additionally, $b_{sr} = \E\{\xi_{sr}Y_{s}\}/\lambda_{sr},$ for $s = 1, 2$ and $r \geq 1$.  
\par
Now, for the given data $\{Y_{s, i}, X_{s, i}\}_{i = 1}^{n_{s}}$ ($s = 1, 2$), we first estimate $V_{s} (t_1, t_2)$ by empirical covariance function $$\widehat{V}_{s}(t_{1}, t_{2}) = \frac{1}{n_{s}}\sum\limits_{i=1}^{n_{s}}\{ X_{s, i}(t_{1}) -\overline{X}_{s}(t_{1}) \}\{ X_{s, i}(t_{2}) -\overline{X}_{s}(t_{2}) \}$$ for $t_{1}, t_{2}\in \sT$, where $\overline{X}_{s}(t) = \frac{1}{n_{s}}\sum\limits_{i=1}^{n_{s}}X_{s, i}(t)$. Next, using empirical spectral decomposition, we have $$\widehat{V}_{s}(t_{1}, t_{2}) = \sum_{r=1}^{\infty}\widehat{\lambda}_{sr}\widehat{\phi}_{sr}(t_{1})\widehat{\phi}_{sr}(t_{2}),$$
where $\widehat{\lambda}_{s1} \geq \widehat{\lambda}_{s2} \geq \ldots \geq 0$ are non-decreasing (almost surely) sequence of empirical eigen-values, and the orthonormal eigen-functions $\{ \widehat{\phi}_{sr} \}_{r \geq 1}$ such that $ \int_{\sT}\widehat{V}_{s}(t_{1}, t_{2})\widehat{\phi}_{sr}(t_{1})dt_{1} = \widehat{\lambda}_{sr}\widehat{\phi}_{sr}(t_{2})$ almost everywhere with respect to $t_{2}\in{\cal{T}}$, for $r \geq 1$ and $s = 1, 2$. Note that, since the rank of $\widehat{V}_{s}$ is finite, based on the augmented version of the expansion of $\widehat{b}_{sr}$, we have 
\begin{equation}
\label{hat:beta}
    \widehat{\beta}_{s}(t) = \sum_{r=1}^{\kappa_{s, n_{s}}}\widehat{b}_{sr}\widehat{\phi}_{sr}(t),
\end{equation}
where $\kappa_{s, n_{s}}$ is the cut-off level such that $\kappa_{s, n_{s}} \rightarrow \infty$ as $n = \min(n_{1}, n_{2}) \rightarrow \infty$ (see, e.g.,  \citet{hall2007methodology}), $$\widehat{b}_{sr} = \frac{1}{n_{s}}\sum\limits_{i = 1}^{n_{s}}\widehat{\xi}_{sr}Y_{s, i}/{\widehat{\lambda}_{sr}} \text{ and }
\widehat{\xi}_{s, r} = \frac{1}{n_{s}}\sum\limits_{i = 1}^{n_{s}}\left[\int_{\sT}[X_{s, i}(t) - \overline{X}_{s}(t)]\widehat{\phi}_{sr}(t)dt\right].$$ 
\par
Next, we discuss the formulation of the test statistic for testing $H_{0}$ against $H_{1}$ described in Statements \eqref{eq:null} and \eqref{eq:alt}, respectively. 

\subsection{Formulation of test statistic}
\label{subsec:formulation}
Let us consider arbitrary time-points $0 \leq t_{1} < \ldots < t_{m_{n}} \leq 1$, and denote $U_{j} = \widehat{\beta}_{1}(t_{j}^{\dagger})$ and $V_{j} = \widehat{\beta_{2}}(t_{j}^{\dagger})$ for $j = 1, \ldots, m_{n}$, where $\widehat{\beta}_{1}(.)$ and $\widehat{\beta}_{2}(.)$ are the same as defined in Equation \eqref{hat:beta}, and  $t_{1}^{\dagger}, \ldots, t_{m_{n}}^{\dagger}$ are from the set $\{t_{1}, \ldots, t_{m_{n}}\}$ such that $\widehat{\beta}_{2}(t_{1}^{\dagger}) \leq \ldots \leq \widehat{\beta}_{2}(t_{m_{n}}^{\dagger})$.
Consider now the following model: 
\begin{equation}
\label{eq:single-index}
    U_{j} = g(V_{j}) + \eta_{j},
\end{equation}
where $g$ is the unknown function described in the hypotheses $H_{0}$ in Statement \eqref{eq:null} and $H_{1}$ in Statement \eqref{eq:alt}, and $\eta_{j}$ is the random error. Note that here the number of time points $m_{n}$ depends on $n$, and $m_{n} \rightarrow \infty$ as $n \rightarrow \infty$. In view of the aforesaid fact, the mutually dependent random variables $U_{j}$, $V_{j}$ and $\eta_{j}$ defined in Model \eqref{eq:single-index} depend on $n$ as well, although close inspection in the proofs indicates that this dependence structure does not have any impact on the limiting properties of the test statistic, and consequently, that of the test. To carry out the testing of the hypothesis problem described in $H_{0}$ against $H_{1}$ (see Statements \eqref{eq:null} and \eqref{eq:alt}), one needs to estimate the second order derivative of $g$, which is indicated from the definition of ${\cal{C}}_{v_1, v_2}$ (see Equation \eqref{eq:C}). In this study, to estimate the second-order derivative of $g$, we adopt the well-known 
local quadratic smoother \citep{fan1996local} i.e., a special case of local polynomial regression with degree 2. The implementation of the local quadratic smoothing technique in this case is as follows. 
\par
Suppose that $K(t)$ denotes the kernel function, and the sequence of positive smoothing bandwidths $\{h_{m_{n}}\}_{n\geq 1}$ is such that $h_{m_{n}} \rightarrow 0$ as $n\rightarrow\infty$. 
Here $K(\cdot)$ is a non-negative function such that $\int K(t)dt = 1$, symmetric around 0, i.e.,  $K(-t) = K(t)$ for all $t\in\bbR$, and $K_{h_{m_{n}}}(t) = {h_{m_{n}}}^{-1}K(t/h_{m_{n}})$ for all $t\in\bbR$. In addition, we define $\sX(v)$ is a matrix of order $m_{n}\times 4$ where $\sX_{l_{1}, l_{2}}(v) = (V_{l_{1}} - v)^{l_{2}-1}$ for $l_{1} = 1, \cdots, m_{n}$ and $l_{2} = 1, 2, 3, 4$; $$\sW_{h_{m_{n}}}(u) = diag\left\{K\left(\frac{V_{1}-v}{h_{m_{n}}}\right), \ldots, K\left(\frac{V_{m_{n}}-v}{h_{m_{n}}}\right)\right\} \in \bbR^{m_{n}\times m_{n}}.$$ Furthermore, we define a vector $\bss_{h_{m_{n}}}$ with length $m_{n}$ whose entries are denoted by $s_{j, h_{m_{n}}}(v, 2)$ for $j = 1, \cdots, m_{n}$ as
\begin{align*}
    \label{eq:s}
       &\bss_{h_{m_{n}}}(v,2) = (s_{1,h_{m_{n}}}(v,2),\ldots,s_{m_{n},h_{m_{n}}}(v,2))^{\tp}\\
       &= 2\be_{3}^{\tp}\left\{\sX(v)^{\tp}\sW_{h_{m_{n}}}(v)\sX(v)\right\}^{-1}\sX(v)^{\tp}\sW_{h_{m_{n}}}(v) \in \bbR^{m_{n}},
       \numberthis
    \end{align*}
where  $\be_{3} = (0,0,1,0)^{\tp}$. Next, the second derivative of the function $g$ is given by  $\widehat{g}''(u) = \widehat{\gamma}_{2}$ obtained from
\begin{align*}
\label{eq:ls}
    &(\widehat{\gamma}_{0}, \widehat{\gamma}_{1}, \widehat{\gamma}_{2}, \widehat{\gamma}_{3})\\
    &= \arg\min_{\gamma_{0}, \gamma_{1}, \gamma_{2}, \gamma_{3}}
    \sum_{j=1}^{m_{n}}
    \left\{ 
        U_{j} - \gamma_{0} - \gamma_{1}(V_{j} - v) - \gamma_{2}(V_{j} - v)^{2} - \gamma_{3}(V_{j} - v)^{3}
    \right\}^{2}K_{h_{m_{n}}}(V_{j}-v).
    \numberthis
\end{align*}
In matrix notation, using the formulation in Equation \eqref{eq:ls} and definition of $\bss_{h_{m_{n}}}$ above, $\widehat{g}_{h_{m_{n}}}''(v)$ can be written as 
\begin{equation}
\label{eq:g-dash}
    \widehat{g}_{h_{m_{n}}}''(v) = \sum_{j=1}^{m_{n}}s_{j, h_{m_{n}}}(v; 2)\widehat{\beta}_{1}(t_{j}^{\dagger}).
\end{equation}
We use local polynomial estimators since in a classical setting with fully observed data, the estimators based on local polynomial regression are known to have useful properties with regard to the boundary condition and sampling design (see \citet{fan1996local}). Moreover, it provides a complete asymptotic description, such as consistency and distributional convergence, that could be useful in Section \ref{sec:theory}. 
\par
Finally, to test $H_{0}$ against $H_{1}$ (see Statements \eqref{eq:null} and \eqref{eq:alt}), we define the test statistic as 
\begin{equation}
\label{stat}
   T_{n} = \sup_{v \in [\widehat{v}_{1}, \widehat{v}_{2}]}|\widehat{g}_{h_{m_{n}}}''(v)|, 
\end{equation}
where $\widehat{v}_{1} = \inf\limits_{t \in [0, 1]} \widehat{\beta}_{2}(t)$ 
and $\widehat{v}_{2} = \sup\limits_{t \in [0, 1]} \widehat{\beta}_{2}(t)$. 
We reject $H_{0}$ at level $\alpha ~ (\in [0, 1])$ if and only if $T_{n} > t_{\alpha}$, where 
$t_{\alpha}$ is such that $\bbP_{H_{0}}\{ T_{n} > t_{\alpha}\} = \alpha$. 
\par
In view of Equations \eqref{eq:s}, \eqref{eq:ls}, \eqref{eq:g-dash} and \eqref{stat}, observe that the expression of the test statistic $T_{n}$ involves the bandwidth, and hence, to implement the test, one needs to choose the bandwidth appropriately to have the optimum performance of the test. The detailed method for bandwidth selection is discussed in Section \ref{sec:band}.
Afterwards, the next step will be to implement the test for a given dataset. The natural answer is to derive the exact distribution of the test statistic $T_{n}$ as described in Equation \eqref{stat}, which enables the computation of the critical value and the power of the test. However, the complex terms involved in $T_{n}$ make the derivation of the exact distribution intractable, which drives us to think of some alternative methodology to implement the test. The next possibility is to derive the asymptotic distribution of $T_{n}$ (see Theorem \ref{thm:T}), and it is indeed true that implementing the test based on the asymptotic distribution of $T_{n}$ makes sense only when the sample size is sufficiently large. For all these reasons, particularly for data with small and moderately small sample sizes, we use a bootstrap technique, which is described in the following. The flowchart of the bootstrap procedure is described in Section \ref{sec:algo}.

\subsection{Selection of bandwidth $h_{m_{n}}$}
\label{sec:band}
The expression of the test statistic $T_{n}$ (see Equation \eqref{stat}) involves the bandwidth, and hence, to implement the test, one needs to choose a bandwidth appropriately to have the optimum performance of the test. One can argue that the higher order derivatives of $g$ can be estimated by the higher order derivatives of $\widehat{g}$, and hence, it is tractable to estimate the higher order derivatives of an unknown function. However, this is not realistic when the information obtained from the data contains a significant amount of interference or irregularities. The quality of estimation may worsen as we elevate the order of derivatives, and the choice of bandwidth becomes more crucial as the order of the derivative increases. Historically, it is observed that the choice of bandwidth does not impact the test (see, e.g.,  \citet{dette2006simple}) when the tuning parameter is sufficiently small. In this article, we consider the rule of thumb approach as discussed in \citet{fan1996local}, where the method starts with the optimal bandwidth that minimizes the mean integrated squared error $h_{0} = C_{\nu, p}(K)\left\{\frac{\int\eta^{2}(v)\omega(v)/f(\omega)dv}{\int \{g^{[p+1]}(v)\}^{2}\omega(v)dv } \right\}^{1/(2p+3)} m_{n}^{1/(2p+3)}$, where $C_{\nu, p}$ is some positive constant that depends on the order of derivative $\nu$, order of the polynomial $p$ in local polynomial regression and the underlying kernel $K$. The formula for $h_{0}$ contains some unknown quantity such as error $\eta(\cdot)$, $(p+1)$-th order derivative of the unknown function $g$, viz., $g^{(p+1)}(\cdot)$ and the density of $V$, viz., $f(v)$; however we fix the weight $\omega$ as positive function that smoothly vary over $v$. Based on the pilot estimates of $g$ and $\eta^{2}$, by fixing $\omega(v) = f(v)\omega_{0}(v)$ for some specific function $\omega_{0}$, we consider the rule of thumb selector 
\begin{equation}
\label{eq:opt_h}
    h_{ROT} =  C_{\nu, p}(K)\left\{ 
        \frac{\breve{\eta}^{2}\int \omega_{0}(v)dv 
    }{
        \sum_{j=1}^{m_{n}}\{ \breve{g}^{[p+1]}(V_{j}) \}^{2}\omega_{0}(V_{j})
    } \right\}^{1/2p+3}.
\end{equation}
However, the aforementioned bandwidth does not perform well when the noise-to-signal ratio, $\eta^{2}/\Var{g(V)}$, is very high. Specifically, in our problem, we are unable to monitor this ratio effectively, necessitating a correction by multiplying $h_{ROT}$ by a constant.

\subsection{Implementation of the test (Bootstrap)}
\label{sec:algo}
In this section, we describe the implementation of the test using a residual bootstrap technique, which is particularly suitable for small to moderately sized datasets. The flowchart of the residual-based bootstrap procedure is described in Algorithm \ref{algo}. The bootstrap procedure consists of the following steps based on the data 
$$\sD = \left\{(Y_{s, i}, X_{s, i}(t_{j}), t_{j}): j = 1, \ldots, m_{n}; i = 1, \ldots, n_{s}; s = 1, 2\right\}.$$ Instead of generating the bootstrap samples from the data $\sD$, we first compute the value of the test statistics $T_{n}$ for each of the datasets based on the point-wise estimate of the slope functions $U_{j} = \widehat{\beta}_{1}(t_{j}^{\dagger})$ and $V_{j} = \widehat{\beta_{2}}(t_{j}^{\dagger})$ for $j = 1, \ldots, m_{n}$ as defined in Equation \eqref{hat:beta} where the set $\{t_{1}^{\dagger}, \ldots, t_{m_{n}}^{\dagger}\}$ is such that $\widehat{\beta}_{2}(t_{1}^{\dagger}) \leq \ldots \leq \widehat{\beta}_{2}(t_{m_{n}}^{\dagger})$, and the bootstrap resamples are generated from central residuals (see lines 6--13 in Algorithm \ref{algo}).
To implement Algorithm \ref{algo}, there is no need to calculate the asymptotic bias and variance of $T_{n}$. In particular,  line 14 computes the $p$-value of the test, and line 15 computes the critical value of the test. Note that the test using Algorithm \ref{algo} is implementable when the sample sizes $n_1$ and $n_2$ are fixed along with the fixed number of time points.

\begin{algorithm}[h]
\begin{algorithmic}[1]
\label{algo}
    \State Estimate $\beta_{1}(.), \beta_{2}(.)$ using  Equation \eqref{hat:beta} as discussed in Section \ref{subsec:estimation_of_beta}.
    \State Estimate $g''(.)$ based on Equation \eqref{eq:g-dash} as discussed in Section \ref{subsec:formulation}.
    \State Obtain the optimal bandwidth $h_{opt}$ using Equation \eqref{eq:opt_h} as discussed in Section \ref{sec:band}. 
    \State Compute $\widehat{v}_{1} = \min\limits_{j = 1, \ldots, m_{n}} V_{j}$ and $\widehat{v}_{2} = \max\limits_{j = 1, \ldots, m_{n}} V_{j}$, 
    where, $V_{j} = \widehat{\beta_{2}}(t_{j}^{\dagger})$ for $j = 1, \ldots, m_{n}$ and set $\{t_{1}^{\dagger}, \ldots, t_{m_{n}}^{\dagger}\}$ is such that $\widehat{\beta}_{2}(t_{1}^{\dagger}) \leq \ldots \leq \widehat{\beta}_{2}(t_{m_{n}}^{\dagger})$.
    \State Compute the value of the test statistic $T_{n} = \displaystyle\sup_{v \in [\widehat{v}_{1}, \widehat{v}_{2}]}|\widehat{g}_{h_{opt}}''(v)|$ based on $\{U_{j}, V_{j}: j = 1, \cdots, m_{n}\}$ where $U_{j} = \widehat{\beta}_{1}(t_{j}^{\dagger})$. 
    \State Compute initial residuals $\Tilde{\eta}_{j} = U_{j} - \widehat{g}(V_{j})$, for $j = 1, \ldots, m_{n}$.
    \State Let $\overline{\eta} = \frac{1}{m_{n}}\sum_{j=1}^{m_{n}}\Tilde{\eta}_{j}$ and define the central residual, $\widehat{\eta}_{j} = \Tilde{\eta}_{j} - \overline{\eta}$.
    \State Fit the simple linear regression (when $g_{H_{0}}(v) = \alpha_{0} + \beta_{0}v$) to estimate the $g$ under $H_{0}$. 
    \For{$b\gets 1, \ldots, B$}
        \State Draw a bootstrap sample $\{\sZ^{*}_{b} = (U_{j, b}^{*}, V_{j})\}$ where 
        $U_{j, b}^{*} = \widehat{g}_{H_{0}}(V_{j}) + \eta_{j, b}^{*}$, where $\eta_{j}^{*}$ are obtained form $\widehat{\eta}_{1}^{*}, \ldots, \widehat{\eta}_{m_{n}}^{*}$ with replacement conditional on $V_{j}$s. 
        \State  Calculate the bootstrap version of $g''$, viz. $\widehat{g}^{''*}_{b}$ based on $h_{opt}$. 
        \State  Calculate the bootstrap version of test statistic $\widehat{T}_{b, n}^{*}$, where $\widehat{T}_{b, n}^{*} = \sup\limits_{v \in [\widehat{v}_{1}, \widehat{v}_{2}]}|\widehat{g}_{b, h_{opt}}^{''*}(v)|$.
    \EndFor
    \State $p$-value of the test is given by $1-B^{*}/B$, where $B^{*} = \max\{b: \widehat{T}_{(b, n)}^{*} \geq T_{n}\}$ for $\widehat{T}_{(1, n)}^{*} \leq \ldots \leq \widehat{T}_{(B, n)}^{*}$.
    \State For any $\alpha\in (0, 1)$, $100(1 - \alpha)\%$  critical value can be obtained as $(1 - \alpha)$-th quantile of $\widehat{T}_{(b, n)}^{*}$, where $b = 1, \ldots, B$.
  \end{algorithmic}
  \caption{The bootstrap-based algorithm to compute the $p$-value of the proposed test based on $T_{n}$.}
  \label{algo}
\end{algorithm}
\par
The validity of the aforementioned residual-based bootstrap technique is established in Theorem \ref{thm:bootstrap} (see Section \ref{sec:AB}). The assertion in Theorem \ref{thm:bootstrap} indicates that the Kolmogorov-Smirnov distance between the conditional distribution function (conditioning on the given dataset) of $T_{n}$ and the distribution function of $T_{n}$ can be made arbitrarily small as the sample size of the data set is sufficiently large.

\section{Main results}
\label{sec:theory}
In this section, we study the asymptotic distribution of $T_{n}$. The following assumptions are needed to have the limiting distribution of $T_{n}$. 
\par
The assumptions on the covariate $X_{s}(t)$ ($s = 1, 2$). Here $t\in {\cal{T}} = [0, 1]$ unless mentioned otherwise. 

\begin{enumerate}[label=(C\arabic*)]
    \item\label{cond:X:moment} $\int\limits_{0}^{1}\E\{X_{s}^{4}(t)\}dt < \infty$ for $s = 1, 2$. 

    \item\label{C2} $\xi_{sr} = \int\limits_{0}^{1} (X_{s}(t) - \E\{X_{s}(t)\})\phi_{sr}(t)dt$ has mean zero and variance $\lambda_{sr}$ such that 
    $\E\{\xi_{sr}^{4}\} \leq C\lambda_{sr}^{2}$. Here $\lambda_{sr}$ and $\phi_{sr}(.)$ are the same as defined in Equation \eqref{Chapter2-qif-Eq:Mercer}, and the constant $C > 1$ is such that $E\{\epsilon_{s}^{2}\} < C$ for $s = 1, 2$.

    \item\label{cond:space} 
    $\lambda_{sr} - \lambda_{s, r+1} \geq C^{-1}r^{-\gamma_{1}-1}$ for $r \geq 1$ and $s = 1, 2$, where $\lambda_{sr}$ is the same as defined in Equation \eqref{Chapter2-qif-Eq:Mercer}, $\gamma_{1} > 1$, and C is the same as defined in Condition \ref{C2}.
    \end{enumerate}

    \begin{remark}
    Condition \ref{cond:X:moment} will be satisfied when the covariate $X_{s}(t)$ has a pointwise finite fourth moment. In fact, since $[0, 1]$ is a compact set, the continuity of the sample paths of $X_{s}(t)$ also ensures Condition \ref{cond:X:moment}. Overall, Condition \ref{cond:X:moment} indicates that the paths of $X_{s}(t)$ should not explode arbitrarily. Condition \ref{C2} provides constrains regarding the total mean deviation of $X_{s}(t)$ from its mean. It asserts that the point-wise peakedness, which is measured by the ratio between the fourth and the second moments of a random variable, of $X_{s}(t)$ should have some bound, which is expected for a reasonably smooth stochastic process. Condition \ref{cond:space} indicates that the variation explained by the consecutive principal components should not be close to each other with a certain rate. Geometrically it interprets that the decomposition of the variation along the different axes should be different from each other by a certain amount.   
    \end{remark}
    \par
    The assumption on the unknown slopes $\beta_{s}(t)$: 
    \par
    \begin{enumerate}[label=(S\arabic*), align=left]

    \item\label{cond:b} $|b_{sr}|\leq C r^{-\gamma_{2}}$, where $b_{sr}$ is the same as defined in Section \ref{subsec:model}, $C$ is some uniform constant and $\gamma_{1} > 1$, and $\gamma_{2}$ is such that $1 + ({\gamma_{1}}/{2}) < \gamma_{2}$.
    
    \item\label{cond:kappa} 
    For $s=1, 2$, $\kappa_{s, n_{s}}/n_{s}^{1/(\gamma_{1}+2\gamma_{2})} \rightarrow \tau_{s}$, where $0 < \tau_{1}, \tau_{2} < \infty$. Here, $\kappa_{s, n_{s}}$ is the same as defined in Equation \eqref{hat:beta}.
    \end{enumerate}

    \begin{remark}
    In Condition \ref{cond:b}, the upper bound of $b_{sr}$ indicates that variation explained by each principal component of $\beta_{s}(.)$ cannot be more than a certain threshold, and moreover, it depends on the rank of the principal components. Condition \ref{cond:kappa} implies that after appropriate normalization, the value of truncation, i.e., $\kappa_{s, n_{s}}$ (see Equation \eqref{hat:beta}) in infinite expansion of $\beta_{s}(.)$ converges to some finite number, i.e., the infinite expansion of $\beta_{s}(.)$ can be approximated by a finite expansion with the largest principal components as long as the truncation is done following a certain order.   
    \end{remark}
    \par
    The assumptions on kernel function $K(.)$: 
    \par
    \begin{enumerate}[label=(K\arabic*), align=left]
    \item\label{cond:kernel} The kernel function $K$ is twice differentiable symmetric density function with compact support $S_{K}$ (for example $[-1, 1]$). For unbounded support $S_{K}$ (generic notation), $\int\limits_{S_{K}}y^{a}K(y)dy < \infty$ and $\int\limits_{S_{K}}y^{a}K^{2}(y)dy < \infty$ for $a \geq 8$.
    \item\label{cond:VC} We define 
    $$\sK_{6} = \left\{ 
        y \mapsto \left(\frac{x - y}{h}\right)^{\gamma} K\left(\frac{x-y}{h}\right): x \in \bbR;
        \gamma = 0, \ldots, 6;
        h > 0
    \right\},$$
    which is a Vapnik–Chervonenkis (VC) type class (see \citet{vapnik2015uniform} for details about VC class of functions) in  $\sL_{2}(Q)$, where $Q$ is an arbitrary probability measure, and $\sL_{2}(Q) = \{Q: \int f^{2} dQ < \infty, f\in\sK_{6}\}$. 
    \end{enumerate}

    \begin{remark}
    Condition \ref{cond:kernel} indicates that the kernel should be smooth enough, and moreover, the integrability conditions on the kernel function imply that the kernel function is supposed to be a light-tailed function. For example, the Gaussian kernel is such a kernel function. Condition \ref{cond:VC} is a well-known condition to achieve uniform convergence over a class of functions and process-level convergence. In fact, $\sK_{6}$ is a $P$-Donsker class, which follows from the assertion in Theorem 2.5.2 in \citet{vaart1996weak} as long as $P\in\sL_{2}(Q)$.
    \end{remark}

The assumption on the arbitrary transformation $g$ :
    \begin{enumerate}[label=(G\arabic*), align=left]
    \item\label{cond:g}
    The function $g$ is twice differentiable, and $g\in \sS(2+\delta_{0}, L)$, where the holder class of  function $\sS(\cdot, \cdot)$ is defined in Section \ref{Intro}, $\delta_{0} \in (2/3, 2]$, and $L > 0$.
    \end{enumerate}

    \begin{remark}
    Condition \ref{cond:g} indicates that the function $g$ needs to satisfy slightly more than twice differentiability for the technical reasons, though the test statistic $T_{n}$ (see Equation \eqref{stat}) and the statements of the hypotheses (see the Statements \eqref{eq:null} and \eqref{eq:alt}) only depend on the second order derivative of $g$.   
    \end{remark}

The assumption on the bandwidth $h_{m_{n}}$ :

    \begin{enumerate}[label=(B\arabic*), align=left]
    \item\label{cond:band} $h_{m_{n}} = m_{n}^{-{1}/{\omega}}$ for some $\omega > 0$ such that  $m_{n}h_{m_{n}}^{9}/\log m_{n} \rightarrow c_{0} \geq 0$ and $m_{n}h_{m_{n}}/\log m_{n} \rightarrow \infty$ as $m_{n}\rightarrow\infty$ along with $n\rightarrow\infty$.
    \end{enumerate}

    \begin{remark}
    Condition \ref{cond:band} indicates the the sequence of bandwidth $h_{m_{n}}$ must satisfy some order condition with respect to the sample size $n$ and the number of discrete points, i.e., $m_{n}$, over the time parameter space $[0, 1]$. 
    \end{remark}

    \par 
    The assumption on $m_n$: 
    \begin{enumerate}[label=(M\arabic*), align=left]
    \item\label{cond:m} $m_{n}\rightarrow\infty$ as $n\rightarrow\infty$.
    \end{enumerate}

    \begin{remark}
    Condition \ref{cond:m} indicates that the number of discrete points chosen over $[0, 1]$ should be sufficiently large as the sample size becomes sufficiently large. 
    \end{remark}
    
    \par
    The assumption on the errors, i.e., $\epsilon_{s}$ $(s = 1, 2)$ and $\eta$: 
    \par
    \begin{enumerate}[label=(E\arabic*), align=left]
    \item\label{cond:error}
    For $s = 1, 2$, $\epsilon_{s, i}$ ($i = 1, \ldots, n_{s}$) are i.i.d.\ random variables (identically distributed with $\epsilon_{s}$) with 
    $\E\{\epsilon_{s}\} = 0$, $\Var\{\epsilon_{s}\} = \sigma_{s}^{2} < \infty$. Moreover, $\epsilon_{1, p}$ and $\epsilon_{2, q}$ are independently distributed for all $p = 1, \ldots, n_1$ and $q = 1, \ldots, n_2$. 
 
    \item\label{cond:eta}
    $\E\{\eta\} = 0$ and $\Var\{\eta\} = \sigma_{n}^{2}$, where $0 < \sigma_{n}^{2} \rightarrow 0$ as $n \rightarrow \infty$. Here note that the random variable $\eta$ depends on $n$ (see Section \ref{subsec:formulation}).

    \item\label{cond:errx}
    For each $s = 1, 2$, and for each $t\in [0, 1]$, $\epsilon_s$ and $X_{s}(t)$ are independent random variables.  
    \end{enumerate}
    \begin{remark}
    Condition \ref{cond:error} indicates that the errors in Model \eqref{eq:model} are independent and identically distributed with mean zero and have non-zero finite variances. This also indicated that the two groups are independent. Condition \ref{cond:eta} is a mild condition on the error distribution based on the Model \eqref{eq:single-index}, and Condition \ref{cond:errx} is common across most of the random design model.  
    \end{remark}

\subsection{Asymptotic properties of $\widehat{g}_{h_{m_{n}}}''(\cdot)$}
\label{sec:theory-g}

We first want to discuss a few intermediate results, which are important components to study the asymptotic distributions of $\widehat{g}_{h_{m_{n}}}''(\cdot)$ and $T_n$, and they are worthy of study because of their own strength. 

\begin{lemma}[\citet{hall2007methodology}]
\label{lemma:beta}
    Under the Conditions \ref{cond:X:moment}, \ref{C2}, \ref{cond:space}, \ref{cond:b}, \ref{cond:kappa}, \ref{cond:m}, \ref{cond:error} and \ref{cond:errx}, 
    $\int\limits_{0}^{1} (\widehat{\beta}_{s}(t) - \beta_{s}(t))^{2}dt = O_{\bbP}\left(n^{-\frac{2\gamma_{2} - 1}{\gamma_{1} + 2\gamma_{2}}}\right)$, for the constants $\gamma_{1}$ and $\gamma_{2}$ defined in the Condition \ref{cond:b}.
\end{lemma}
\begin{lemma}
\label{lemma:density}
    Under the Conditions \ref{cond:X:moment}, \ref{C2}, \ref{cond:space}, \ref{cond:b}, \ref{cond:kappa}, \ref{cond:m}, \ref{cond:error}, and \ref{cond:errx} for each $t\in [0, 1]$,  $\widehat{\beta}_{s} (t)$, which depends on $n_{s}$, has a continuous distribution function with a compact support $\left[\min\limits_{t \in [0, 1]}\beta_{s}(t),  \max\limits_{t \in [0, 1]} \beta_{s}(t)\right],$ 
    and the associated probability density function $f_{\widehat{\beta}_{s}(t)}(.)$ is twice differentiable on the support, and bounded away from zero and infinity. Moreover, $f_{\widehat{\beta}_{s}(t)}^{[r]}(v_{0}) = O(f_{\widehat{\beta}_{s}(t)}^{[r+1]}(v_{0}))$ as $n_{s}\rightarrow\infty$ for any $r \geq 0$ and $s = 1, 2$, where $v_{0}$ is any arbitrary point contained in the support of $f_{\widehat{\beta}_{s}(t)}$, and $f^{[j]}_{\widehat{\beta}_{s}(t)}(v_{0})$ denotes the $j$-th derivative $f_{\widehat{\beta}_{s}(t)}$ at $v_{0}$. 
\end{lemma}
\begin{remark}
    Lemma \ref{lemma:beta} asserts that $\widehat{\beta}_{s}(t)$ converges to $\beta_{s}(t)$ in $L_{2}$ sense with a certain rate of convergence. Lemma \ref{lemma:density} indicates that for each $t$, $\widehat{\beta}_{s}(t)$ is a continuous random variable with a positive probability density function over a certain interval. Moreover, the probability density function of $\widehat{\beta}_{s}(t)$ at a fixed $t$ is reasonably light-tailed function, which follows from the fact that  $f_{\widehat{\beta}_{s}(t)}^{[r]}(v_{0}) = O(f_{\widehat{\beta}_{s}(t)}^{[r+1]}(v_{0}))$ as $n_{s}\rightarrow\infty$ for any $r \geq 0$ and $s = 1, 2$. This further indicates that $\widehat{\beta}_{s}(t)$ is expected to have finite moments. 
\end{remark}
\par
The following theorem demonstrates the order of point-wise bias and variance of $\widehat{g}''_{h_{m_{n}}}$ mentioned in Equation \eqref{eq:g-dash}.
\begin{theorem}
\label{thm:bias}
Assume that $\eta_{j}$ (see Model \eqref{eq:single-index}), which depends on $n$, holds Condition \ref{cond:eta}. Then 
under the Conditions \ref{cond:kernel}, \ref{cond:g} and the conditions in Lemmas \ref{lemma:beta} and \ref{lemma:density}, for any $v_{0} \in \bbR$ and for some $\delta > 0$, 
we have
\begin{align*}
&\E\left\{\widehat{g}_{h_{m_{n}}}''(v_{0})\right\} - g''(v_{0}) \numberthis\\
&= O\left(\left\{\frac{1}{f_{\widehat{\beta}_{2}(t)} (v_{0})} + 
    h_{m_{n}}f_{\widehat{\beta}_{2}(t)}^{[1]}(v_{0}) + 
    \sqrt{\frac{f_{\widehat{\beta}_{2}(t)}(v_{0})}{m_{n}h_{m_{n}}}}\right\}
    \left\{f_{\widehat{\beta}_{2}(t)}(v_{0})h_{m_{n}}^{\delta+2} + 
    \sqrt{\frac{f_{\widehat{\beta}_{2}(t)}(v_{0})}{m_{n}h_{m_{n}}}}\right\}\right),\\
    &~\mbox{and}\\
&\Var\left\{\widehat{g}_{h_{m_{n}}}''(v_{0})\right\}= O\left(\frac{\sigma^{2}_{n}}{m_{n}h_{m_{n}}^{5}f_{\widehat{\beta}_{2}(t)}(v_{0})}\right). \numberthis
\end{align*} Here $\widehat{g}_{h_{m_{n}}}''(.)$ is the same as defined in Equation \eqref{eq:g-dash}.
\end{theorem}

\begin{remark}
Theorem \ref{thm:bias} asserts the order of the bias term and the variance of $\widehat{g}_{h_{m_{n}}}''$ at an arbitrary point $v_{0}$. It follows from Condition \ref{cond:band} and in view of Lemma \ref{lemma:density} that variance of $\widehat{g}_{h_{m_{n}}}''(.)$ converges to zero as $\min(n_1, n_2)\rightarrow\infty$, though the bias is non-negligible. 
\end{remark}
Next, we study the uniform convergence of $\widehat{g}_{h_{m_{n}}}''$ in Theorem \ref{thm:emp}. Before stating the theorem, we introduce additional notations. Suppose that $\bOmega\in\bbR^{4\times 4}$  with $$((\bOmega_{k_{1}, k_{2}}))_{1\leq k_{1}\leq 4, 1\leq k_{2}\leq 4} = \int t^{k_{1}+k_{2}-2}K(t)dt.$$ In addition, we denote 
\begin{equation*}
    \bPsi_{h_{m_{n}}}(v; z_{1}, z_{2}) = (\psi_{0, h_{m_{n}}}(v; z_{1}, z_{2}), \psi_{1, h_{m_{n}}}(v; z_{1}, z_{2}), \psi_{2, h_{m_{n}}}(v; z_{1}, z_{2}), \psi_{3, h_{m_{n}}}(v; z_{1}, z_{2}))^{\tp},
\end{equation*}
where $\psi_{c, h_{m_{n}}}(v, z_{1}, z_{2}) = z_{2}\left(\frac{z_{1} - v}{h_{m_{n}}}\right)^{c}K\left(\frac{z_{1} - v}{h_{m_{n}}}\right)$, for $c \in \{ 0, 1, 2, 3\}$, and  further, designate  
$\bPsi_{h_{m_{n}}}(v) = \int \bPsi_{h_{m_{n}}}(v, Z_{1}, Z_{2})d\bbP_{m_{n}}(Z_{1}, Z_{2}).$
\par
Let us now discuss a few concepts. Suppose that the sequence $\sG_{m_{n}}$ is defined as $\sG_{m_{n}} = \{f : f(V)\in\bbR, V\in\mathcal{V}\},$ where $\mathcal{V}$ is the $\sigma$-field defined on the domain space of $f$.
Now, 
consider the following empirical process indexed by $\sG_{m_{n}}$, viz., $$\bbG_{m_{n}}f = \frac{1}{\sqrt{m_{n}}}\sum\limits_{j=1}^{m_{n}}\left\{ f(V_{i}) - \E\{f(V_{i})\right\}$$ for $f \in \sG_{m_{n}}$.
Besides, $B_{m_{n}}$ is a centered Gaussian process index by $\sG_{m_{n}}$ with covariance function $\E\left\{B_{m_{n}}(f_{1})B_{m_{n}}(f_{2})\right\} = \cov\{ f_{1}(V), f_{2}(V)\}$ for all $f_{1}, f_{2} \in \sG_{m_{n}}$. In fact, it is also possible to establish that  $|\bbG_{m_{n}} - B_{m_{n}}|_{\sG_{m_{n}}} = \sup\limits_{f\in \sG_{m_{n}}}|(\bbG_{m_{n}} - B_{m_{n}})f| = O_{\bbP}(r_{m_{n}})$, where $r_{m_{n}} \rightarrow 0$ as $n \rightarrow \infty$. All these facts give us the asymptotic properties of $\widehat{g}_{h_{m_{n}}}''$ in Theorem \ref{thm:emp} and weak convergence of $T_{n}$, which is stated in Theorem \ref{thm:T} in Section \ref{subsec:T}. 


\begin{theorem}
\label{thm:emp}
Under the Conditions \ref{cond:kernel}, \ref{cond:VC}, \ref{cond:g}, \ref{cond:band}, \ref{cond:eta} and the conditions in Lemmas \ref{lemma:beta} and \ref{lemma:density},
we have 
\begin{align*}
    &\sup_{v \in [\widehat{v}_{1}, \widehat{v}_{2}]}\left|\frac{\sqrt{m_{n}h_{m_{n}}^{5}}\left(\widehat{g}_{h_{m_{n}}}''(v) - 
    \E\{\widehat{g}''(v)\}\right) - \frac{1}{\sqrt{h_{m_{n}}}}\bbG_{m_{n}}[\be_{3}^{\tp}\bOmega^{-1}\bPsi_{h_{m_{n}}}(v)]}{\frac{1}{\sqrt{h_{m_{n}}}}\bbG_{m_{n}}[\be_{3}^{\tp}\bPsi_{h_{m_{n}}}(v)]}\right|\\ 
    &=  O(h_{m_{n}}) + O_{\bbP}\left(
        \sqrt{\frac{\log m_{n}}{m_{n}h_{m_{n}}}}
    \right).
    \numberthis
\end{align*}
Moreover, $\|\widehat{g}_{h_{m_{n}}}'' - g''\|_{\infty} 
    = O\left(h_{m_{n}}^{2+\delta} + \sqrt{\frac{h_{m_{n}}^{3}}{m_{n}}} \right) 
    + O_{\bbP}\left(
        \sqrt{\frac{\log m_{n}}{m_{n}^{2}h_{m_{n}}^6}}\right)$.
\end{theorem}
\begin{remark}
The statement of Theorem \ref{thm:emp} implies that~ $\widehat{g}_{h_{m_{n}}}''$ uniformly converges to $g''$ over a certain random interval. In other words, from statistical point of view, one can claim that $\widehat{g}_{h_{m_{n}}}''$ is a good candidate to estimate $g''$ at any arbitrary point inside the specified interval. 
\end{remark}

\subsection{Asymptotic distribution of $T_{n}$}
\label{subsec:T}
In Section \ref{sec:theory-g}, we studied the asymptotic properties of $\widehat{g}_{h_{m_{n}}}''$, and Theorems 
\ref{thm:bias} and \ref{thm:emp} indicate that $\widehat{g}_{h_{m_{n}}}''$ can approximate $g''$ arbitrary well under some regularity conditions. However, note that Theorems 
\ref{thm:bias} and \ref{thm:emp} do not assert anything about the weak convergence of the process $\widehat{g}_{h_{m_{n}}}''(.)$, which could enable us to insight the weak convergence of the test statistic $T_{n}$ defined in Equation \eqref{stat}. Here we study the asymptotic distribution of $T_{n}$, which is useful to implement the test when the sample size is sufficiently large enough.  
\begin{theorem}
\label{thm:T}
    Under the assumptions 
    \ref{cond:X:moment}, \ref{C2}, \ref{cond:space}, 
    \ref{cond:b}, \ref{cond:kappa},
    \ref{cond:kernel}, \ref{cond:VC}, \ref{cond:band}, 
    \ref{cond:error}, \ref{cond:eta} and \ref{cond:errx}, 
    there exists a Gaussian process $B_{m_{n}}$ defined on $\sG_{m_{n}}$ (see the description before the statement of Theorem \ref{thm:emp}), such that for any $f_{1}, f_{2} \in \sG_{m_{n}}$, $\E\left\{B_{m_{n}}(f_{1})B_{m_{n}}(f_{2})\right\} = \cov\{f_{1}(V), f_{2}(V)\}$ and 
    \begin{equation*}
        \sup_{t\in \bbR}\left|\bbP\left\{ \sqrt{m_{n}h_{m_{n}}^{5}}(T_{n} - T) \leq t\right\} - 
        \bbP\left\{\sup_{f \in \sG_{m_{n}}} |B_{m_{n}}(f)| \leq t\right\}\right|
        = O\left(\left(\frac{\log^{7}m_{n}}{m_{n}h_{m_{n}}^{5}}\right)^{1/8}\right), \numberthis
    \end{equation*} where $T_{n}$ is the same as defined in Equation \eqref{stat}, $T = \sup\limits_{v\in[v_1, v_2]}|g''(v)|$. Here $v_{1} = \inf\limits_{t\in[0, 1]} \beta_{2}(t)$ and $v_{2} = \sup\limits_{t\in[0, 1]} \beta_{1}(t)$.
\end{theorem}

\begin{remark}
In order to establish the result stated in Theorem \ref{thm:T}, the main idea of the study is to look at whether $\sqrt{m_{n}h_{m_{n}}^{5}}(T_{n} - T)\stackrel{d}= Z_{m_{n}}$ or not, where $Z_{m_{n}} = \sup\limits_{f\in \sG_{m_{n}}} \bbG_{m_{n}}f$. Then, next step is to find a certain random variable $\tilde{Z}_{m_{n}}$ so that $|Z_{m_{n}} - \tilde{Z}_{m_{n}}| = O_{\bbP}(r_{m_{n}})$ where $r_{m_{n}} \rightarrow 0$ as $n \rightarrow \infty$. Eventually, in this case, it is possible to show that $\tilde{Z}_{m_{n}}=\sup\limits_{f\in \sG_{m_{n}}}B_{m_{n}}f$, where $B_{m_{n}}$ is a centered Gaussian process as the same as defined in the description before the statement of Theorem \ref{thm:emp}.
\end{remark}

\begin{remark}
Theorem \ref{thm:T} shows that the test statistic after appropriate normalization can be approximated by the distribution of the random variable associated with the supremum of a certain Gaussian process, which leads to performing the testing the statistical hypothesis problem described in the Section \ref{sec:model} (equivalently Statements \eqref{eq:null} and \eqref{eq:alt}). 
\end{remark}
\begin{remark}
Theorem \ref{thm:T} further asserts that the test based on $T_{n}$ is consistent, i.e., $$\mathbb{P}_{H_{1}}(T_{n} > t_{\alpha})\rightarrow 1$$ as $n\rightarrow\infty$, where $t_{\alpha}$ is such that $$\mathbb{P}_{H_{0}}(T_{n} > t_{\alpha})\rightarrow\alpha$$ as $n\rightarrow\infty$. Here $\alpha\in (0, 1)$.  
\end{remark}
\par
We now would like to end this section with a discussion on the implementation of the test for testing null in Statement \eqref{eq:null} against alternative in Statement \eqref{eq:alt} using the result described in Theorem \ref{thm:T}. First of all, the test based on $T_{n}$ will be rejected at $100(1-\alpha)\%$ level of significance when $T_{n} \geq T + \frac{b_{\alpha}}{m_{n}h_{m_{n}}^{5}}$, where $b_{\alpha}$ is the $100 (1 - \alpha)\%$ quantile of the distribution of $\sup\limits_{f \in \sG_{m_{n}}} |B_{m_{n}}(f)|$, and observe that under $H_{0}$ (i.e., under Statement \eqref{eq:null}), $T = 0$, and under $H_{1}$ (i.e., under Statement \eqref{eq:alt}), $T > 0$. Now, since exact computation of any quantile of the distribution of the supremum of a certain Gaussian process is not tractable, we approximate the Gaussian process $B_{m_{n}}(.)$ by a certain sufficiently large dimensional multivariate normal distribution and component-wise maxima is considered as a realization of the random variable $\sup\limits_{f \in \sG_{m_{n}}} |B_{m_{n}}(f)|$. We repeat this experiment a large number of times, and then $100( 1 - \alpha)\%$ quantile of the empirical distribution function of the component-wise maxima can be considered as an approximation of $b_{\alpha}$, which is denoted as $\hat{b}_{\alpha}$. Then, when the alternative hypothesis (i.e., Statement \eqref{eq:alt}) is true, we compute $T_{n}$ from the given dataset obtained from the alternative hypothesized model and repeat this experiment a large number (say, L) of times. Finally, using the fact that $T > 0$ under alternative hypothesis, i.e., when Statement \eqref{eq:alt} is true, the power of the test will be $$\frac{1}{L}\sum\limits_{l = 1}^{L}\textbf{1}{\left\{T_{n, l}\geq\frac{\hat{b}_{\alpha}}{m_{n}h_{m_{n}}^{5}}\right\}},$$ where $T_{n, l}$ is the value of $T_{n}$ for the $l$-th ($l\in\{1, 2, \ldots, L\}$) simulated data obtained from alternative hypothesized model.  

\subsection{Asymptotic validity of the bootstrap method}
\label{sec:AB}
We described the bootstrap algorithm for how to implement the proposed test for small sample size data in Section \ref{sec:algo}. As we mentioned there as well, one needs to check the validity of the said bootstrap methodology described in Algorithm \ref{algo}, and Theorem \ref{thm:bootstrap} validates it. 
\begin{theorem}
    \label{thm:bootstrap}
    Under the assumptions \ref{cond:kernel}, \ref{cond:VC}, \ref{cond:band}, given data $(V_1, \ldots, V_{m_{n}})$, for each $b = 1, \ldots, B$, we have
    \begin{align*}
        &\sup_{t\in \bbR}\left| 
            \bbP\left\{ \sqrt{m_{n}h_{m_{n}}^{5}}|\widehat{T}_{b, n}^{*} - {T}_{n} | < t \mid (V_1, \ldots, V_{m_{n}}) \right\} 
            - \bbP\left\{\sqrt{m_{n}h_{m_{n}}^{5}}|{T}_{n} -  T| < t\right\}
        \right|\\
        &= O\left(\left(\frac{\log^{7}m_{n}}{m_{n}h_{m_{n}}^{5}}\right)^{1/8}\right) \numberthis
    \end{align*}
    almost surely for all bootstrap sample sequences (see Algorithm \ref{algo}) using the data $(V_1, \ldots, V_{m_{n}})$. Here $V_{i}$'s ($i = 1, \ldots, m_{n}$) are the same as defined in Section \ref{subsec:formulation}, where $\widehat{T}_{b, n}^{*}$ is the same as defined in line 11 in Algorithm \ref{algo}, $T_n$ is the same as defined in Equation \eqref{stat}, and $T$ is the same as defined in Theorem \ref{thm:T}. 
\end{theorem}

\begin{remark}
Observe that from Algorithm \ref{algo}, we used the data $(V_1, \ldots, V_{m_{n}})$ instead of $(U_1, V_1), \ldots, (U_{m_{n}}, V_{m_{n}})$ in implementing the test using residual based bootstrap methodology. For this reason, in the statement of Theorem \ref{thm:bootstrap}, the conditioning random variables are $(V_1, \ldots, V_{m_{n}})$. Moreover, another point needs to be mentioned that the size of all bootstrap resamples described in Algorithm \ref{algo} is the same, i.e., $m_{n}$, and it enables us to concisely present the result. After inspecting the proof of Theorem \ref{thm:bootstrap}, we have observed that similar results hold for unequal sample sizes of the bootstrap resamples. 
\end{remark}

\begin{remark}
 Theorem \ref{thm:bootstrap} asserts that after appropriate normalization, the bootstrap version of the test statistic has the same asymptotic distribution as that of the original test statistic with a similar Berry-Essen bound.    
\end{remark}

\section{Finite sample studies} 
\label{sec:simulation}

In this section, we study the finite sample performance of the proposed test under different situations. Now, let $\sT = [0, 1]$. The basis function $\phi(.)$ (see Equation \eqref{Chapter2-qif-Eq:Mercer}) is considered as Fourier basis for each group. That is, for each $s = 1, 2$ and for $r\in\mathbb{Z}^{+}$, $$\phi_{sr}(t) =  \sin(2\pi rt) 1_{\{r = 2l\}} + \cos(2\pi rt) 1_{\{r = 2l + 1\}},$$ where $l\in\mathbb{N}\cup\{0\}$ and $t \in [0, 1]$. We now generate the $(X_{s}, Y_{s})$ as follows: 
\begin{equation}
\label{eq:model-sim}
    Y_{s} = \int_{0}^{1} \bbeta_{s}(t)X_{s}(t)dt + \epsilon_{s}, \text{ where } \beta_{2}(t) = \sin(\pi t/2), \text{ and } \beta_{1}(t) = g(\beta_{2}(t)). 
\end{equation}
In this study, we consider the following examples of $g: {\sL}_{2}([0, 1])\rightarrow  {\sL}_{2}([0, 1])$. 
\begin{enumerate}
    \item \label{fn:linear} $g(v) = 10v + 1$
    \item \label{fn:poly} $g(v) = 1- 4v + 2v^{2} + 2v^{3} - v^{4}$
    \item \label{fn:exp} $g(v) = 0.3\exp(-3(v+1)^{2}) + 0.7\exp(-7(v-1)^{2})$
    \item \label{fn:sincos-log} $g(v) = \sin(8v) + \cos(8v) + \log(4/3 + v)$.
\end{enumerate}
\par
First of all, we assume that the sample size of the two groups is the same, i.e., $n_1 = n_2$, and it is denoted by $n$. Here $n \in \{100, 500, 1000, 2000\}$. For $s = 1, 2$, $\epsilon_{s}$ are generated from  $N(0, 0.1^2)$ distribution, and  
$$X_{s}(t) = c_{0} + \sum\limits_{r=1}^{50}a_{r}\sin(2\pi r t) + \sum\limits_{r=1}^{50}b_{r}\cos(2\pi rt).$$  Here, the observations on $c_{0}$ are generated from the standard normal distribution, the observations on $a_{r}$ are generated from $N(0, r^{-2\alpha_{0}})$ distribution, and the observations on $b_{r}$ are generated from $N(0, r^{-2\alpha_{0}})$ distribution. In this study, $\alpha_{0} \in \{1, 2, 3\}$. The integration in Equation \eqref{eq:model-sim} is approximated by a finite Riemann sum over 500 equidistant grid points over $[0, 1]$. Besides, $m = 1000$ for the aforementioned choices of $n$, 
and $t_j = \frac{j - 0.5}{m}$, where $j = 1, \ldots, m$. Here $\{t_{1}, \ldots, t_{m}\}$ is the partition of $[0, 1]$, where $\beta_{1}(.)$ and $\beta_{2}(.)$ are estimated at each time point $t_{j}$. 
\par
Recall from Section \ref{subsec:challenges} that the sample size (i.e., $n$) and the number of discrete points (i.e., $m_n$) are supposed to have strong influence on the results as $\hat{\beta}_{1}(t)$ and $\hat{\beta}_{2}(t)$ are observed on the discrete time points at $\{t_{1}, \ldots, t_{m_{n}}\}$. In this numerical study, we consider various transformations, including exponential transformation, logarithmic transformation, and many trigonometric transformations, and we have observed that for large values of $n$ and for corresponding large values of $m_{n}$, the results are becoming almost the same irrespective of the transformation between $m_{n}$ and $n$. Strictly speaking, as long as Condition \ref{cond:m} described in Section \ref{sec:theory} (i.e., $m_{n}\rightarrow\infty$ as $n\rightarrow\infty$) holds in some sense in practice, the final results become stable. In accordance with it, we reported the values of $n$ and $m_{n}$ described in the preceding paragraph.   
\par
Now, we would like to describe how the simulation studies are carried out. Here, for each of the situations, we perform 500 simulation replicates. To make it consistent with discussion in the previous sections, we estimate the slope functions based on the FPCA-based approach. In the first step of estimation, we run FPCA of functional covariates using one of the software that implements the estimation of eigenvalues and eigen-functions using `\texttt{FPCA}' function in available in \texttt{fdapace} packages \citep{fdapace} in R \citep{R}. The bandwidths are selected using generalized cross-validation,  and the Epanechnikov kernel $K(x) = 0.75(1-x^{2})_{+}$ is used for estimation, where $(a)_{+} = \max(a, 0)$ for any $a\in\mathbb{R}$. We select the number of basis functions based on the fraction of variance explained (FVE) criteria. In the second step, we assume 2000 equally spaced points in the range of $\widehat{\beta}_{2}$. The second order derivative of $g$ is estimated by `\texttt{locpol}' function in \texttt{locpol} package \citep{locpol} in R. We apply a local polynomial estimation method of degree 3, using weights defined by the indicator function on the interval $[0.1, 0.9] \subset [0, 1]$ to mitigate boundary effects, and a Gaussian kernel $K(x) = (1/\sqrt{2\pi})\exp(-0.5x^{2})$, where $x\in\mathbb{R}$. Additionally, for bandwidth selection in Step 2, we employ the rule-of-thumb technique \citep{fan1996local} available in the \texttt{locpol} package using the \texttt{thumbBw} function. Through extensive simulation studies, we observed that the optimal bandwidth chosen by \texttt{thumbBw} tends to produce under-smoothing when estimating the derivative of the target function $g$. To address this issue, we adjust the bandwidth by multiplying it by certain constants to reduce the under-smoothing effect.   
\par
To obtain the critical value of the test, three methodologies are implemented. The first methodology is the well-known Monte Carlo-based procedure. To implement this methodology, we repeatedly (500 times) generate the data from the Model \eqref{fn:linear} (i.e., $H_{0}$ is true) described at the beginning of the section and compute the value of the test statistic $T_{n}$ for each repetition. Afterward, as we are conducting the test at a 5\% level of significance, the 95\% quantile of the values of the test statistic is considered as the estimated critical value. The second approach employs the bootstrap procedure outlined in Algorithm \ref{algo}. To obtain the critical value of the test at the 5\% significance level using the bootstrap methodology, the reader is referred to line 15 of Algorithm \ref{algo}, where the procedure for selecting the critical value when $\alpha = 0.05$ is described. The third methodology relies on the asymptotic distribution, using the estimated $b_{\alpha}$, where $b_{\alpha}$ is $100~(1 - \alpha)\%$ quantile of the distribution of $\sup\limits_{f \in \sG_{m_{n}}} |B_{m_{n}}(f)|$, and ${\cal{G}}_{m_{n}}$ and $B_{m_{n}}(.)$ are the same as defined in the statement of Theorem \ref{thm:T}. For this method as well, to obtain the critical value at 5\% level of significance, we set $\alpha = 0.05$.  
\par
Now, let us denote $\hat{c}_{0.95}$ as the critical value of the test at 5\% level of significance. 
Next, to compute the power of the test, we repeatedly (500 times) generate the data from various models, i.e., the models described in Examples \eqref{fn:linear}, \eqref{fn:poly}, \eqref{fn:exp}, and \eqref{fn:sincos-log} at the beginning of the section, and let us denote $T_{n, b}$ is the value of $T_{n}$ for the $b$-th repetition ($b = 1, \ldots, 500$). Finally, the power of the test is computed as $\frac{1}{500}\sum\limits_{b = 1}^{500} \textbf{1}{(T_{n, b} > \hat{c}_{0.95})},$ where $\textbf{1}{(.)}$ is the usual indicator function. 
\par
In Table \ref{tab:full}, we report the size/power of the proposed test under different simulation examples. In addition, the estimates of the slope functions $\widehat{\beta}_{s}$ for $s = 1, 2$, and the estimate of the second order derivative of $g$ (viz., $\widehat{g}''$) are shown in Figures \ref{fig:linear}, \ref{fig:poly}, \ref{fig:exp} and \ref{fig:sincos-log} for different choices of $g$ described in the beginning of this section. In Table \ref{tab:full}, the test statistic and its standard deviation are presented. In general, all these results indicate that the proposed test performed reasonably well for the examples considered here. In conclusion, for smaller sample sizes, the tests that use critical values obtained either through the bootstrap methodology or the Monte Carlo simulation outperform the test that relies on the asymptotic critical value. In contrast, for larger sample sizes, the performance of the test based on the asymptotic distribution becomes comparable to that of the bootstrap-based test. Moreover, for sufficiently large sample sizes, all three testing procedures attain power approaching one, and the reported values for Example 1 indicate that the test based on all three approaches preserves the size of the test as well.

\begin{table}[]
    \centering
    \begin{tabular}{rrcccc}
    \toprule
    $n$ & $\alpha_{0}$ & \multicolumn{3}{c}{size/power}\\
    & & (Monte Carlo Approach) & (Bootstrap Approach) & (Asymptotic Approach)\\
    \midrule
    \multicolumn{5}{c}{Ex. \eqref{fn:linear}: $g(v) = 10v + 1$. $T = \sup_{v \in [0.1, 0.9]}|g''(v)| = 0$}\\
    \midrule
100 & 1 & 0.050 & 0.050 & 0.046 \\
500 & 1 & 0.050 & 0.051  & 0.054\\  
1000 & 1 & 0.050 & 0.050 & 0.050\\  
2000 & 1 & 0.050 & 0.049 & 0.040\\  
\hdashline
100 & 2 & 0.050 & 0.050  & 0.050\\
500 & 2 & 0.050 & 0.051 & 0.050\\  
1000 & 2 & 0.050 & 0.049 & 0.042\\  
2000 & 2 & 0.050 & 0.051  & 0.052\\ 
\hdashline
100 & 2 & 0.050 & 0.049 & 0.048\\
500 & 3 & 0.050 & 0.051 & 0.046\\  
1000 & 3 & 0.050 & 0.049 & 0.030\\  
2000 & 3 & 0.050 & 0.050 & 0.046\\  
    \midrule
\multicolumn{5}{c}{Ex. \eqref{fn:poly}: 
$g(v) = 1- 4v + 2v^{2} + 2v^{3} - v^{4}$. $T = \sup_{v \in [0.1, 0.9]}|g''(v)| = 7$}\\
    \midrule
100 & 1 & 1.000 & 0.990 &  0.943\\
500  & 1 & 1.000 & 0.998 &  1.000\\
1000 & 1 & 1.000 & 0.998&   1.000\\
2000 & 1 & 1.000 & 0.998&  1.000\\
\hdashline
100 & 2 & 1.000& 0.984 & 0.923\\
500  & 2  & 1.000& 0.990 & 1.000\\
1000 & 2  & 1.000 & 1.000 & 1.000\\
2000 & 2 & 1.000 & 1.000 & 1.000\\
\hdashline
100 & 3 & 0.824 & 0.990 &  0.890\\
500  & 3 & 1.000 & 0.996 &  1.000\\
1000 & 3 & 1.000 & 0.984 &  1.000\\
2000 & 3 & 1.000 & 0.990 &  1.000\\
    \midrule
\multicolumn{5}{c}{Ex. \eqref{fn:exp}: $g(v) =  0.3\exp(-3(v+1)^{2}) + 0.7\exp(-7(v-1)^{2})$. $T = \sup_{v \in [0.1, 0.9]}|g''(v)| = 7.854$}\\
    \midrule
100 & 1 & 0.840 & 0.998 & 0.840\\
500  & 1 & 1.000  & 0.990 & 1.000\\
1000 & 1 & 1.000 & 0.994 & 1.000\\
2000 & 1 & 1.000 & 0.998 & 1.000\\
\hdashline
100 & 2 & 0.448 & 0.996 & 0.448\\
500  & 2  & 0.946 & 1.000 & 0.946\\
1000 & 2  & 0.998 & 0.998 & 0.996\\
2000 & 2  & 1.000 & 0.994 & 1.000\\
\hdashline
100 & 3 & 0.148 & 0.866 & 0.146\\
500  & 3 & 0.584 & 0.928 & 0.570\\
1000 & 3 & 0.816  & 0.926 & 0.772\\
2000 & 3 & 0.968 & 0.934 & 0.958\\ 
    \midrule
\multicolumn{5}{c}{Ex. \eqref{fn:sincos-log}: $g(v) =  \sin(8v) + \cos(8v) + \log(4/3 + v)$. $T = \sup_{v \in [0.1, 0.9]}|g''(v)| = T = 90.986$}\\ 
    \midrule
100 & 1 & 1.000 & 1.000& 1.000\\
500  & 1 & 1.000 & 1.000& 1.000\\
1000 & 1 & 1.000 & 1.000& 1.000\\
2000 & 1 & 1.000 & 0.992& 1.000\\
\hdashline
100 & 2 & 1.000 & 1.000& 1.000\\
500  & 2 & 0.991 & 1.000& 1.000\\
1000 & 2 & 0.997 & 1.000& 1.000\\
2000 & 2 & 1.000 & 1.000& 1.000\\
\hdashline
100 & 3 & 1.000 & 0.996 & 1.000\\
500  & 3 & 1.000 & 0.994& 1.000\\
1000 & 3 & 1.000  & 0.996 & 1.000\\
2000 & 3 & 1.000 & 0.994 & 1.000\\
    \bottomrule
    \end{tabular}
    \caption{Simulated Size and Power for different choices of $g$. $T$ is defined in Theorem \ref{thm:T}}
    \label{tab:full}
\end{table}

\section{DTI data analysis}
\label{sec:real-data}

In this section, we illustrate the efficacy of the proposed methodology for testing of hypothesis problem in scalar-on-function regression on Diffusion Tensor Imaging (\texttt{DTI}) dataset. This dataset is publicly available in \texttt{refund} \citep{refund} package in R. The description of the data set is as follows. 
\par
The DTI data were collected at Johns Hopkins University and the Kennedy-Krieger Institute, where the clinical objective was to study the cerebral white-matter tracts of patients with multiple sclerosis (MS). MS is an autoimmune disease that causes lesions in the white matter tracts of affected individuals, leading to severe disability, and the DTI, i.e., diffusion tensor imaging tractography, allows to study of the white matter tracts by measuring the diffusivity of the water in the brain (See \citet{basser1994mr, le2001diffusion} for more details).
\par
In this study, among the seven variables of the DTI data set, we consider the corpus callosum (CCA) as the functional covariate (denoted as $X(t)$), and the score of the Paced Auditory Serial Addition Test (PASAT) is the real-valued response variable (denoted as $Y$). The CCA is a bundle of nerve fibers that allows the brain’s left and right hemispheres to communicate, and to quantify it, Fractional anisotropy (FA) tract profiles contain continuous summary of the white matter tracts parameterized by the distance along the tracts (denoted as time parameter $t$) derived from DTI. The PASAT test is a neuropsychological test, and it counts the total number of correct answers out of 60 questions in the test. It is a well-known test for patients with MS.
\par
Now, it is believed in the medical science that PASAT scores of the patients with MS, i.e., the response variable $Y$ with respect to the continuous summary of the white matter tracts
obtained from the CCA, i.e., the functional covariate $X(t)$, may vary over the time of conducting the tests (i.e., different visits). This work studies whether the slope parameter
associated with $X(t)$, in $Y$ on $X(t)$ functional regression model, for visit 1, has any non-linear relationship with the slope parameter associated with $X(t)$, in $Y$ on $X(t)$ functional regression model, for visit 2. In other words, non-technically speaking, from a medical science
point of view, the researchers would like to know whether the rate of change of PASAT score with respect to the continuous summary of the white matter for the patients with MS has any non-linear relationship or not over different visits. This is the main practical motivation of including DTI analysis in this work.
\par
A few more technical details about this data set: Here $X(t)$ is observed over 93 grid-points, i.e., the number of discrete time points is 93. The grid points are uniformly/equally spaced along the length of the tract. This data set has 142 unique subjects (see \cite{refund}), of whom 100 have MS disease. These 100 subjects constitute the sample of interest for the current study. These 100 subjects with MS contributed multiple visits, with an average of 3.4 repeated measurements per subject, having standard deviation $= 1.54$, minimum value $= 2$, and maximum value $= 8$. However, in this work, as we mentioned in the preceding paragraph, we restrict our analysis to the baseline visit and the immediate subsequent visit. Regarding the preprocessing of this data set, a standard pre-processing pipeline exists, and detailed steps can be found in \citet{le2001diffusion}. However, in this dataset, we did not need to perform any further pre-processing, as the data frame available in the refund package is already in a ready-to-use format. Furthermore, no smoothing steps were applied prior to conducting this data analysis.
\par
Result: Figure \ref{fig:dti} demonstrates the estimated slope parameters and the estimated second derivative of the unknown function $g$. Using the proposed bootstrap methodology in Algorithm \ref{algo}, we obtain the $p$-value as 0.571, which indicates that the data does favor the null hypothesis at $5\%$ level of significance. From the medical science point of view, one may conclude from this study that either there is some linear functional relationship between the effect of functional covariate CCA on scalar response PASAT score at two consecutive visits, or the effect of functional covariates on scalar response PASAT score at the second visit does not depend on that of the first visit.

\section{Discussion}
\label{sec:discussion}
In this article, we tried to identify an unknown transformation between two slope functions associated with a scalar on a functional linear regression model for two independent datasets. To investigate it, we formulated a test of the hypothesis problem and derived the asymptotic distribution of the test statistic. Moreover, the residual bootstrap-based technique has also been adopted to implement the test for small or moderately small sample-sized data, along with its theoretical justification. In this regard, we would like to emphasize that this work can be extended to the function-on-function linear regression (see \citet{MR4734835}) as well. To avoid the complexity of notations, the work is done for scalar-on-function regression. 
\par
The test and its methodology considered in this article can be extended for more than two independent samples. Here we are describing the case of three samples for notational simplicity. Suppose that in Model \eqref{eq:model}, consider $s\in \{1, 2, 3\}$, and we want to test 
\begin{center}
\noindent $H_{0}: \beta_{1}(t) = g_{1}(\beta_{2}(t))$ for some linear transformation $g_{1}$ and $\beta_{3}(t) = g_{2}(\beta_{2}(t))$ for some linear transformation $g_{2}$.

  against    

\noindent $H_{1} : \beta_{1}(t) = g_{1}^{*}(\beta_{2}(t))$ for some non-linear transformation $g_1^{*}$ and $\beta_{3}(t) = g_{2}^{*}(\beta_{2}(t))$ for some non-linear transformation $g_2^{*}$.
\end{center}

\noindent Recall the test statistic $T_{n}$ from Equation \eqref{stat} in the case of two independent sample, where $n = \min (n_1, n_2)$. Now, similarly define the test statistic $T_{n^{*}}$ for the second and the third independent samples. Finally, in order to test the aforesaid null hypothesis described in $H_{0}$, one can propose the test statistic as $T_{n_1, n_2, n_3} = \min(T_{n}, T_{n^{*}}).$ In fact, in principle, one may extend this methodology for any finitely many independent samples. Similarly, one may modify the bootstrap algorithm \ref{algo} as well for the case of three independent samples or finitely many independent samples.

\begin{appendix}
\section{Technical details}
In this section of the appendix, we provide all technical details related to Theorems \ref{thm:bias}, \ref{thm:emp}, \ref{thm:T}, and \ref{thm:bootstrap}.
\par
For any sequence of random variables $\{Z_{n}\}_{n\geq 1}$, we denote $Z_{n} = O_{r}(a_{n})$ if $\E|Z_{n}|^{r} = O(a_{n}^{r})$. Therefore, it is straightforward to see that $Z_{n} = \E\{Z_{n}\} + O_{r}([\E\{Z_{n} - \E Z_{n}\}^{r}]^{1/r})$.
Moreover, we define $Z_{n} = O_{\bbP}(a_{n})$ if for any positive $\epsilon$, $\bbP\{ |Z_{n}/a_{n}| > \epsilon\} \rightarrow 0$ as $n \rightarrow \infty$. 
 Due to Markov's inequality, $O_{r}(a_{n})$ implies $O_{\bbP}(a_{n})$ for any positive valued sequence $a_{n}\rightarrow 0$ (see, e.g, \citet{MR1385671}). Here we will state a few more lemmas, which will be useful in proving the main results of this work. 
\par
\begin{lemma}
\label{lemma:gen}
    For a sequence of uniformly bounded and continuous random functions $\{ X_{n}(t): t\in \sT\}_{n \geq 1}$ defined on a compact support, if $\int_{\sT} X^{2}(t)dt = O_{\bbP}(a_{n})$ then $\sup\limits_{t \in \sT} |X(t)| = O_{\bbP}(a_{n}^{1/2})$.
\end{lemma}
\begin{proof}
    Without loss of generality, assume that $\sT = [0, 1]$. Then, observe that $$\int X_{n}^{2}(t)dt \leq \{ \sup |X_{n}(t)|\}^{2}.$$ Afterwards, note that, for any $\epsilon > 0$ and $n \geq 1$, one has 
    $$\left\{ \omega: \left(\int X_{n}^{2}(t, \omega) dt/ a_{n} \right)^{1/2} > \epsilon \right\} \supseteq \left\{\omega: \sup |X_{n}(t, \omega)|/a_{n}^{1/2} > \epsilon \right\}.$$ Therefore, as $n \rightarrow \infty$, 
    $\bbP\left\{ \sup |X_{n}(t)|/a_{n}^{1/2} > \epsilon \right\} \leq \bbP\left\{ \left(\int X_{n}^{2}(t) dt/ a_{n} \right)^{1/2} > \epsilon \right\} \rightarrow 0$. 
\end{proof}
\begin{proof}[Proof of Lemma \ref{lemma:beta}]
    Using the derivation of Theorem 1 in \citet{hall2007methodology} and Lemma \ref{lemma:gen}, the result is immediate. 
\end{proof}
\begin{proof}[Proof of Lemma \ref{lemma:density}]
    Using the similar argument mentioned in Theorem 1 of \citet{imaizumi2018pca}, the result is immediate. 
\end{proof}

Before proving Theorems \ref{thm:bias}, \ref{thm:emp}, \ref{thm:T} and \ref{thm:bootstrap}, recall 
$\bU_{n_{1}} = (U_{1}, \ldots, U_{m_{n_{1}}})^{\tp}$, $\bV_{n_{2}} = (V_{1}, \ldots, V_{m_{n_{2}}})^{\tp}$ where $U_{j} = \widehat{\beta}_{1}(t_{j})$, estimate of $\beta_{1}(t_{j})$ based on $n_{1}$ samples and $V_{j} = \widehat{\beta}_{2}(t_{j})$, estimate of $\beta_{2}(t_{j})$ based on $n_{2}$ samples for $j =1, \ldots, m$. Moreover, further recall $\eta_{j} = \widehat{\beta}_{1}(t_{j}) - g(\widehat{\beta}_{2}(t_{j})) = U_{j} - g(V_{j})$, which depends on $n$ (see Section \ref{subsec:formulation}) is a random variable with mean zero and finite variance $\sigma_{n}^{2}$.

\begin{proof}[Proof of Theorem \ref{thm:bias}]
Observe that  
\begin{equation}
\label{eq:taylor}
    g(V_{j}) = g(v_{0}) + g'(v_{0})(V_{j} - v_{0}) + \frac{1}{2}g''(v_{0})(V_{j} - v_{0})^{2} + R_{m_{n}}(v_{0}, V_{j}), 
\end{equation}
where $$R_{m_{n}}(v_{0}, V_{j}) = g(V_{j}) - g(v_{0}) - g'(v_{0})(V_{j} - v_{0}) - \frac{1}{2}g''(v_{0})(V_{j} - v_{0})^{2}.$$ Let us now denote $\bR = (R(v_{0}, V_{1}), \ldots. R(v_{0}, V_{m_{n}}))^{\tp}$ and re-define 
\begin{align}\label{24}
    \sX(v_{0}) &= \begin{pmatrix}
             1 & V_{1} - v_{0} & (V_{1}-v_{0})^2 & (V_{1}-v_{0})^3\\
             \vdots & \vdots & \vdots & \vdots\\
             1 & V_{m_{n}} - v_{0} & (V_{m_{n}} - v_{0})^{2} & (V_{m_{n}}-v_{0})^3
             \end{pmatrix},
    \\
    \sX_{h_{m_{n}}}(v_{0}) &= \begin{pmatrix}
             1 & \frac{V_{1} - v_{0}}{h_{m_{n}}} & \frac{(V_{1}-v_{0})^2}{h_{m_{n}}^{2}} & \frac{(V_{1}-v_{0})^3}{h_{m_{n}}^{3}}\\
             \vdots & \vdots & \vdots & \vdots\\
             1 & \frac{V_{m_{n}} - v_{0}}{{m_{n}}} & \frac{V_{m_{n}} - v_{0})^{2}}{{m_{n}}^{2}} & \frac{(V_{m_{n}}-v_{0})^3}{{m_{n}}^{3}}
             \end{pmatrix},
    \\
    \mbox{and}~~~~\Gamma_{h_{m_{n}}} &= \begin{pmatrix}
        1 &  0 & 0 & 0\\
        0 & h_{m_{n}}^{-1} & 0 & 0\\
        0 & 0 & h_{m_{n}}^{-2} & 0 \\
        0 & 0 & 0 & h_{m_{n}}^{-3}
    \end{pmatrix}.
\end{align}
Thus, for a given point $v = v_{0} \in \bbR$, in the same spirit of Equations \eqref{eq:s} and \eqref{eq:g-dash}, observe that
\begingroup
\allowdisplaybreaks
\begin{align*}
    &\widehat{g}_{h_{n}}''(v_{0}) = \sum_{j=1}^{m}s_{j, h_{m_{n}}}(v_{0}, 2)\widehat{\beta}_{1}(t_{j})~\mbox{(see Equation \eqref{eq:g-dash})}\\
    &= 2\be_{3}^{\tp}\{ \sX(v_{0})^{\tp}\sW_{h_{m_{n}}}(v_{0})\sX(v_{0})\}^{-1}
    \{\sX(v_{0})^{\tp}\sW_{h_{m_{n}}}(v_{0})\bU_{n_{1}}\}~\mbox{(see Equation \eqref{eq:s})}\\
    &= 2\be_{3}^{\tp}\left\{
        \frac{1}{mh_{m_{n}}}\sum_{j=1}^{m_{n}}K\left( \frac{V_{j} - v_{0}}{h_{m_{n}}}\right)
        \left(1, 
            \frac{V_{j} - v_{0}}{h_{m_{n}}}, 
            \frac{(V_{j} - v_{0})^{2}}{h_{m_{n}}^{2}},
            \frac{(V_{j} - v_{0})^{3}}{h_{n}^{3}}
        \right)^{\tp\otimes^{2}}
    \right\}^{-1}\\
    &\qquad\times
    \Bigg\{
        \frac{1}{m_{n}h_{m_{n}}}\sum_{j=1}^{m_{n}}K\left( \frac{V_{j} - v_{0}}{h_{m_{n}}}\right)
        \left(1, \frac{V_{j} - v_{0}}{h_{m_{n}}}, 
        \frac{(V_{j} - v_{0})^{2}}{h_{m_{n}}^{2}}, 
        \frac{(V_{j} - v_{0})^{3}}{h_{m_{n}}^{3}}\right)^{\tp}\Bigg.\\
    &\qquad\qquad \Bigg.
        \times \left[g(v_{0}) + g'(v_{0})(V_{j} - v_{0}) + \frac{1}{2}g''(v_{0})(V_{j} - v_{0})^{2} + R_{m}(v_{0}, V_{j}) + \eta_{j}\right]
    \Bigg\}\\
    &= g''(v_{0}) + 2\be_{3}^{\tp}\left\{
        \frac{1}{mh_{m_{n}}}
        \sum_{j=1}^{m_{n}}K\left( \frac{V_{j} - v_{0}}{h_{m_{n}}}\right)
        \left(
        1, \frac{V_{j} - v_{0}}{h_{m_{n}}}, 
        \frac{(V_{j} - v_{0})^{2}}{h_{m_{n}}^{2}}, 
        \frac{(V_{j} - v_{0})^{3}}{h_{m_{n}}^{3}}\right)^{\tp\otimes^{2}}
    \right\}^{-1}\\
    &\qquad\times
    \left\{
        \frac{1}{mh_{n}}\sum_{j=1}^{m_{n}}K\left( \frac{V_{j} - v_{0}}{h_{m_{n}}}\right)
        \left(1, 
        \frac{V_{j} - v_{0}}{h_{m_{n}}}, 
        \frac{(V_{j} - v_{0})^{2}}{h_{m_{n}}^{2}}, 
        \frac{(V_{j} - v_{0})^{3}}{h_{m_{n}}^{3}}\right)^{\tp}
        \left[R_{m_{n}}(v_{0}, V_{j}) + \eta_{j}\right]
    \right\}\\
    &= g''(v_{0}) + 2\be_{3}^{\tp}\bA_{1}^{-1}(v_{0})\left\{\bA_{2}(v_{0}) +  \bA_{3}(v_{0})\right\},
    \numberthis
    \label{eq:main-g2}
\end{align*}
\endgroup
where
\begingroup
\allowdisplaybreaks
\begin{align*}
    \bA_{1}(v_{0}) &= 
        \frac{1}{m_{n}h_{m_{n}}}\sum_{j=1}^{m_{n}}K\left( \frac{V_{j} - v_{0}}{h_{m_{n}}}\right)
        \left(1, 
        \frac{(V_{j} - v_{0})}{h_{m_{n}}}, 
        \frac{(V_{j} - v_{0})^{2}}{h_{m_{n}}^{2}}, 
        \frac{(V_{j} - v_{0})^{3}}{h_{m_{n}}^{3}}\right)^{\tp\otimes^{2}},
    \\
    \bA_{2}(v_{0}) &= 
        \frac{1}{m_{n}h_{m_{n}}}
        \sum_{j=1}^{m_{n}}
        R_{m_{n}}(v_{0}, V_{j})
        K\left( \frac{V_{j} - v_{0}}{h_{m_{n}}}\right)
        \left(1, 
        \frac{(V_{j} - v_{0})}{h_{m_{n}}}, 
        \frac{(V_{j} - v_{0})^{2}}{h_{m_{n}}^{2}}, 
        \frac{(V_{j} - v_{0})^{3}}{h_{m_{n}}^{3}}\right)^{\tp},
    \\    
    \mbox{and}~\bA_{3}(v_{0}) &= 
        \frac{1}{m_{n}h_{m_{n}}}
        \sum_{j=1}^{m}
        \eta_{j}
        K\left( \frac{V_{j} - v_{0}}{h_{m_{n}}}\right)
        \left(1, 
        \frac{(V_{j} - v_{0})}{h_{m_{n}}}, 
        \frac{(V_{j} - v_{0})^{2}}{h_{m_{n}}^{2}}, 
        \frac{(V_{j} - v_{0})^{3}}{h_{m_{n}}^{3}}\right)^{\tp}.
\end{align*}
\endgroup
\par
First, we aim to control the $\bA_{1}(v_{0})$. Under condition \ref{cond:kernel} and in view of the assertion in Lemma \ref{lemma:density}, using Taylor's series expansion, for any integer $c \in \{ 0, \ldots, 6\}$, we have 
\begingroup
\allowdisplaybreaks
\begin{align*}
\label{eq:exp}
    &\E\left\{ \frac{1}{m_{n}h_{m_{n}}}
    \sum_{j=1}^{m_{n}}
    \left(\frac{V_{j} - v_{0}}{h_{m_{n}}}\right)^{c}
    K\left( \frac{V_{j} - v_{0}}{h_{m_{n}}}\right)\right\}\\
    &= \int K_{h_{m_{n}}}(v - v_{0})
    \left(\frac{v-s_{0}}{h_{m_{n}}}\right)^{c}f_{\widehat{\beta}_{2}}(v)dv\\
    &= \int v^{c}K(v)f_{\widehat{\beta}_{2}}(v_{0}+h_{m_{n}}v)dv\\
    &=\int v^{c}K(v)
    \left\{ f_{\widehat{\beta}_{2}}(v_{0}) + 
            h_{m_{n}}v f_{\widehat{\beta}_{2}}^{[1]}(v_{0}) + 
            \frac{1}{2}h_{m_{n}}^{2}v^{2} f_{\widehat{\beta}_{2}}^{[2]}(v_{0}^{*})\right\}dv, \mbox{ where~$v_{0}^{*}\in (v, v_{0})$}\\
    & = \begin{cases}
        O\left(h_{m_{n}}f_{\widehat{\beta}_{2}}^{[1]}(v_{0})\right), & \text{for odd $c$ with $\nu_{c+1, 1} < \infty$}\\
        & \text{and $f_{\widehat{\beta}_{2}}^{[1]}(v_{0}) < \infty$}\\
        \nu_{c,1}f_{\widehat{\beta}_{2}}(v_{0}) + 
        O\left(h_{m_{n}}f_{\widehat{\beta}_{2}}^{[2]}(v_{0})\right), & \text{for even $c$ is with $\nu_{c+1, 1} < \infty$}\\
        & \text{and $f_{\widehat{\beta}_{2}}^{[2]}(v_{0}) < \infty$}.\\
    \end{cases}
  \numberthis
\end{align*}
\endgroup
Here $v_{0}$ is an arbitrary point $v_{0}$ in the support of $f_{\widehat{\beta}_{s}} (.)$, where $f_{\widehat{\beta}_{s}}^{[j]}(.)$ denotes the $j$-th derivatives of $f_{\widehat{\beta}_{s}}(.)$. 

Next, arguing in a similar way, one can derive 
\begingroup
\allowdisplaybreaks
 \begin{align*}
 \label{eq:expectation}
    &\E\left\{\frac{1}{m_{n}^{2}h_{m_{n}}^{2}}\sum_{j=1}^{m_{n}}\left(\frac{V_{j}-v_{0}}{h_{m_{n}}}\right)^{c}K^{2}\left(\frac{V_{j}-v_{0}}{h_{m_{n}}}\right) \right\}\\
    &=\frac{1}{m_{n}h_{m_{n}}^{2}}
    \E\left\{
        \left(\frac{V-v_{0}}{h_{m_{n}}}\right)^{c}K^{2}\left(\frac{V-v_{0}}{h_{m_{n}}}\right)
    \right\}\\
    &= \frac{1}{m_{n}h_{m_{n}}}
    \int K^{2}_{h_{n}}(v - v_{0})
    \left(\frac{v-v_{0}}{h_{m_{n}}}\right)^{c}f_{\widehat{\beta}_{2}}(v)dv\\
     &= \frac{1}{m_{n}h_{m_{n}}}\int v^{c}K^{2}(v)\left\{ f_{\widehat{\beta}_{2}}(v_{0}) + h_{m_{n}}v f_{\widehat{\beta}_{2}}^{[1]}(v_{0}) + 0.5h_{m_{n}}^{2}v^{2} f_{\widehat{\beta}_{2}}^{[2]}(v_{0}^{**}) \right\}dv,~\mbox{where $v_{0}^{**}\in (v, v_{0})$}\\
    & = \frac{1}{m_{n}h_{m_{n}}}\begin{cases}
        O\left(h_{m_{n}}f_{\widehat{\beta}_{2}}^{[1]}(v_{0})\right), & \text{for odd $c$ with $\nu_{c+1, 2} < \infty$}\\
        & \text{and $f_{\widehat{\beta}_{2}}^{[1]}(v_{0}) < \infty$}\\
        \nu_{c,2}f_{\widehat{\beta}_{2}}(v_{0}) + 
        O\left(h_{m_{n}}^{2}f_{\widehat{\beta}_{2}}^{[2]}(v_{0})\right), & \text{for even $c$ with $\nu_{c+1, 2} < \infty$}\\
        & \text{and $f_{\widehat{\beta}_{2}}^{[2]}(v_{0})<\infty$}\\
    \end{cases}.
    \numberthis
\end{align*}
\endgroup
\begingroup
\allowdisplaybreaks
Therefore, for $c \in \{ 0, \ldots, 6\}$, using condition \ref{cond:kernel} and Lemma \ref{lemma:density}, we have
\begin{align*}
    &\Var\left\{
        \frac{1}{m_{n}h_{m_{n}}}
        \sum_{j=1}^{m_{n}}\left(\frac{V_{j} - v_{0}}{h_{m_{n}}}\right)^{c}
        K\left( \frac{V_{j} - v_{0}}{h_{m_{n}}}\right)
    \right\}\\
    &\lesssim \E
    \left\{
        \frac{1}{m_{n}^{2}h_{m_{n}}^{2}}
        \sum_{j=1}^{m_{n}}\left(\frac{V_{j} - v_{0}}{h_{m_{n}}}\right)^{2c}
        K^{2}\left( \frac{V_{j} - v_{0}}{h_{m_{n}}}\right)
    \right\} = O\left(\frac{f_{\widehat{\beta}_{2}}(v_{0})}{m_{n}h_{m_{n}}}\right).
    \numberthis
\end{align*}
\endgroup
Besides, Lemma \ref{lemma:density} asserts that, conditioning on $V_{1}, \ldots, V_{m_{n}}$, 
\begin{align*}
    \bA_{1}(v_{0}) = f_{\widehat{\beta}_{2}}(v_{0})
    \underbrace{
    \begin{pmatrix}
        1 & 0 & \nu_{2,1} & 0\\
        0 & \nu_{2,1} & 0 & \nu_{4,1}\\
        \nu_{2,1} & 0 & \nu_{4,1} & 0\\
        0 & \nu_{4,1} & 0 & \nu_{6,1}
    \end{pmatrix}}_{=\bOmega}
    + O\left(h_{m_{n}}f_{\widehat{\beta}_{2}}^{[1]}(v_{0})\right) + O_{2}\left(\sqrt{\frac{f_{\widehat{\beta}_{2}}(v_{0})}{m_{n}h_{m_{n}}}}\right).
    \numberthis
\end{align*}
and a few steps of straightforward algebra gives us 
\begin{align*}
    \be_{3}^{\tp}\bA_{1}^{-1}(v_{0}) = 
    \frac{1}{f_{\widehat{\beta}_{2}}(v_{0})}(a_{1}, 0, a_{2}, 0)^{\tp} + O\left(h_{m_{n}}f_{\widehat{\beta}_{2}}^{[1]}(v_{0})\right) + O_{2}\left(\sqrt{\frac{f_{\widehat{\beta}_{2}}(v_{0})}{m_{n}h_{m_{n}}}}\right),
    \numberthis
\end{align*}
where $$a_{1} = \frac{\nu_{2,1}\nu_{4,1}^{2} - \nu_{2,1}^{2}\nu_{6,1}}{\nu_{2,1}^{2}\nu_{4,1}^{2} - \nu_{4,1}^{2} - \nu_{2,1}^{3}\nu_{6,1} + \nu_{2,1}\nu_{4,1}\nu_{6,1}},$$ and $$a_{2} = \frac{-\nu_{4,1}^{2} + \nu_{2,1}\nu_{6,1}}{\nu_{2,1}^{2}\nu_{4,1}^{2} - \nu_{4,1}^{3} - \nu_{2,1}^{3}\nu_{6,1} + \nu_{2,1}\nu_{4,1}\nu_{6,1}}.$$ 
\par
Afterwards, for controlling the second term of Equation \eqref{eq:main-g2}, viz., $\bA_{2}(v_{0})$, we condition on $V_{1}, \ldots, V_{m_{n}}$ and  for $c \in \{ 0, 1, 2, 3\}$, we obtain, 
\begingroup
\allowdisplaybreaks
\begin{align*}
    &\frac{1}{m_{n}h_{m_{n}}}
    \sum_{j=1}^{m_{n}}R_{m_{n}}(v_{0}, V_{j})
    \left(\frac{V_{j} - v_{0}}{h_{m_{n}}}\right)^{c}
    K\left(\frac{V_{j} - v_{0}}{h_{m_{n}}} \right)\\
    &=\frac{1}{h_{m_{n}}}\int R_{m_{n}}(v_{0}, v)
    \left(\frac{v - v_{0}}{h_{m_{n}}}\right)^{c}
    K\left(\frac{v - v_{0}}{h_{m_{n}}}\right)f_{\widehat{\beta}_{2}}(v)dv + 
    O_{2}\left(\sqrt{\frac{f_{\widehat{\beta}_{2}}(v_{0})}{m_{n}h_{m_{n}}}}\right)\\
    &=\int R_{m_{n}}(v_{0}, v_{0} + h_{m_{n}}v)
    v^{c}K(v)f_{\widehat{\beta}_{2}}(v_{0} + h_{m_{n}}v)dv +  
    O_{2}\left(\sqrt{\frac{f_{\widehat{\beta}_{2}}(v_{0})}{m_{n}h_{m_{n}}}}\right)\\
    &\overset{(i)}{=} \int \frac{1}{2}
    \left(g''(v_{\star}) - g''(v_{0})\right)
    v^{2+c}K(u)
    f_{\widehat{\beta}_{2}}(v_{0} + h_{m_{n}}v)dv + 
    O_{2}\left(\sqrt{\frac{f_{\widehat{\beta}_{2}}(v_{0})}{m_{n}h_{m_{n}}}}\right)\\
    &=O\left(f_{\widehat{\beta}_{2}}(v_{0})h_{m_{n}}^{\delta+2}\right) + 
    O_{2}\left(\sqrt{\frac{f_{\widehat{\beta}_{2}}(v_{0})}{m_{n}h_{m_{n}}}}\right).
    \numberthis
\end{align*}
\endgroup
The equality holds in $(i)$ due to mean value theorem, where $|v_{\star} - v_{0}| \leq h_{m_{n}}v$. Along with the similar argument, he third term in Equation \eqref{eq:main-g2}, viz., $\bA_{3}(v)\rightarrow 0$ in probability, as $n\rightarrow\infty$ in view of the fact that $\E\{\eta\} = 0$ and Var$(\bA_{3}(v))\rightarrow 0$ as $n\rightarrow\infty$. Hence, by condition \ref{cond:kernel} and Lemma \ref{lemma:density}, it follows from \eqref{eq:main-g2} that  
\begin{align*}
    &\E\left\{ \widehat{g}_{h_{n}}''(v_{0})\right\} - g''(v_{0})\\
    &= \E\left\{ 
        2\be_{3}^{\tp}\left[
            \frac{\bA_{1}^{-1}(v_{0})}{f_{\widehat{\beta}_{2}}(v_{0})}
            + O(h_{n}) + O_{2}(1/\sqrt{mh_{n}})
        \right]
        \bA_{2}(v_{0})
    \right\}\\
    &O\left(\left\{\frac{1}{f_{\widehat{\beta}_{2} (v_{0})}} + h_{m_{n}}f_{\widehat{\beta}_{2}}^{[1]}(v_{0}) + \sqrt{\frac{f_{\widehat{\beta}_{2}}(v_{0})}{mh_{m_{n}}}}\right\}
    \left\{f_{\widehat{\beta}_{2}}(v_{0})h_{m_{n}}^{\delta+2} + \sqrt{\frac{f_{\widehat{\beta}_{2}}(v_{0})}{mh_{m_{n}}}}\right\}\right). \numberthis\\
\end{align*}
Next, we consider 
\begingroup
\allowdisplaybreaks
\begin{align*}
    &\E\left\{\sX^{\tp}(v_{0})\sW_{h_{m_{n}}}^{2}(v_{0})\sX(v_{0})\right\}\\
    &=\frac{1}{m_{n}^{2}h_{m_{n}}^{2}}
    \E\left\{
        \sum_{j=1}^{m_{n}}K^{2}_{h_{m_{n}}}
        \left(\frac{V_{j} - v_{0}}{h_{m_{n}}}\right)
        \left(1, 
        \frac{V_{j} - v_{0}}{h_{m_{n}}}, 
        \frac{(V_{j} - v_{0})^{2}}{h_{m_{n}}^{2}},
        \frac{(V_{j} - v_{0})^{3}}{h_{m_{n}}^{3}}\right)^{\tp \otimes^{2}}
    \right\}\\
    &= \frac{1}{m_{n}h_{m_{n}}}
    \left\{ 
    f_{\widehat{\beta}_{2}}(v_{0})
    \begin{pmatrix}
        \nu_{2,0} & 0 & \nu_{2,2} & 0\\
        0 & \nu_{2,2} & 0 & \nu_{4,2}\\
        \nu_{2,2} & 0 & \nu_{4,2} & 0\\
        0 & \nu_{4,2} & 0 & \nu_{6,2}
    \end{pmatrix} + O\left(h_{m_{n}}f_{\widehat{\beta}_{2}}(v_{0})\right)\right\} .
    \numberthis
\end{align*}
\endgroup
The last equality in the above expression holds due to the similar derivation of the expression \eqref{eq:expectation}. Hence, we have the following using condition \eqref{cond:eta}. 
\begin{align*}
    &\Var\left\{ \widehat{g}_{h_{m_{n}}}''(v_{0})\right\}\\
    &\leq4\E\left\{
    \be_{3}^{\tp}\bA_{1}^{-1}(v_{0})
    \sX^{\tp}(v_{0})\sW_{h_{m_{n}}}(v_{0})\bU_{n_{1}}\bU_{n_{1}}^{\tp}\sW_{h_{m_{n}}}(v_{0})\sX(v_{0})
    \bA_{1}^{-1}(v_{0})\be_{3}
    \right\}\\
    &\lesssim
        \sigma_{n}^{2}\be_{3}^{\tp}\bA_{1}(v_{0})^{-1}
        \E\left\{
            \sX_{h_{m_{n}}}(v_{0})^{\tp}\sW_{h_{m_{n}}}(v_{0})\sW_{h_{m_{n}}}(v_{0})\sX_{h_{m_{n}}}(v_{0})
        \right\}
        \bA_{1}(v_{0})^{-1}\be_{3}\\
    &= O\left(\frac{\sigma^{2}_{n}}{m_{n}h_{m_{n}}^{5}f_{\widehat{\beta}_{2}}(v_{0}))}\right).    
    \numberthis    
\end{align*} The proof is complete. 
\end{proof}

\begin{lemma}
    \label{lemma:random-bound}
    Let $\{Z_{n}(v)\}_{n \geq 1}$ be a sequence of random functions defined on $[v_{1}, v_{2}] \subset \bbR$, such that $\sup\limits_{v \in [v_{1}, v_{2}]}|Z_{n}(v)| = O_{\bbP}(a_{n})$, where $\{a_{n}\}_{n\geq 1}$ is a sequence of positive real numbers with $\displaystyle\lim_{n\rightarrow\infty} a_{n} = 0$. Suppose that for any $\delta > 0$, there exits $M (\delta) > 0$ such that $\bbP [\widehat{v}_{1} < v_{1} - M] < \delta$ and $\bbP [\widehat{v}_{2} > v_{2} + M] < \delta$, where $\widehat{v}_{1}$ and $\widehat{v}_{2}$ are estimators of $v_1$ and $v_2$, respectively. Then, $\sup\limits_{v \in [\widehat{v}_{1}, \widehat{v}_{2}]} |Z_{n}(v)| = O_{\bbP}(a_{n})$.
\end{lemma}

\begin{proof}
   Observe that for a fixed $v$, $Z_{n}(v)$ is a sequence of random variables. Now, for any $\epsilon > 0$, there exits a constant $C (\epsilon) > 0$ and $N (\epsilon) \in \bbN$ such that for all $n \geq N$, 
    \begin{equation}
    \label{eq:random1}
        \bbP\left\{\sup_{v \in [v_{1}, v_{2}]}|Z_{n}(v)| > C a_{n}\right\} < \epsilon.
    \end{equation}
    Further, observe that $$\bbP\{ \widehat{v}_{1} \in [v_{1} - M, v_{2} + M] \cap \widehat{v}_{2} \in [v_{1} - M, v_{2} + M]
    \} \geq (1 - 2\delta)$$ in view of the assertions on $\widehat{v}_{l}$ and $v_{l}$ stated in the statement of this lemma, where $l \in \{1, 2\}$. In addition, we also have 
    \begin{equation}
    \label{eq:random2}
    \sup_{v \in [\widehat{v}_{1}, \widehat{v}_{2}]}|Z_{n}(v)|\leq \sup_{v \in [{v}_{1}-M, {v}_{2}+M]}|Z_{n}(v)|  
    \end{equation}
    Moreover, in view of the fact that  $\sup\limits_{v \in [v_{1}-M, v_{2} + M]}|Z_{n}(v)| = O_{\bbP}(a_{n})$, one can conclude that for every $\epsilon > 0$, there exists a generic constant $C > 0$ and $N(\epsilon) \in \bbN$ such that for $n \geq N$, 
    \begin{equation}
    \label{eq:random3}
        \bbP\left\{\sup_{v \in [{v}_{1}-M, {v}_{2}+M]}|Z_{n}(v)| > Ca_{n}\right\} > \epsilon - 2\delta.
    \end{equation}
    Hence, by combining inequalities in \eqref{eq:random2} and \eqref{eq:random3}, 
    \begin{align*}
        & \bbP\left\{ \sup\limits_{v \in [\widehat{v}_{1}, \widehat{v}_{2}]} |Z_{n}(v)| > C a_{n}\right\}\\ 
        & \leq \bbP\left\{\sup_{v \in [{v}_{1}-M, {v}_{2}+M]}|Z_{n}(v)| > Ca_{n}\right\}\\
        & \qquad + 
        \bbP\left\{\widehat{v}_{1} \notin [v_{1} - M, v_{2} + M] \cup \widehat{v}_{2} \notin [v_{1} - M, v_{2} + M] \right\}\\
        & < (\epsilon - 2\delta) + 2\delta = \epsilon. \numberthis
    \end{align*}
    This completes the proof.
\end{proof}


\begin{proof}[Proof of Theorem \ref{thm:emp}]
First, we need to need to show that, under \ref{cond:kernel}, \ref{cond:VC}, \ref{cond:g} \ref{cond:band} and the assertion in Lemma \ref{lemma:density}, for some real constants $v_{1} < v_{2}$, the scaled difference $d_{h_{m_{n}}}(v)= \sqrt{m_{n}h_{m_{n}}^{5}}\left(\widehat{g}_{h_{m_{n}}}''(v) - \E\left\{\widehat{g}_{h_{m_{n}}}''(v)\right\}\right)$ has the following approximation:
\begin{equation}
\label{eq:d}
    \sup_{v \in [v_{1}, v_{2}]}\left| 
        \frac{
            d_{m_{n}} - 
            \frac{1}{\sqrt{h_{m_{n}}}}
            \bbG_{m_{n}}[\be^{\tp}\bOmega^{-1}\bPsi_{m_{n}}(v)]
        }{
            \frac{1}{\sqrt{h_{m_{n}}}}
            \bbG_{m_{n}}[\be_{3}^{\tp}\bPsi_{h_{m_{n}}}(v)]
        }
    \right| = 
    O(h_{m_{n}}) + O_{\bbP}\left(
        \sqrt{\frac{\log m_{n}}{m_{n}h_{m_{n}}}}
    \right), 
\end{equation} 
where $\bbG_{m_{n}}$, $\bOmega$ and $\bPsi_{m_{n}}$ are the same as defined in Section \ref{sec:theory}. 
\par
In order to establish \eqref{eq:d}, note that the scaled difference can be expressed as 
\begin{align*}
\label{eq:decom}
    d_{h_{m_{n}}}(v) &= \sqrt{m_{n}h_{m_{n}}^{5}}
    \left(
    \widehat{g}_{h_{m_{n}}}''(v) - \E\left\{\widehat{g}_{h_{m_{n}}}''(v)\right\}
    \right)\\
    &= \sqrt{m_{n}h_{m_{n}}^{5}}\left(\widehat{g}_{h_{m_{n}}}''(v) - g''(v)\right) + 
    \sqrt{m_{n}h_{m_{n}}^{5}}\left(g''(v) - \E\left\{ \widehat{g}_{h_{m_{n}}}''(v)\right\}\right).
    \numberthis
\end{align*}
Using the fact in Theorem \ref{thm:bias}, observe that 
$$\sqrt{m_{n}h_{m_{n}}^{5}}\left(g''(v) - \E\left\{ \widehat{g}_{h_{m_{n}}}''(v)\right\}\right) = o_{p}(1)$$ as $m_{n}h_{m_{n}}^{1/5} \rightarrow c$, for some constant $c$. 
Therefore, one can conclude that $$d_{h_{m_{n}}}(v) = \sqrt{m_{n}h_{m_{n}}^{5}}\left(\widehat{g}_{h_{m_{n}}}''(v) - g''(v)\right) +  o_{\bbP}(1).$$ Now, one can simplify the expression of $\widehat{g}_{h_{m_{n}}}''(v)$ as follow. 
\begin{align*}
\label{eq:g1}
    &\widehat{g}_{h_{m_{n}}}''(v) 
    = 2\be_{3}^{\tp}\left\{\sX^{\tp}(v)\sW_{h_{m_{n}}}(v)\sX(v)\right\}^{-1}
    \sX^{\tp}(v)\sW_{h_{m_{n}}}\bU_{n_{1}}\mbox{(see \eqref{eq:g-dash})}\\
    &=\frac{2\be_{3}^{\tp}\Gamma_{h_{m_{n}}}}{h_{m_{n}}^{2}}
    \left\{\frac{1}{m_{n}h_{n}}\Gamma_{h_{m_{n}}}\sX^{\tp}(v)\sW_{h_{m_{n}}}(v)\sX(v)\Gamma_{h_{m_{n}}}\right\}^{-1}
    \frac{1}{h_{m_{n}}}\Gamma_{h_{m_{n}}}\sX^{\tp}(v)\sW_{h_{m_{n}}}\bU_{n_{1}}\\
    &=\frac{2\be_{3}^{\tp}}{h_{m_{n}}^{2}}
    \left\{\frac{1}{m_{n}h_{m_{n}}}\sX_{h_{m_{n}}}^{\tp}(v)\sW_{h_{m_{n}}}(v)\sX_{h_{m_{n}}}(v)\right\}^{-1}
    \frac{1}{m_{n}h_{m_{n}}}\sX_{h_{m_{n}}}^{-1}(v)\sW_{h_{m_{n}}}(v)\bU_{n_{1}}\\
    &=\frac{2\be_{3}^{\tp}}{h_{m_{n}}^{2}}
    \left(
        \frac{\bOmega^{-1}}{f_{\widehat{\beta}_{2}}(v)} + O(h_{m_{n}}) + O_{\bbP}\left(\sqrt{\frac{\log m_{n}}{m_{n}h_{m_{n}}}}
        \right)
    \right)
    \frac{1}{m_{n}h_{m_{n}}}\sX_{h_{m_{n}}}^{\tp}(v)\sW_{h_{m_{n}}}(v)\bU_{n_{1}}\mbox{(using \eqref{eq:d})},
    \numberthis
\end{align*}
where $\bU_{n_{1}} = (U_{1}, \ldots, U_{m_{n_{1}}})^{\tp}$.
\par
\vspace{0.1in}
\noindent Further, we also have
\begingroup
\allowdisplaybreaks
\begin{align*}
    \frac{1}{m_{n}h_{m_{n}}}\sX_{h_{m_{n}}}^{\tp}(v)\sW_{h_{m_{n}}}(v)\bU_{n_{1}}
    &=\begin{pmatrix}
        \frac{1}{m_{n}h_{m_{n}}}\sum\limits_{i=1}^{m_{n}}U_{i}K\left(\frac{V_{i} - v}{h_{m_{n}}}\right)\\
        \frac{1}{m_{n}h_{m_{n}}}\sum\limits_{i=1}^{m_{n}}U_{i}\left(\frac{V_{i} - v}{h_{m_{n}}}\right)K\left(\frac{V_{i} - v}{h_{m_{n}}}\right)\\
        \frac{1}{m_{n}h_{m_{n}}}\sum\limits_{i=1}^{m_{n}}U_{i}\left(\frac{V_{i} - v}{h_{m_{n}}}\right)^{2}K\left(\frac{V_{i} - v}{h_{m_{n}}}\right)\\
        \frac{1}{m_{n}h_{m_{n}}}\sum\limits_{i=1}^{m_{n}}U_{i}\left(\frac{V_{i} - v}{h_{m_{n}}}\right)^{3}K\left(\frac{V_{i} - v}{h_{m_{n}}}\right)
    \end{pmatrix}\mbox{(using \eqref{eq:s} and \eqref{24})}\\
    &=\frac{1}{h_{m_{n}}}\int \bPsi(v; Z_{1}, Z_{2})d\bbP_{m_{n}}(Z_{1}, Z_{2}), 
    \numberthis
\end{align*}
\endgroup
where $$\bPsi_{h_{m_{n}}}(v; z_{1}, z_{2}) = (\psi_{0, h_{m_{n}}}(v; z_{1}, z_{2}), \psi_{1, h_{m_{n}}}(v; z_{1}, z_{2}), \psi_{2, h_{m_{n}}}(v; z_{1}, z_{2}), \psi_{3, h_{m_{n}}}(v; z_{1}, z_{2}))^{\tp}$$ as defined in Section \ref{sec:theory}. Further, observe that each element of the vector $\bPsi_{h_{m_{n}}}$ is defined as $$\psi_{c, h_{m_{n}}}(v, z_{1}, z_{2}) = z_{2}\left(\frac{z_{1} - v}{h_{m_{n}}}\right)^{c}K\left(\frac{z_{1} - v}{h_{m_{n}}}\right)$$ for $c \in \{ 0, 1, 2, 3\}$. Therefore, using \eqref{eq:g1}, one has 
\begin{align*}
\label{eq:gdoubledash}
    \widehat{g}_{h_{m_{n}}}''(v)
    &= \frac{2}{h_{m_{n}}^{3}}
    \int \frac{\be_{3}^{\tp}\bOmega^{-1}\bPsi(v, Z_{1}, Z_{2})}{f_{\widehat{\beta}_{2}}(v)}d\bbP_{m_{n}}(Z_{1}, Z_{2})\\
    & \qquad + 
    \frac{2}{h_{m_{n}}^{3}}
    \int \be_{3}^{\tp}\bPsi_{h_{m_{n}}}(v)(v; Z_{1}, Z_{2})d\bbP_{m_{n}}(Z_{1}, Z_{2})
    \left\{ O\left(h_{m_{n}}\right) + O_{\bbP}\left(\sqrt{\frac{\log m_{n}}{m_{n}h_{m_{n}}}}\right)\right\}\\
    &= \frac{1}{h_{m_{n}}^{3}}\bbP_{m}[\be_{3}^{\tp}\bA_{1}^{-1}\bPsi_{h_{m_{n}}}(v)]\\
    & \qquad + \frac{1}{h_{m_{n}}^{3}}\bbP_{m_{n}}[\be_{3}^{\tp}\bPsi_{h_{m_{n}}}(v)]\left\{ O\left(h_{m_{n}}\right) + O_{\bbP}\left(\sqrt{\frac{\log m_{n}}{m_{n}h_{m_{n}}}}\right)\right\},
    \numberthis
\end{align*}
where $$\bbP_{m_{n}}[\be_{3}^{\tp}\bOmega^{-1}\bPsi_{h_{m_{n}}}(v)] = \int 2\be_{3}^{\tp}\bOmega^{-1}\bPsi_{h_{m_{n}}}(v; Z_{1}, Z_{2})d\bbP_{m_{n}}(Z_{1}, Z_{2})$$ and $$\bbP_{m_{n}}[\be_{3}^{\tp}\bPsi_{h_{m_{n}}}(v)] = \int 2\be_{3}^{\tp}\bPsi_{h_{m_{n}}}(v; Z_{1}, Z_{2})d\bbP_{m_{n}}(Z_{1}, Z_{2}).$$
Hence, we have 
\begin{align*}
\label{eq:Egdoubledash}
    \E\left\{\widehat{g}_{h_{m_{n}}}''(v)\right\} &= 
    \frac{1}{h_{m_{n}}^{3}}\bbP[\be_{3}^{\tp}\bOmega\bPsi_{h_{m_{n}}}(v)]\\
    & \qquad + \frac{1}{h_{m_{n}}^{3}}\bbP[\be_{3}^{\tp}\bPsi_{h_{m_{n}}}(v)]\left\{ O\left(h_{m_{n}}\right) + O\left(\sqrt{\frac{\log m_{n}}{m_{n}h_{m_{n}}}}\right)\right\},  
    \numberthis
\end{align*}
where $$\bbP[\be_{3}^{\tp}\bOmega^{-1}\bPsi_{h_{m_{n}}}(v)] = \E\left\{\bbP_{m_{n}}[\be_{3}^{\tp}\bA_{1}^{-1}\bPsi_{h_{m_{n}}}(v)]\right\}$$ and $$\bbP[\be_{3}^{\tp}\bPsi_{h_{m_{n}}}(v)] = \E\left\{\bbP_{m_{n}}[\be_{3}^{\tp}\bPsi_{h_{m_{n}}}(v)]\right\}.$$ Summarizing the facts in \eqref{eq:gdoubledash} and \eqref{eq:Egdoubledash}, we have 
\begingroup
\allowdisplaybreaks
\begin{align*}
\label{eq:d-exp}
    &d_{m_{n}}(v) = \sqrt{m_{n}h_{m_{n}}^{5}}\left(\widehat{g}_{h_{m_{n}}}''(v) - \E\left\{\widehat{g}''_{h_{m_{n}}}(v)\right\}\right)\\
    &= \sqrt{\frac{{m_{n}}}{h_{m_{n}}}}
    \left(\bbP_{m_{n}} - \bbP\right)[\be^{\tp}\bOmega^{-1}\bPsi_{h_{m_{n}}}(v)]\\
    & \qquad + \sqrt{\frac{{m_{n}}}{h_{m_{n}}}}
    \left(\bbP_{m_{n}} - \bbP\right)[\be_{3}^{\tp}\bPsi_{h_{m_{n}}}(v)]\left\{ O\left(h_{m_{n}}\right) + O_{\bbP}\left(\sqrt{\frac{\log m_{n}}{m_{n}h_{m_{n}}}}\right)\right\}\\
    &= \frac{1}{\sqrt{h_{m_{n}}}}\bbG_{m_{n}}[\be^{\tp}\bOmega^{-1}\bPsi_{h_{m_{n}}}(v)]~\mbox{(using...)}\\
    & \qquad + \frac{1}{\sqrt{h_{m_{n}}}}\bbG_{m_{n}}[\be_{3}^{\tp}\bPsi_{h_{m_{n}}}(v)]\left\{ O\left(h_{m_{n}}\right) + O_{\bbP}\left(\sqrt{\frac{\log m_{n}}{m_{n}h_{m_{n}}}}\right)\right\}\\
    \numberthis
\end{align*}
\endgroup
Moreover, the equality in \eqref{eq:d-exp} is valied for all $v$, and hence, we have 
\begin{align*}
    &\sup_{v \in [v_{1}, v_{2}]}\left| 
        \frac{
        \sqrt{m_{n}h_{m_{n}}^{5}}\left(\widehat{g}_{h_{m_{n}}}''(v)- 
            \E\{\widehat{g}_{h_{m_{n}}}(v)\}\right) - 
            \frac{1}{\sqrt{h_{m_{n}}}}\bbG_{m_{n}}[\be_{3}^{\tp}\bOmega^{-1}\bPsi_{h_{m_{n}}}(v)]
        }{
            \frac{1}{\sqrt{h_{m_{n}}}}\bbG_{h_{m_{n}}}[\be_{3}^{\tp}\bPsi_{h_{m_{n}}}(v)]
        }
    \right|\\
    &= O(h_{m_{n}}) + O_{\bbP}\left( 
        \sqrt{
            \frac{\log m_{n}}{m_{n}h_{m_{n}}}
        }
    \right). 
    \numberthis
\end{align*}
Finally, using conditions \ref{cond:kernel}, \ref{cond:VC}, \ref{cond:g} \ref{cond:band} and the assertion in Lemma \ref{lemma:density} and the fact in Lemmas \ref{lemma:beta} and \ref{lemma:random-bound}, we have 
\begin{align*}
    &\sup_{v \in [\widehat{v}_{1}, \widehat{v}_{2}]}\left| 
        \frac{
        \sqrt{m_{n}h_{m_{n}}^{5}}\left(\widehat{g}_{h_{m_{n}}}''(v)- 
            \E\{\widehat{g}_{h_{m_{n}}}(v)\}\right) - 
            \frac{1}{\sqrt{h_{m_{n}}}}\bbG_{m_{n}}[\be_{3}^{\tp}\bOmega^{-1}\bPsi_{h_{m_{n}}}(v)]
        }{
            \frac{1}{\sqrt{h_{m_{n}}}}\bbG_{h_{m_{n}}}[\be_{3}^{\tp}\bPsi_{h_{m_{n}}}(v)]
        }
    \right|\\
    &= O(h_{m_{n}}) + O_{\bbP}\left( 
        \sqrt{
            \frac{\log m_{n}}{m_{n}h_{m_{n}}}
        }
    \right), 
    \numberthis
\end{align*}
which proves the first part of this theorem. 
\par
\vspace{0.1in}
Now, recall the derivation in \eqref{eq:d-exp} and observe that 
\begin{equation}\label{47}
    \sqrt{m_{n}h_{m_{n}}^{5}}\left( \widehat{g}_{h_{m_{n}}}''(v) - \E\{\widehat{g}_{h_{m_{n}}}''(v)\}\right) = \frac{1}{\sqrt{h_{m_{n}}}}\bbG_{m_{n}}[\be_{3}^{\tp}\bOmega^{-1}\bPsi_{h_{m_{n}}}(v)] +  o_{\bbP}(1),
\end{equation}
and rewrite the expression of $\bG_{m_{n}}[\cdot]$ as follows. 
\begin{align*}
\label{eq:Gstep1}
    &\frac{1}{m_{n}h_{m_{n}}}\sum_{i=1}^{m_{n}}
    \left(
        \be_{3}^{\tp}\bOmega^{-1}\bPsi_{h_{m_{n}}}(v; U_{i}, V_{i}) - 
        \E\{\be_{3}^{\tp}\bOmega^{-1}\bPsi_{h_{m_{n}}}(v; U_{i}, V_{i})\}
    \right)\\
    &= \frac{1}{m_{n}h_{m_{n}}}\sum_{i=1}^{m_{n}}
        U_{i}\left(
        \be_{3}^{\tp}\bOmega^{-1}\Tilde{\bPsi}_{h_{m_{n}}}(v; V_{i}) - 
        \E\{\be_{3}^{\tp}\bOmega^{-1}\Tilde{\bPsi}_{h_{m_{n}}}(v; V_{i})\}
    \right),
    \numberthis
\end{align*}
where $$\bPsi(v; Z_{i}, Z_{2}) = Z_{1}\Tilde{\bPsi}(v; Z_{2})$$ for some function $\Tilde{\bPsi}$. Since the function $\Tilde{\bPsi}$s are linear combination of the functions from $\sK_{6}$ defined in Condition \ref{cond:VC} in Section \ref{sec:theory}, it follows from \citet{einmahl2005uniform} that  
\begin{align*}
\label{eq:Gstep2}
 &\sup_{v \in [v_{1}, v_{2}]}
    \left|
    \sum_{i=1}^{m_{n}}
        U_{i}\left(
        \be_{3}^{\tp}\bOmega^{-1}\Tilde{\bPsi}_{h_{m_{n}}}(v; V_{i}) - 
        \E\{\be_{3}^{\tp}\bOmega^{-1}\Tilde{\bPsi}_{h_{m_{n}}}(v; V_{i})\}\right)
    \right|\\
    &= O_{\bbP}\left( \sqrt{m_{n}h_{m_{n}} \log m_{n}} \right).
    \numberthis
\end{align*}
Hence, using combining Equations \eqref{eq:Gstep1} and \eqref{eq:Gstep2}, we have 
\begin{align*}
    \label{eq:line1}
    \sup_{v \in [v_{1}, v_{2}]}\left|
    \frac{1}{\sqrt{h_{m_{n}}}}\bbG_{m_{n}}[\be_{3}^{\tp}\bOmega^{-1}\bPsi_{h_{m_{n}}}(v)]
    \right| &= O_{\bbP}\left(\sqrt{ 
        \frac{\log m_{n}}{m_{n}h_{m_{n}}}}
    \right)
    \numberthis
\end{align*}
and using the same argument, 
\begin{align*}
    \label{eq:line2}
    \sup_{v \in [v_{1}, v_{2}]}\left|
    \frac{1}{\sqrt{h_{m_{n}}}}\bbG_{m_{n}}[\be_{3}^{\tp}\bPsi_{h_{m_{n}}}(v)]
    \right| &= O_{\bbP}\left(\sqrt{ 
        \frac{\log m_{n}}{m_{n}h_{m_{n}}}}
    \right).
    \numberthis
\end{align*}
Therefore, using \eqref{47}, we have 
\begin{align*}
\label{eq:line3}
    &\widehat{g}_{h_{m_{n}}}''(v) - \E\{\widehat{g}_{h_{m_{n}}}''(v)\}\\
    &= \frac{1}{\sqrt{m_{n}h_{m_{n}}^{5}}}
    \left\{ 
        \frac{1}{\sqrt{h_{m_{n}}}}\bbG_{m_{n}}[\be_{3}^{\tp}\bOmega^{-1}\bPsi_{h_{m_{n}}}(v)] \right.\\
    & \qquad   \left. +
        \frac{1}{\sqrt{h_{m_{n}}}}\bbG_{m_{n}}[\be_{3}^{\tp}\bPsi_{h_{m_{n}}}(v)]
        \left(
          O(h_{m_{n}}) + O_{\bbP}\left(
            \sqrt{\frac{\log m_{n}}{m_{n}h_{m_{n}}}} 
          \right)  
        \right)
    \right\}.
    \numberthis
\end{align*}
Afterwards, using \eqref{eq:line1}, \eqref{eq:line2} and \eqref{eq:line3}, we have
\begin{align*}
    \| \widehat{g}_{h_{m_{n}}}'' - \E\{\widehat{g}_{h_{m_{n}}}''\} \|_{\infty} &= 
    O_{\bbP}\left(
        \sqrt{\frac{\log m_{n}}{m_{n}^{2}h_{m_{n}}^6}}\left\{
            1 + O(h_{m_{n}}) + O_{\bbP}\left(
            \sqrt{\frac{\log m_{n}}{m_{n}h_{m_{n}}}} 
          \right)  
        \right\}
    \right)\\
    &= O_{\bbP}\left(
        \sqrt{\frac{\log m_{n}}{m_{n}^{2}h_{m_{n}}^6}}\right).
    \numberthis
\end{align*}
Therefore, using the upper bound of the bias of $\widehat{g}_{h_{m_{n}}}''$ stated in Theorem \ref{thm:bias}, we have 
\begin{align*}
    \| \widehat{g}_{h_{m_{n}}}'' - g''\|_{\infty} = O\left(h_{m_{n}}^{2+\delta}\right) + h_{m_{n}}^{2}O\left(\frac{1}{\sqrt{m_{n}h_{m_{n}}}}\right) + O_{\bbP}\left(
        \sqrt{\frac{\log m_{n}}{m_{n}^{2}h_{m_{n}}^6}}\right).
    \numberthis
\end{align*} 
This completes the proof of the second part of this theorem. 
\end{proof}


\begin{proof}[Proof of Theorem \ref{thm:T}]
First, we would like to establish that under the assumptions \ref{cond:kernel}, \ref{cond:VC} and \ref{cond:band}, we have
    \begin{equation}
    \label{eq:T-tilde}
        \sup_{t\in \bbR}\left|\bbP\left\{ \sqrt{m_{n}h_{m_{n}}^{5}}(\Tilde{T}_{n} - T) \leq t\right\} - 
        \bbP\left\{\sup_{f \in \sG_{m_{n}}} |B_{m_{n}}(f)| \leq t\right\}\right|
        = O\left(\left(\frac{\log^{7}m_{n}}{m_{n}h_{m_{n}}^{5}}\right)^{1/8}\right), 
    \end{equation} where for any non-random $v_{1} < v_{2}$, $$\Tilde{T}_{n} = \sup\limits_{v \in [v_{1}, v_{2}]}|\widehat{g}_{h_{m_{n}}}''(v)|,$$ $$T = \sup\limits_{v\in[v_1, v_2]}|g''(v)|,$$ and $B_{m_{n}}$ is the same as defined in Section \ref{sec:theory}.
Now, observe that 
\begin{align*}
    \Tilde{T}_{n} - T &= \sup_{v \in [v_{1}, v_{2}]}|\widehat{g}_{h_{m_{n}}}(v)| - \sup_{v \in [v_{1}, v_{2}]}|g''(v)|\\
    & \geq \sup_{v\in [v_{1}, v_{2}]}\left|\widehat{g}_{h_{m_{n}}}''(v) - g''(v)\right|.
\end{align*}
Now, using \ref{cond:kernel}, \ref{cond:VC} and \ref{cond:band}, it is enough to show that, 
\begin{align*}
        &\sup_{t\in \bbR}\left|\bbP\left\{ \sqrt{m_{n}h_{m_{n}}^{5}}
        \| \widehat{g}_{h_{m_{n}}}'' - g''\|_{\infty}
        \leq t\right\} - 
        \bbP\left\{\sup_{f \in \sG_{m_{n}}} |B_{m_{n}}(f)| \leq t\right\}\right|\\
        &= O\left(\left(\frac{\log^{7}m_{n}}{m_{n}h_{m_{n}}^{5}} \right)^{1/8}\right), 
    \numberthis
\end{align*}
where $\| \widehat{g}_{h_{m_{n}}}'' - g''\|_{\infty} = \sup\limits_{v\in [v_{1}, v_{2}]}\left|\widehat{g}_{h_{m_{n}}}''(v) - g''(v)\right|$. Now, using the fact in  Theorem \eqref{thm:emp}, the scaled difference between the estimated derivatives of $g''$ and the true $g''$ can be expressed as 
\begin{equation}
\label{eq:scale-diff}
\sqrt{m_{n}h_{m_{n}}^{5}}\left\|\widehat{g}_{h_{m_{n}}}'' - g'' \right\|_{\infty} = \sqrt{m_{n}h_{m_{n}}^{5}}\left\|\widehat{g}_{h_{m_{n}}}'' - \E\{g''\} \right\|_{\infty} + O_{\bbP}(1).
\end{equation}
Next, using Corollary 2.2 and the derivation of Proposition 3.1. in \citet{chernozhukov2014gaussian}, one can conclude that there exists a tight Gaussian process $B_{m_{n}}$ and constants $A_{1}, A_{2} > 0$ such that for any $\gamma > 0$, 
\begin{equation}\label{58}
    \bbP\left( \left|\sqrt{m_{n}h_{m_{n}}^{5}}\| \widehat{g}_{h_{m_{n}}}'' - \E\{\widehat{g}_{h_{m_{n}}}''\}\|_{\infty} - \sup_{f \in \sG_{m_{n}}}|B_{m_{n}}(f)| \right| \leq \frac{A_{1}\log^{2/3} m_{n}}{\gamma^{1/3}m_{n}^{1/6}h_{m_{n}}^{5/6}} \right) \leq A_{2}\gamma,
\end{equation}
for large $m_{n}$. Hence, using \eqref{eq:scale-diff} and \eqref{58}, we have 
\begin{equation}
    \bbP\left( \left|\sqrt{m_{n}h_{m_{n}}^{5}}\| \widehat{g}_{h_{m_{n}}}'' - {g}''\|_{\infty} - \sup_{f \in \sG_{m_{n}}}|B_{m_{n}}(f)| \right| \leq \frac{A_{3}\log^{2/3} m_{n}}{\gamma^{1/3}m_{n}^{1/6}h_{m_{n}}^{5/6}} \right) \leq A_{2}\gamma,
\end{equation}
for some constants $A_{3}$. Afterward, applying anti-concentration inequality (for more details, see Theorem 3 of \citet{chernozhukov2015comparison}) 
on $\sqrt{m_{n}h_{m_{n}}^{5}}
        \| \widehat{g}_{h_{m_{n}}}'' - g''\|_{\infty}
        $, for large $m_{n}$, there exists a positive constant $A_{4}$ such that  
\begin{align}\label{60}
    &\sup_{t\in \bbR}\left|\bbP\left\{ \sqrt{m_{n}h_{m_{n}}^{5}}
        \| \widehat{g}_{h_{m_{n}}}'' - g''\|_{\infty}
        \leq t\right\} - 
        \bbP\left\{\sup_{f \in \sG_{m_{n}}} |B_{m_{n}}(f)| \leq t\right\}\right|\nonumber\\
    & \qquad \leq A_{4} \E\left\{\sup_{f \in \sG_{m_{n}}}|B_{m_{n}}(f)|\right\} \frac{A_{1}\log^{2/3} m_{n}}{\gamma^{1/3}m_{n}^{1/6}h_{m_{n}}^{5/6}} + A_{2}\gamma
\end{align}
Moreover, using Dudley's inequality of Gaussian process (see Corollary 2.2.8 of \citet{MR1385671}), we have 
\begin{equation}\label{61}
    \E\left\{\sup_{f \in \sG_{m_{n}}}|B_{m_{n}}(f)|\right\} = O(\sqrt{\log m_{n}}).  
\end{equation}
Finally, With the optimal $\gamma = \left(\frac{\log^{7}m_{n}}{m_{n}h_{m_{n}}^{5}} \right)^{1/8}$, and in view of \eqref{60} and \eqref{61}, we have  
\begin{align*}
    &\sup_{t\in \bbR}\left|\bbP\left\{ \sqrt{m_{n}h_{m_{n}}^{5}}
        \| \widehat{g}_{h_{m_{n}}}'' - g''\|_{\infty}
        \leq t\right\} - 
        \bbP\left\{\sup_{f \in \sG_{m_{n}}} |B_{m_{n}}(f)| \leq t\right\}\right|\\
    &= O\left(\left(\frac{\log^{7}m_{n}}{m_{n}h_{m_{n}}^{5}} \right)^{1/8}\right).
    \numberthis
\end{align*}
This concludes the proof of \eqref{eq:T-tilde}.
Finally, the straightforward application of the facts in Lemmas \ref{lemma:beta} and \ref{lemma:random-bound} on $\Tilde{T}_{n}$, the proof of this theorem follows. 
\end{proof}

\begin{proof}[Proof of Theorem \ref{thm:bootstrap}]
    From the assertion in Theorem \ref{thm:T}, it is known that there exists a Gaussian process $B_{m_{n}}$ defined on $\sG_{m_{n}}$ such that 
    \begin{equation}\label{63}
\sqrt{m_{n}h_{m_{n}}^{5}}\|\widehat{g}_{h_{m_{n}}}'' - g''\|_{\infty} \approx \sup_{f \in \sG_{m_{n}}}|B_{m_{n}}|.  
    \end{equation}
   Further, observe that conditioning on $V_{1}, \ldots, V_{m_{n}}$, we have
    \begin{equation}\label{64}
\sqrt{m_{n}h_{m_{n}}^{5}}\|\widehat{g''}_{h_{m_{n}}}^{*} - \widehat{g''}\|_{\infty} = \sqrt{m_{n}h_{m_{n}}^{5}}\|\widehat{g''}_{h_{m_{n}}}^{*} - \E\{\widehat{g''}_{h_{m_{n}}}^{*}|\sV_{m_{n}}\}\|_{\infty}, 
    \end{equation} where $\sV_{m_{n}} = \{V_{1}, \ldots, V_{m_{n}}\}$. 
    Similar to the proof of Theorem \ref{thm:T}, one can approximate $$\sqrt{m_{n}h_{m_{n}}^{5}}\|\widehat{g''}_{h_{m_{n}}}^{*} - \E\{\widehat{g''}_{h_{m_{n}}}^{*}|\sX\}\|_{\infty}$$ by the maximum of the empirical process, and therefore, we have 
    \begin{align*}\label{65}
    &\sup_{t\in \bbR}\left|\bbP\left\{ \sqrt{m_{n}h_{m_{n}}^{5}}
        \| \widehat{g}_{h_{m_{n}}}''^{*} - \widehat{g}_{h_{m_{n}}}''\|_{\infty}
        \leq t \mid \sV_{m_{n}}\right\} - 
        \bbP\left\{\sup_{f \in \sG_{m_{n}}} |B^{\star}_{m_{n}}(f)| \leq t\mid \sV_{m_{n}}\right\}\right|\\
        &= O\left(\left(\frac{\log^{7}m_{n}}{m_{n}h_{m_{n}}^{5}} \right)^{1/8}\right),
        \numberthis
    \end{align*}
    where $B^{\star}_{m_{n}}$ is a Gaussian processes defined on $\sG_{m_{n}}$ such that for any $f_{1}, f_{2} \in \sG_{m_{n}}$, 
    \begin{align}
        \E\{B^{\star}_{m_{n}}(f_{1})B^{\star}_{m_{n}}(f_{2})\mid \sV_{m_{n}}\} &= \widehat{\cov}\{f_{1}(V), f_{2}(V)\},         
    \end{align} 
    where $$\widehat{\cov}\{f_{1}(V), f_{2}(V)\} = \frac{1}{m_{n}}\sum_{j=1}^{m_{n}}f_{1}(V_{j})f_{2}(V_{j}) - \left\{ \frac{1}{m_{n}}\sum_{j=1}^{m_{n}}f_{1}(V_{j}) \right\} \left\{ \frac{1}{m_{n}}\sum_{j=1}^{m_{n}}f_{2} (V_{j})\right\}.$$
    As $\widehat{\cov}\{f_{1}(V), f_{2}(V)\}$ converges in probability to its population version
    \begin{equation}
        \cov\{f_{1}(V), f_{2}(V)\} = \E\left\{f_{1}(V)f_{2}(V)\right\} - \E\left\{f_{1}(V)\}\E\{f_{2}(V)\right\}, 
    \end{equation}
    then by the Gaussian approximation asserted in Theorem 2 of \citet{chernozhukov2015comparison}, \begin{equation}\label{67}
    \sup\limits_{f\in \sG_{m_{n}}}|B^{\star}_{m_{n}}(f) - B_{m_{n}}(f)|\stackrel{p}\rightarrow 0.\end{equation} Hence, using \eqref{63}, \eqref{65} and \eqref{67}, we have 
    \begin{equation}
        \sqrt{m_{n}h_{m_{n}}^{5}}\| \widehat{g}_{h_{m_{n}}}'' - g''\|_{\infty} \approx \sup_{f \in \sG_{m_{n}}}|B_{m_{n}}(f)| \approx \sup_{f \in \sG_{m_{n}}}|B^{\star}_{m_{n}}(f)| \approx 
        \sqrt{m_{n}h_{m_{n}}^{5}}\| \widehat{g}_{h_{m_{n}}}''^{*} - \widehat{g}''\|_{\infty}, 
    \end{equation}
    and finally, further application of \eqref{65}, we have
    \begin{align*}
        &\sup_{t\in \bbR}\left|\bbP\left\{ \sqrt{m_{n}h_{m_{n}}^{5}}
        \| \widehat{g}_{h_{m_{n}}}''^{*} - \widehat{g}_{h_{m_{n}}}''\|_{\infty}
        \leq t \mid \sV_{m_{n}}\right\} - 
        \bbP\left\{ \sqrt{m_{n}h_{m_{n}}^{5}}
        \| \widehat{g}_{h_{m_{n}}}''^{*} - \widehat{g}_{h_{m_{n}}}''\|_{\infty}
        \leq t \mid \sV_{m_{n}}\right\}
        \right|\\
        &= O\left(\left(\frac{\log^{7}m_{n}}{m_{n}h_{m_{n}}^{5}} \right)^{1/8}\right).
    \end{align*}
    It completes the proof. 
\end{proof}

\section{Figures related to finite sample performance and real data analysis.}
In this section of the appendix, we present all figures related to the finite-sample performance and the DTI data analysis. These figures illustrate the quality of the estimates for $\beta_{1}, \beta_{2}$ and $g''$.

\begin{figure}[H]
     \centering
     \begin{subfigure}[b]{0.32\textwidth}
         \centering
         \includegraphics[width=\textwidth, height = \textwidth, keepaspectratio]{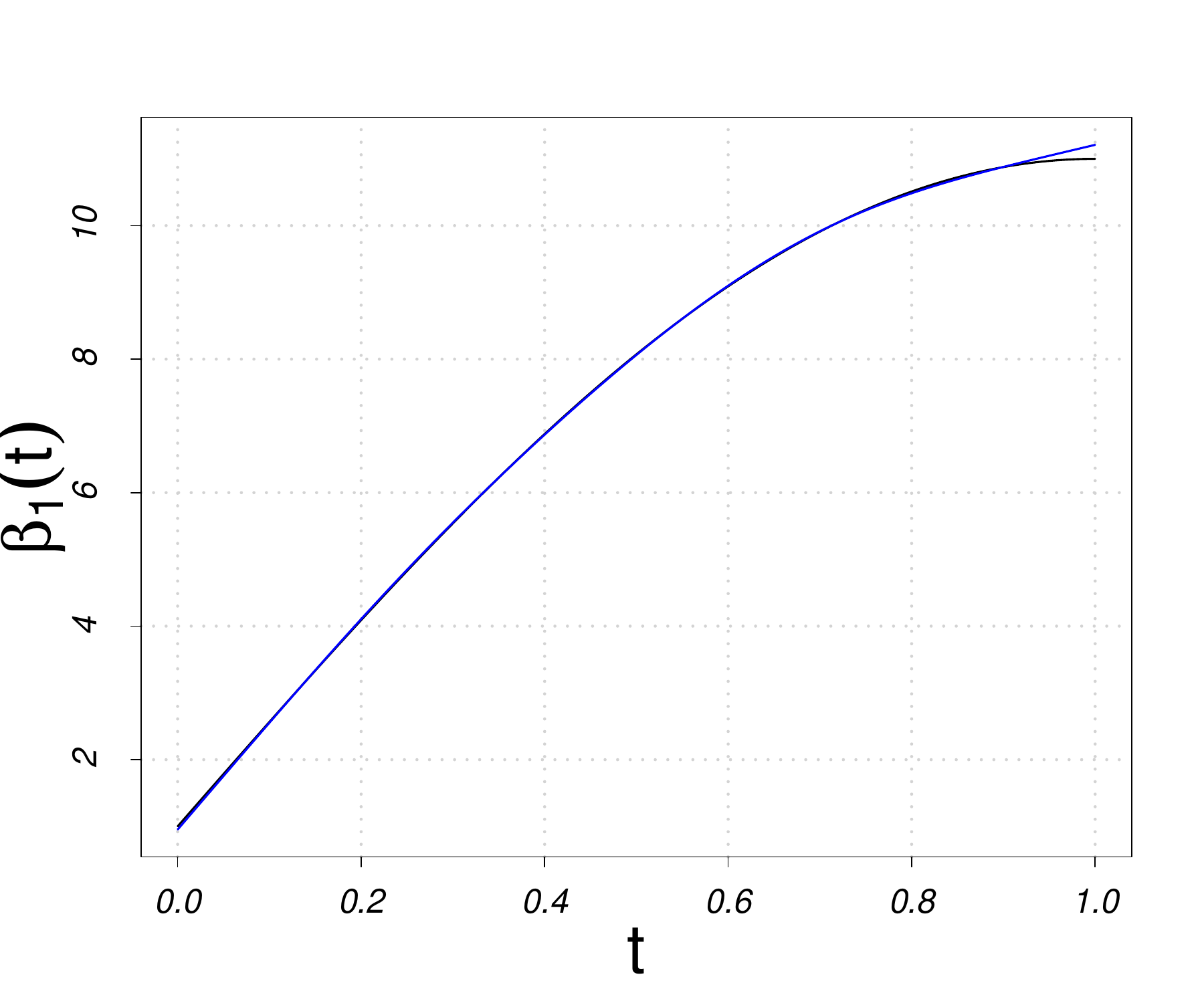}
         \caption{$\beta_{1}(t)$}
     \end{subfigure}
     \begin{subfigure}[b]{0.32\textwidth}
         \centering
         \includegraphics[width=\textwidth, height = \textwidth, keepaspectratio]{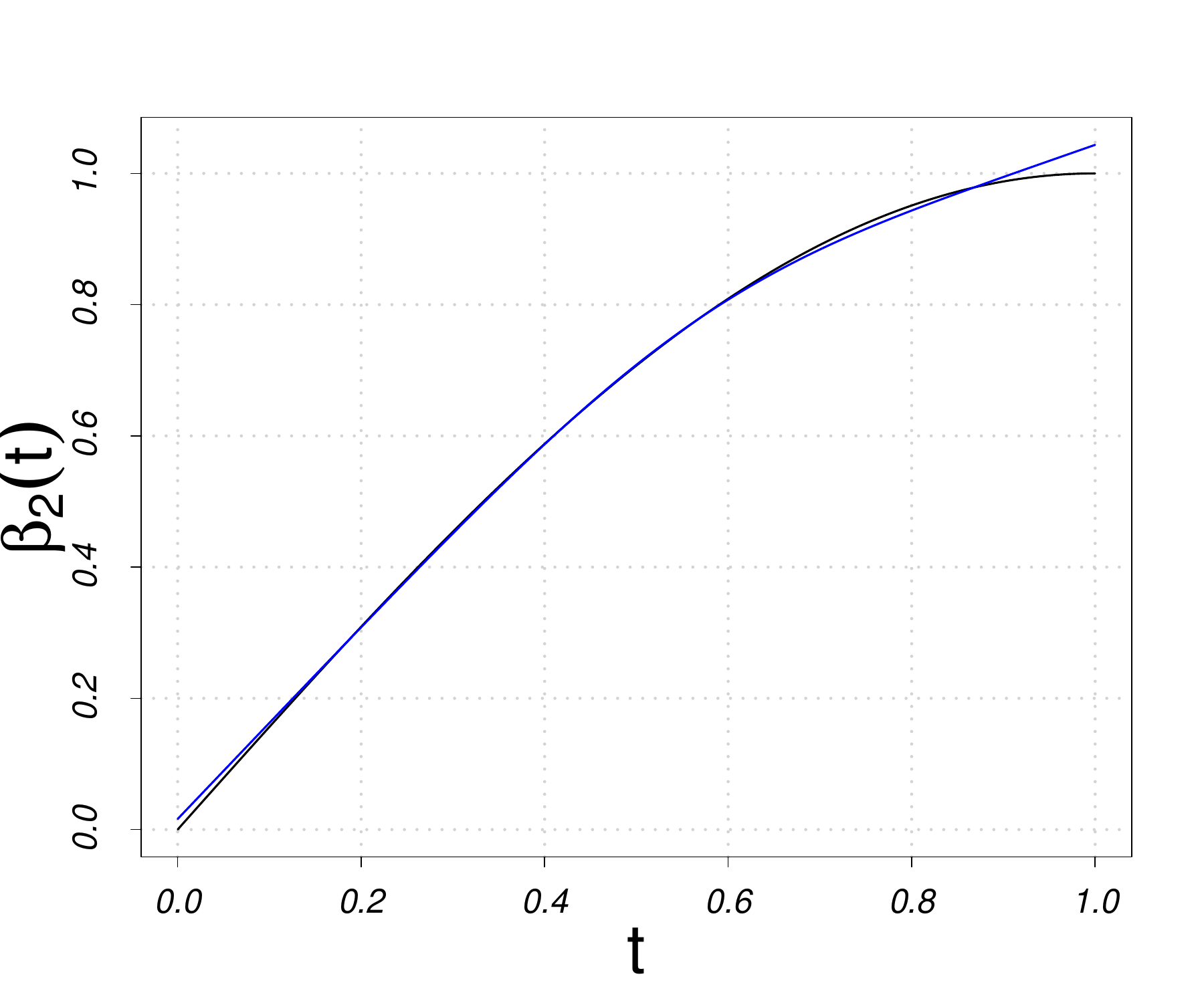}
         \caption{$\beta_{2}(t)$}
     \end{subfigure}
     \begin{subfigure}[b]{0.32\textwidth}
         \centering
         \includegraphics[width=\textwidth, height = \textwidth, keepaspectratio]{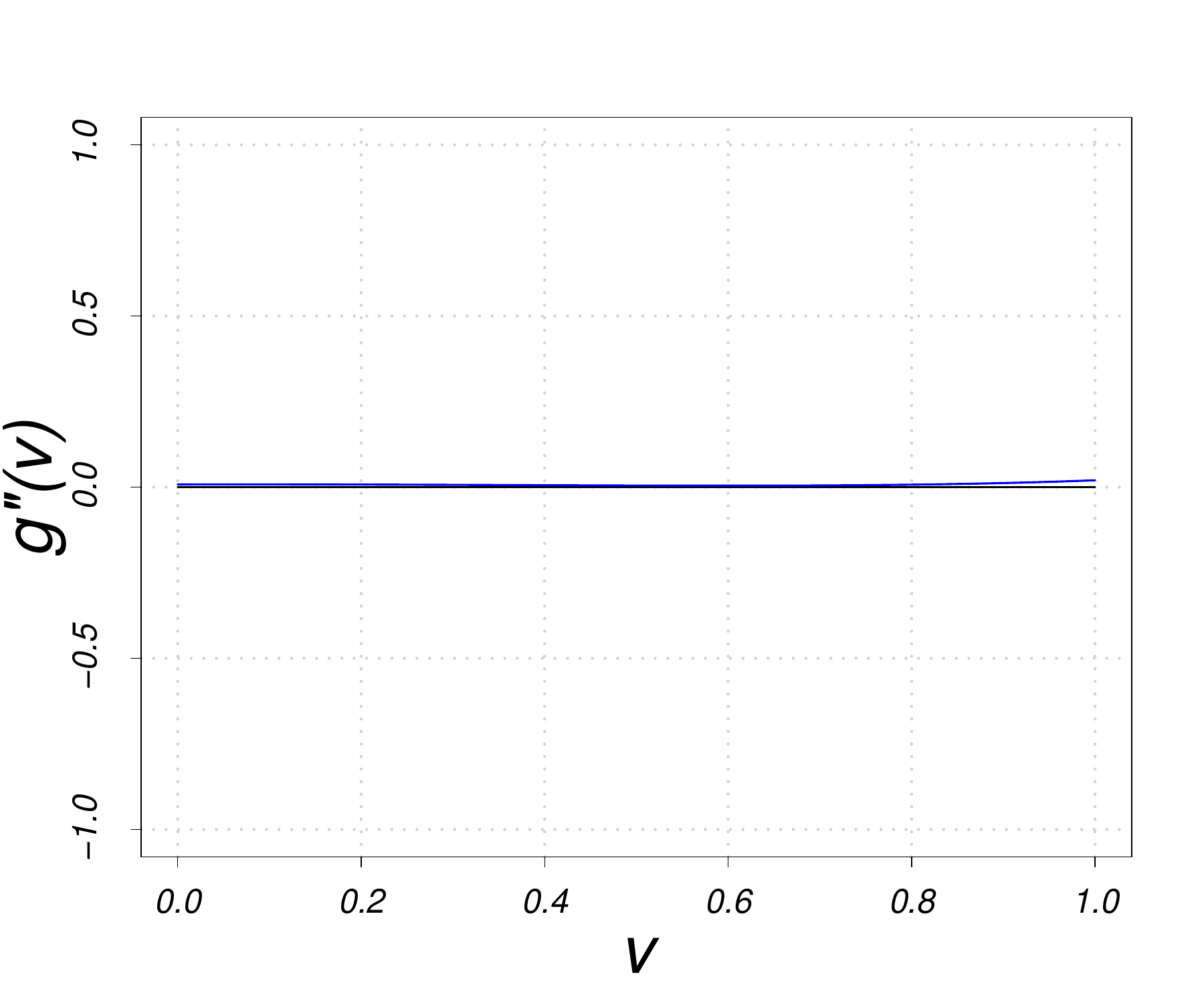}
         \caption{$g''(v)$}
     \end{subfigure}
     \caption{
     Comparison of estimated (blue) and true (black) functional coefficients $\beta_{1}(t)$, $\beta_{2}(t)$ from Model \eqref{eq:model_1} and second derivative $g''(v)$ from Model \eqref{eq:single-index} for Example \eqref{fn:linear} with
     $g(v) = 10v + 1$.}
     \label{fig:linear}
\end{figure}

\begin{figure}[H]
     \centering
     \begin{subfigure}[b]{0.32\textwidth}
         \centering
         \includegraphics[width=\textwidth, height = \textwidth, keepaspectratio]{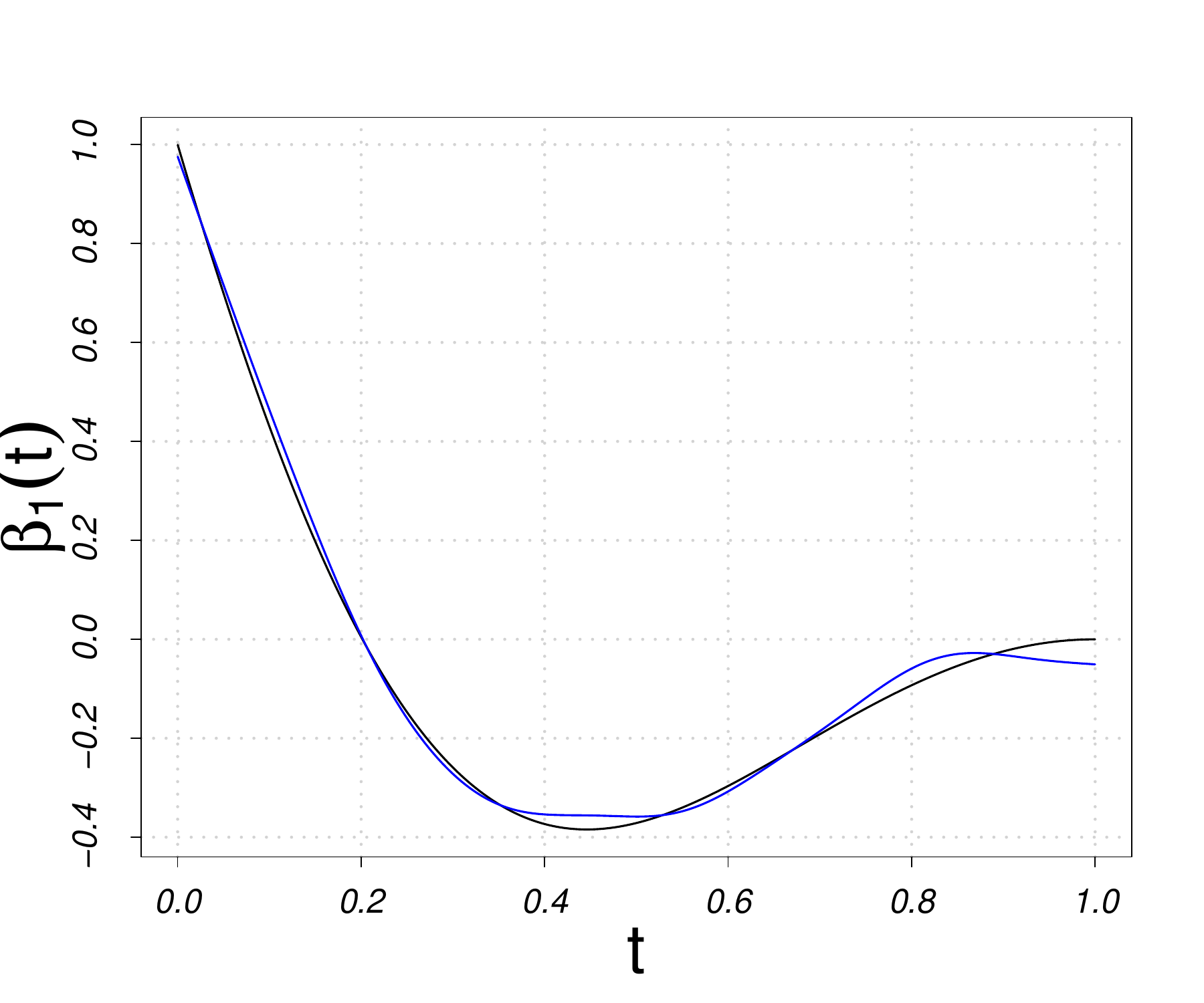}
         \caption{$\beta_{1}(t)$}
     \end{subfigure}
     \begin{subfigure}[b]{0.32\textwidth}
         \centering
         \includegraphics[width=\textwidth, height = \textwidth, keepaspectratio]{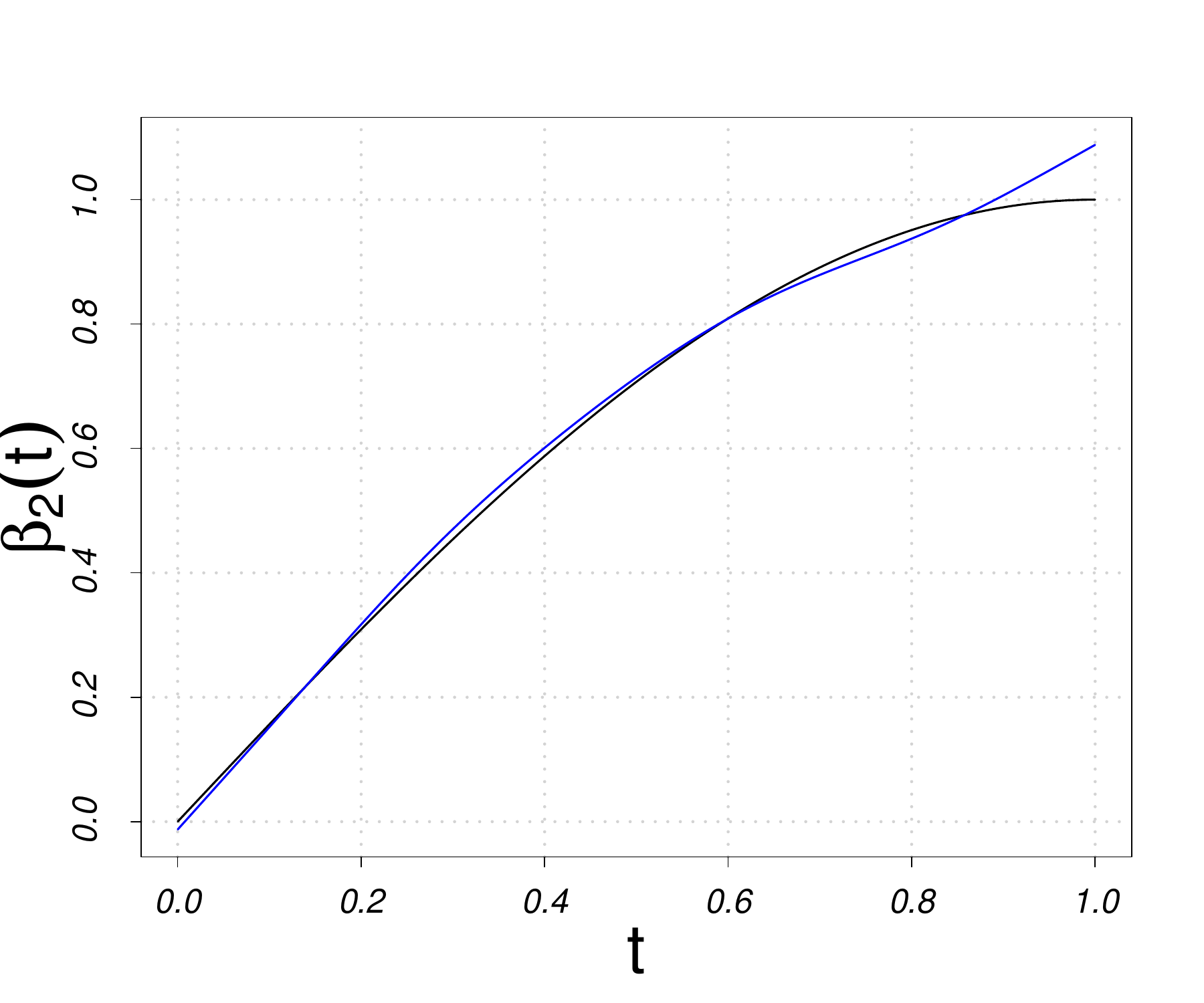}
         \caption{$\beta_{2}(t)$}
     \end{subfigure}
     \begin{subfigure}[b]{0.32\textwidth}
         \centering
         \includegraphics[width=\textwidth, height = \textwidth, keepaspectratio]{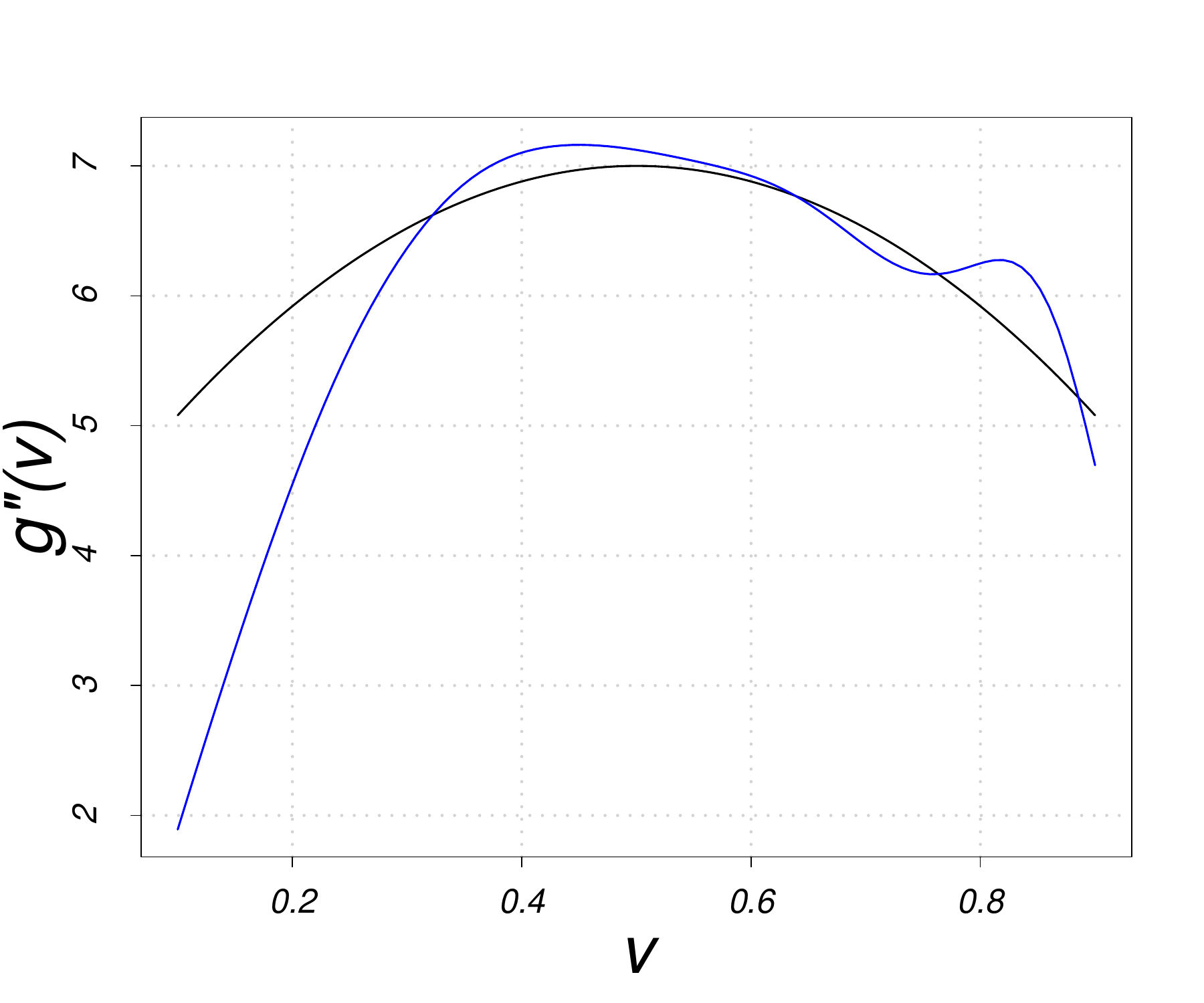}
         \caption{$g''(v)$}
     \end{subfigure}
     
     \caption{
     Comparison of estimated (blue) and true (black) functional coefficients $\beta_{1}(t)$, $\beta_{2}(t)$ from Model \eqref{eq:model_1} and second derivative $g''(v)$ from Model \eqref{eq:single-index} for Example \eqref{fn:poly} with
     $g(v) = 1- 4v + 2v^{2} + 2v^{3} - v^{4}$.}
     \label{fig:poly}
\end{figure}

\begin{figure}[H]
     \centering
     \begin{subfigure}[b]{0.32\textwidth}
         \centering
         \includegraphics[width=\textwidth, height = \textwidth, keepaspectratio]{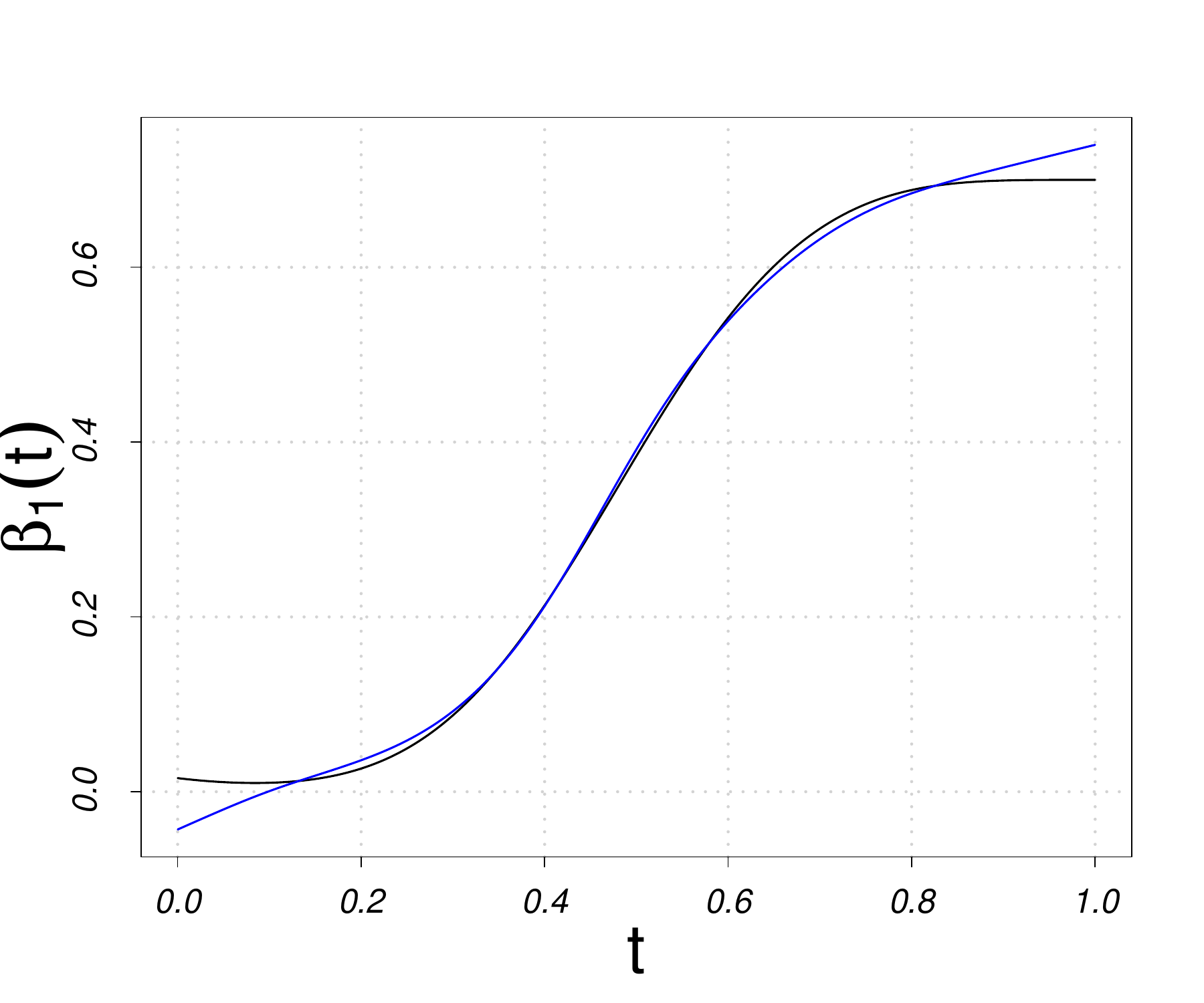}
         \caption{$\beta_{1}(t)$}
     \end{subfigure}
     \begin{subfigure}[b]{0.32\textwidth}
         \centering
         \includegraphics[width=\textwidth, height = \textwidth, keepaspectratio]{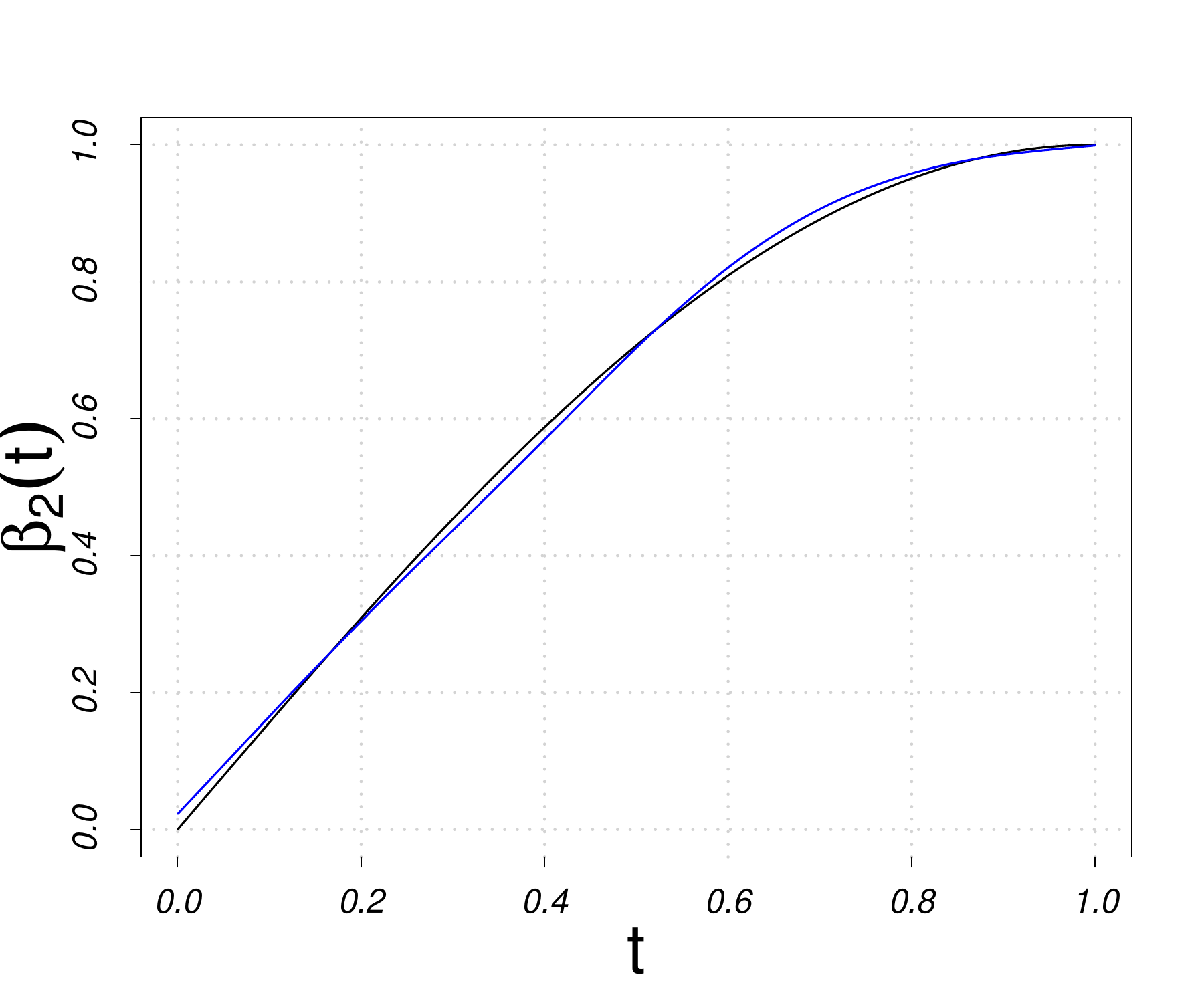}
         \caption{$\beta_{2}(t)$}
     \end{subfigure}
     \begin{subfigure}[b]{0.32\textwidth}
         \centering
         \includegraphics[width=\textwidth, height = \textwidth, keepaspectratio]{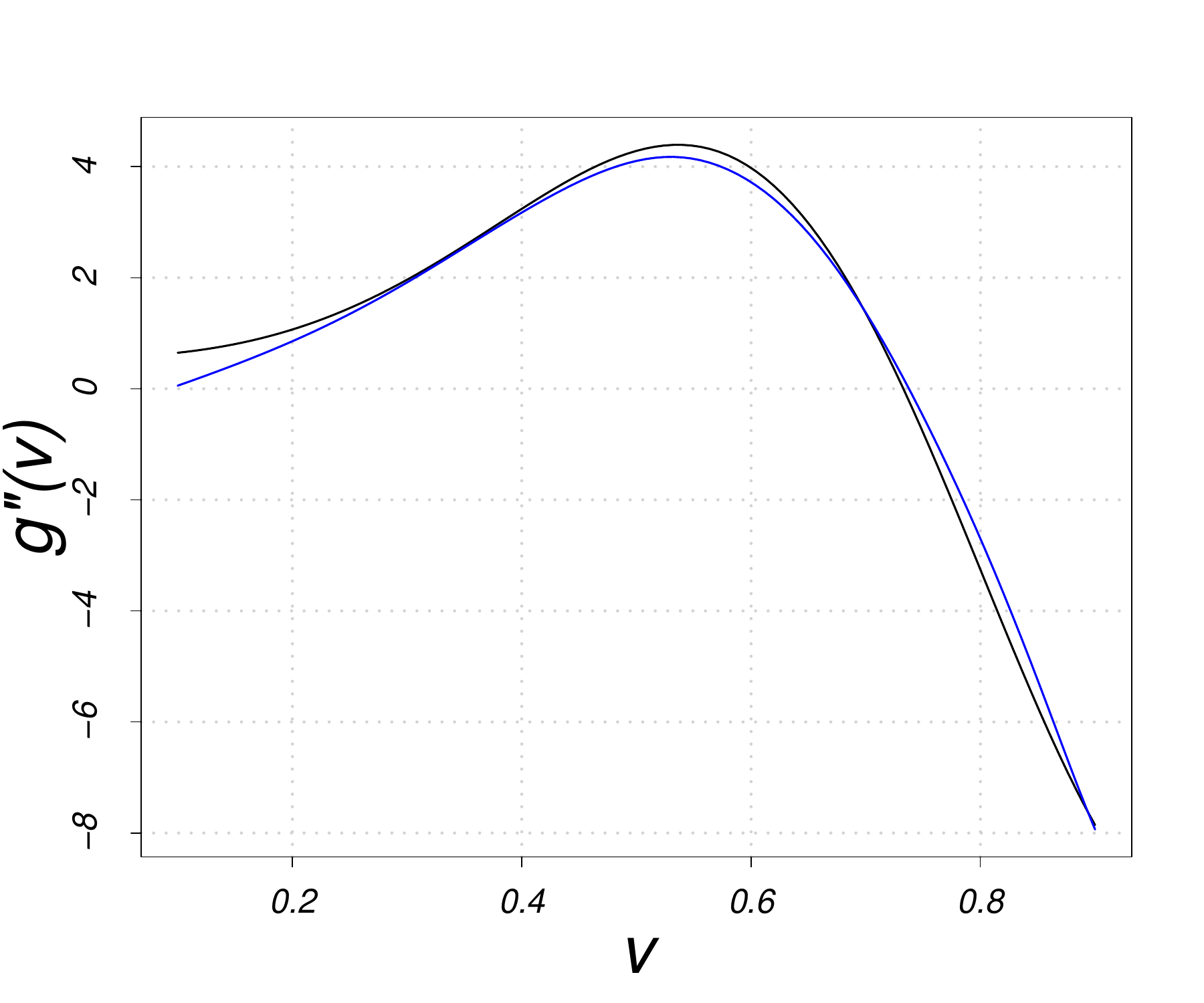}
         \caption{$g''(v)$}
     \end{subfigure}
     \caption{
     Comparison of estimated (blue) and true (black) functional coefficients $\beta_{1}(t)$, $\beta_{2}(t)$ from Model \eqref{eq:model_1} and second derivative $g''(v)$ from Model \eqref{eq:single-index} for Example \eqref{fn:exp} with
     $g(v) = 0.3\exp(-3(v+1)^{2}) + 0.7\exp(-7(v-1)^{2})$.}
     \label{fig:exp}
\end{figure}

\begin{figure}[H]
     \centering
     \begin{subfigure}[b]{0.32\textwidth}
         \centering
         \includegraphics[width=\textwidth, height = \textwidth, keepaspectratio]{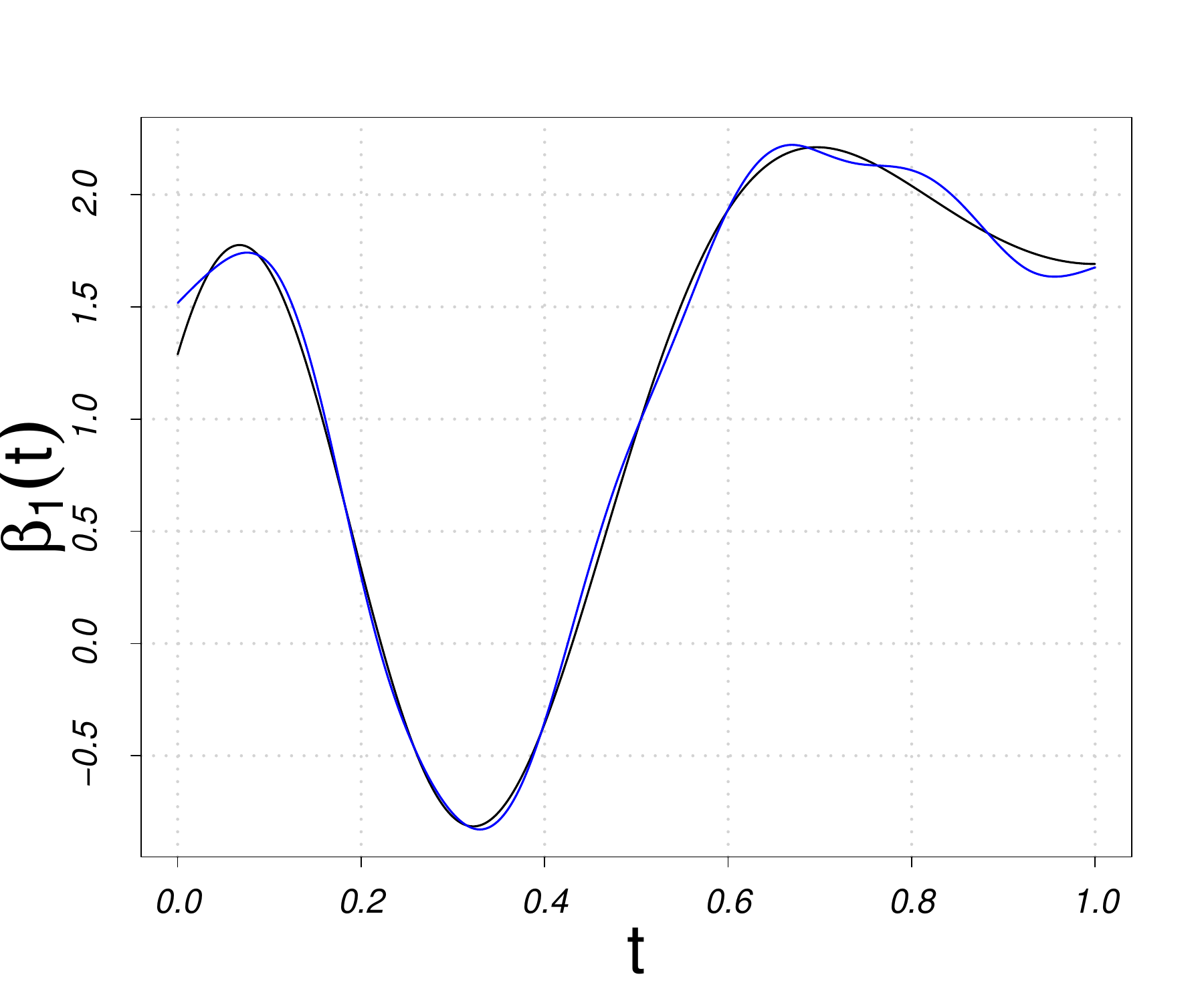}
         \caption{$\beta_{1}(t)$}
     \end{subfigure}
     \begin{subfigure}[b]{0.32\textwidth}
         \centering
         \includegraphics[width=\textwidth, height = \textwidth, keepaspectratio]{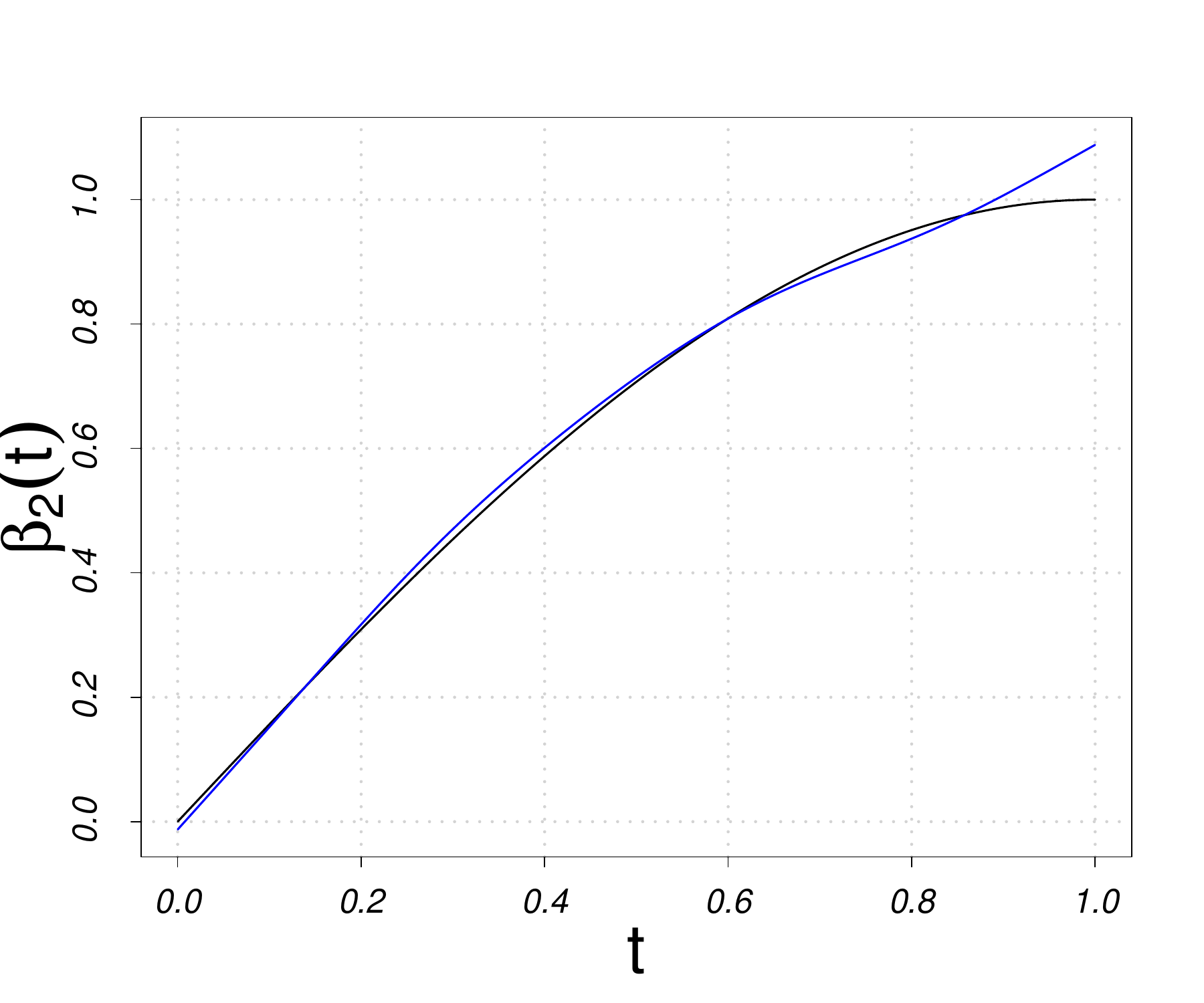}
         \caption{$\beta_{2}(t)$}
     \end{subfigure}
     \begin{subfigure}[b]{0.32\textwidth}
         \centering
         \includegraphics[width=\textwidth, height = \textwidth, keepaspectratio]{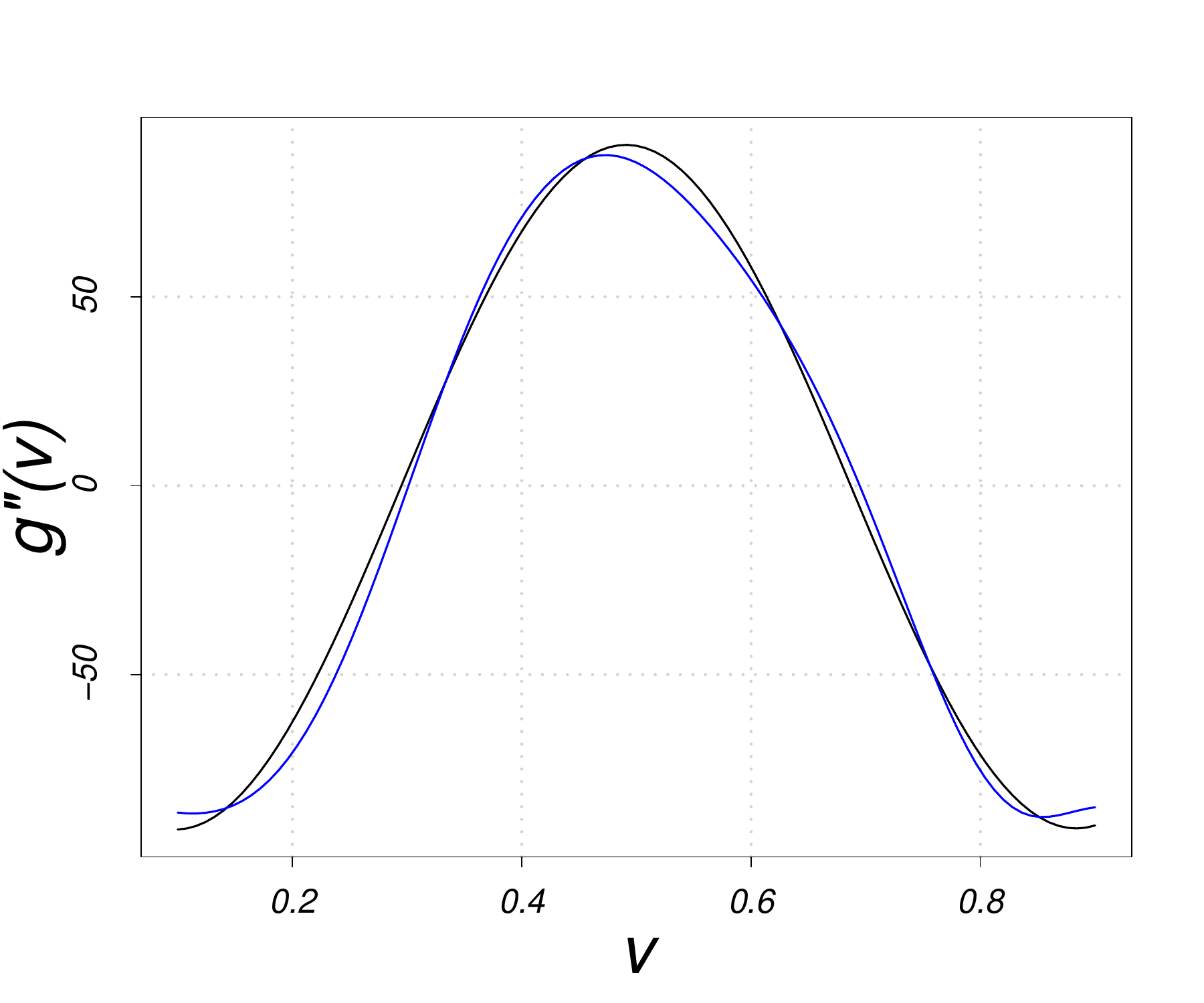}
         \caption{$g''(v)$}
     \end{subfigure}
     \caption{
     Comparison of estimated (blue) and true (black) functional coefficients $\beta_{1}(t)$, $\beta_{2}(t)$ from Model \eqref{eq:model_1} and second derivative $g''(v)$ from Model \eqref{eq:single-index} for Example \eqref{fn:sincos-log} with
     $g(v) = \sin(8v) + \cos(8v) + \log(4/3 + v)$.}
     \label{fig:sincos-log}
\end{figure}

\begin{figure}[H]
    \begin{subfigure}[b]{0.49\textwidth}
         \centering
         \includegraphics[width=\textwidth, height = \textwidth, keepaspectratio]{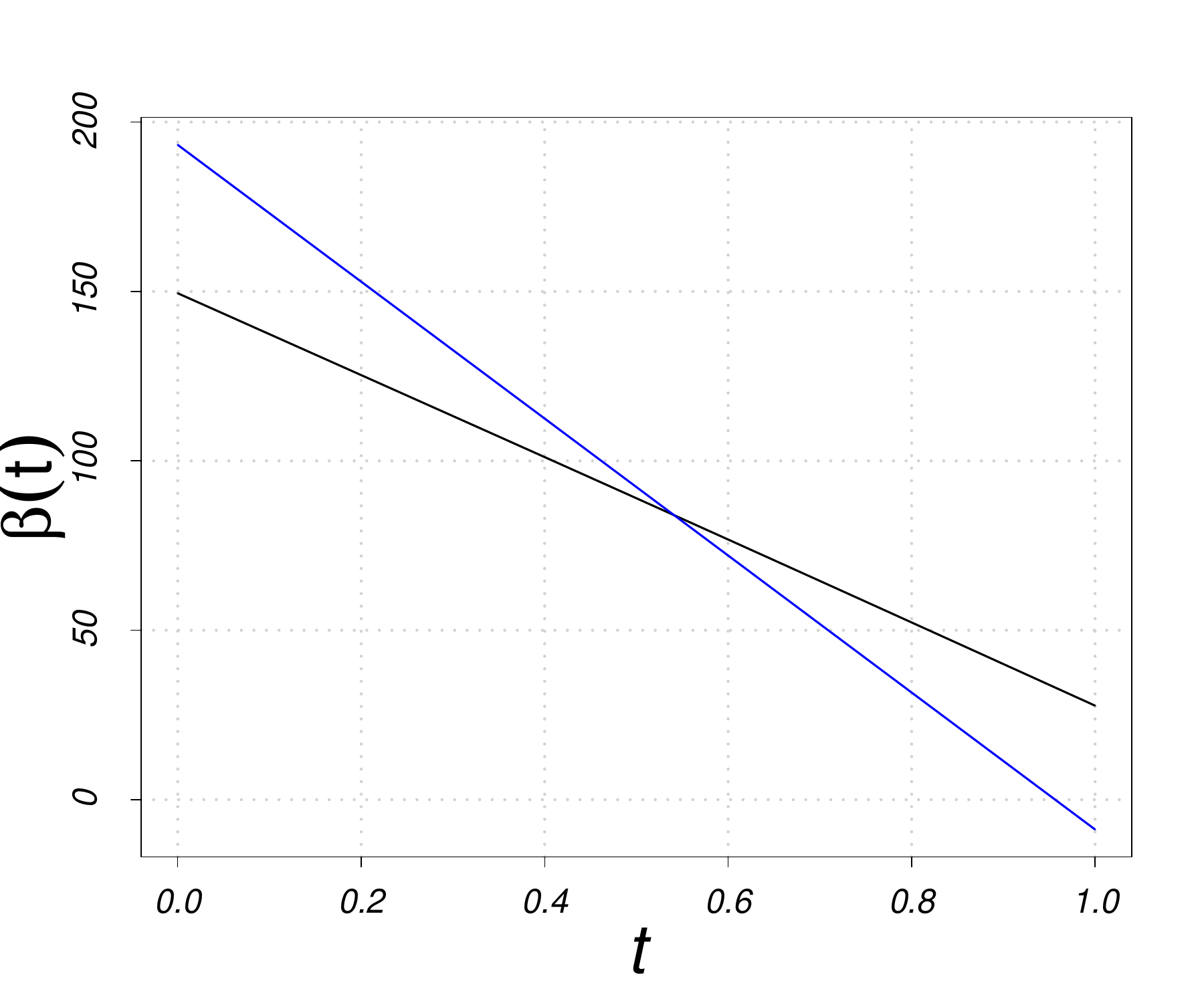}
         \caption{$\beta_{1}(t)$ (black) and $\beta_{2}(t)$ (blue)}
    \end{subfigure}
    \begin{subfigure}[b]{0.49\textwidth}
         \centering
         \includegraphics[width=\textwidth, height = \textwidth, keepaspectratio]{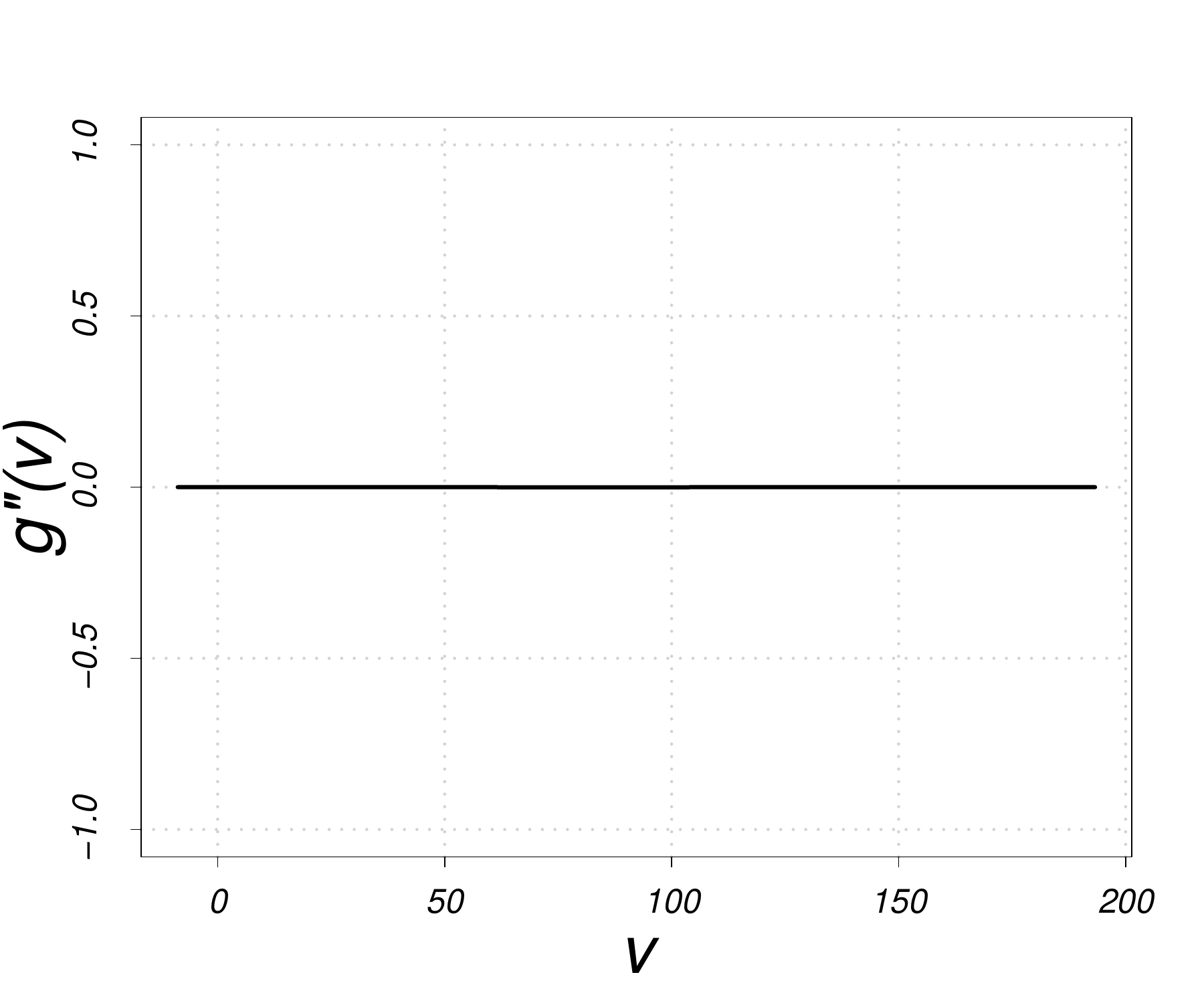}
         \caption{$g''(v)$}
    \end{subfigure}
    \centering
    \caption{
    The plots of estimated coefficients $\beta_{1}(t)$, $\beta_{2}(t)$ from Model \eqref{eq:model_1} and second derivative $g''(v)$ from Model \eqref{eq:single-index} based on the DTI analysis.}
    \label{fig:dti}
\end{figure}

\end{appendix}

\vspace{0.1 in}


\begin{acks}[Acknowledgments]
Subhra Sankar Dhar presented a part of this work at Renmin University of China when he was visiting there, and many stimulating questions asked by the audience improved the content of the article. He is also thankful to Dr. Hengrui Cai of UC Irvine for a few interesting suggestions on various parts of this work. 
\end{acks}

\begin{funding}
Subhra Sankar Dhar was supported in part by research grant CRG/2022/001489, Government of India. 
\end{funding}

\bibliographystyle{imsart-nameyear}      
\bibliography{main}  

\end{document}